%% file: main.tex
\documentclass{article}
\usepackage[T1]{fontenc}
\usepackage[english]{babel}
\usepackage[utf8]{inputenc}
\usepackage{csquotes}

\usepackage{fancyhdr}

\usepackage{amsthm}
\usepackage{amsmath}
\usepackage{amssymb}
\usepackage{mathrsfs}
\usepackage{stmaryrd}
\usepackage{verbatim}
\usepackage{tkz-tab}
\usepackage[all,cmtip]{xy}
\usepackage{mathtools}
\usepackage{stmaryrd}
\usepackage{mathdots}

\usepackage{tikz}
\usetikzlibrary{arrows.meta,positioning,decorations.pathreplacing,calc}
\usepackage{tikz-cd}

\usepackage[
backend=biber,
style=alphabetic,
sorting=anyt
]{biblatex}
\usepackage{graphicx}
\usepackage{esvect} 
\usepackage{hyperref}

\usepackage{tcolorbox}
\usepackage{xcolor}
\hypersetup{
    colorlinks = true,
    linkcolor=blue,
    citecolor=blue,
    urlcolor={blue!80!black}
}
\usepackage{setspace}
\usepackage{dsfont}
\usepackage{float}
\usepackage{pythonhighlight}

\newcommand{\hooklongrightarrow}{\lhook\joinrel\longrightarrow}

\everymath{\displaystyle}

\makeatletter
\@ifpackageloaded{stix}{%
}{%
  \DeclareFontEncoding{LS2}{}{\noaccents@}
  \DeclareFontSubstitution{LS2}{stix}{m}{n}
  \DeclareSymbolFont{stix@largesymbols}{LS2}{stixex}{m}{n}
  \SetSymbolFont{stix@largesymbols}{bold}{LS2}{stixex}{b}{n}
  \DeclareMathDelimiter{\lBrace}{\mathopen} {stix@largesymbols}{"E8}%
                                            {stix@largesymbols}{"0E}
  \DeclareMathDelimiter{\rBrace}{\mathclose}{stix@largesymbols}{"E9}%
                                            {stix@largesymbols}{"0F}
}
\makeatother

\makeatletter
\newcommand*{\boxwedge}{%
  \mathbin{%
    \mathpalette\@boxwedge{}%
  }%
}
\newcommand*{\@boxwedge}[2]{%
  \sbox0{$#1\boxplus\m@th$}%
  \dimen2=.5\dimexpr\wd0-\ht0-\dp0\relax % side bearing
  \dimen@=\dimexpr\ht0+\dp0\relax
  \def\lw{.06}% linw width as factor for height of \boxplus
  \kern\dimen2 % side bearing
  \tikz[
    line width=\lw\dimen@,
    line join=round,
    x=\dimen@,
    y=\dimen@,
  ]
  \draw
    (\lw/2,0) rectangle (1-\lw,1-\lw)
    (\lw,0) -- (.5,1-\lw-\lw/2) -- (1-\lw-\lw/2 ,0)
  ;%
  \kern\dimen2 % side bearing
}
\makeatletter

\usepackage[left=2cm,right=2cm,top = 2.1cm,bottom=2cm,headheight=28pt]{geometry}
\newtheorem{Lemma}{Lemma}[section]
\newtheorem{Proposition}[Lemma]{Proposition}
\newtheorem{Corollary}[Lemma]{Corollary}
\newtheorem{Theorem}[Lemma]{Theorem}
\newtheorem{Definition}[Lemma]{Definition}

\theoremstyle{definition}
\newtheorem{Remark}[Lemma]{Remark}

\theoremstyle{definition}
\newtheorem{Example}[Lemma]{Example}

\everymath{\displaystyle}

\usepackage{varwidth}
\usepackage[shortlabels]{enumitem}

\title{Equivariant Poisson $2$-Algebra Bundles over Configuration Spaces}
\author{Hai Ch\^au Nguy\^en \\ Universit\'e Claude Bernard Lyon 1, ICJ, UMR5208}
\date{\today}

\usepackage[size=tiny]{todonotes}%leo

\begin{document}

\maketitle

\begin{abstract}
We study equivariant vector bundles over configuration spaces with diagonals included, viewed as orbifold quotients $M^n/\mathfrak{S}_n$ by symmetric groups. Working in the equivalent language of equivariant vector bundles, we construct an induced-equivariance functor and prove its adjunction with restriction. We then define Hadamard and Cauchy tensor products and show that they form a symmetric $2$-monoidal structure. We construct the corresponding tensor and symmetric algebra bundles and prove that, for a local vector bundle $V \rightarrow M$, the bundle $\mathbf{S}^{\boxtimes}(\mathbf{S}^{\otimes}(V))$ is the free commutative $2$-algebra generated by $V$. Finally, we show that any skew-symmetric bundle map $k : V \boxtimes V \rightarrow \mathbf{I}_{\otimes}$ induces a compatible Poisson bracket on this $2$-algebra bundle.
\end{abstract}

\tableofcontents

\newpage

\setcounter{section}{-1}

\section{Introduction}

In this paper, we study an algebraic description of observables in relativistic classical field theory. Multilocal observables in classical covariant field theory are most naturally described by integral kernels depending on several spacetime points. The analytic difficulty is that these kernels are typically distributional and become singular precisely when points collide, i.e. along partial diagonals. These singularities lead to the problem of multiplication of distributions, which is controlled by wavefront-sets \cite{hormander}. They play a central role in causal and microlocal approaches to covariant field theory \cite{brunetti+fredenhagen}, and more generally in perturbative algebraic quantum field theory and in the locally covariant framework of \cite{brunetti+fredenhagen+verch} and \cite{rejzner}. 

A classical field on a spacetime $M$ is a section $\varphi \in \Gamma(E)$ of a vector bundle $E \rightarrow M$. An observable is a smooth functional on the space of fields $\Gamma(E)$. Such functionals admit descriptions as integral kernels: a local observable is one whose kernel is supported at a single spacetime point, for instance a section $f$ of the dual bundle $E^*$ induces a local observable ${\mathcal F}_f(\varphi) = \int_{M} \big \langle f(x), \varphi(x) \big \rangle \mathrm{d} \mu(x)$ that is linear in $\varphi$. More generally, polynomial local observables are described by kernels $f_n \in \Gamma \Big( \mathbf{S}^{\otimes n}(E^*) \Big)$ integrated against powers $\varphi(x)^{\otimes n}$ at a single point. Here $\mathbf{S}^{\otimes n}(E^*)$ denotes the $n$-th symmetric algebra (coinvariants of the $\mathfrak{S}_n$-action on $(E^*)^{\otimes n}$), which is naturally dual to the space $\mathbf{\Sigma}^{\otimes n}(E)$ of symmetric tensors (invariants in $E^{\otimes n}$). Multilocal observables arise naturally by multiplying local observables supported at different points: for instance
\[ \big( {\mathcal F}_f \cdot {\mathcal F}_g \big)(\varphi) = \int_{M \times M} \big \langle f(x), \varphi(x) \big \rangle \big \langle g(y), \varphi(y) \big \rangle \mathrm{d} \mu(x) \otimes \mathrm{d} \mu(y). \]
The integration measures $\mathrm{d} \mu$ appearing here implicitly require a choice of volume form or, more intrinsically, densities as sections of the density bundle. We do not address this issue in this paper.
In this way, the space of multilocal observables is naturally indexed by configurations of points in $M$, and carries two algebraic operations: the fibrewise (Hadamard) product, corresponding to local polynomial observables, and the Cauchy product, corresponding to combining observables supported on separate configurations. The singular behavior occurs precisely when the points in a configuration collide (for instance $y=x$) and controlling these singularities is one of the central problems of perturbative quantum field theory.

Recently, in \cite{fkr}, the authors introduced an algebraic formulation over purely geometric objects that encode multilocality by vector bundles over unordered configuration spaces. These carry two natural tensor products: a fiberwise Hadamard product, which describes polynomial local observables, and a Cauchy product encoding disjoint union, which describes multilocal observables, thus forming a $2$-monoidal category in the sense of \cite{aguiar+mahajan}. The presence of two compatible tensor products in the algebra of observables was first observed in \cite{borcherds} and made explicit as a $2$-monoidal structure in \cite{herscovich}, where it serves as the categorical setting for the renormalization picture and for the quantization of perturbative QFT via Laplace pairings developed in \cite{brouder+fauser+frabetti+oeckl}. The construction in \cite{fkr} realizes this structure geometrically over open configuration spaces. Using symmetric algebras for both products, the authors construct a Poisson $2$-algebra bundle whose distributional sections are designed to model multilocal classical observables. For a different but related local-to-global approach to organizing the observables of QFT, see the framework of factorization algebras of \cite{costello+gwilliam}.

While removing the diagonals allows one to remain in the well-known territory of smooth manifolds, it also excludes precisely the singular locus where important analytic phenomena occur. On the analytic side, the open configuration spaces lack the compactness properties enjoyed by $M^n$, which are needed for functional-analytic constructions on section spaces. On the physical side, renormalization by Epstein-Glaser consists in extending distributions from $M^n \setminus \Delta_n$ to $M^n$ \cite{epstein+glaser}, and the ambiguity of this extension encodes the renormalization issues. An algebraic framework that excludes the diagonals therefore cannot capture this ambiguity within its own structure.

The main idea of this paper is to restore the diagonals that are removed in the previous approach. The diagonals appear as singular points of the orbifolds $M^n/\mathfrak{S}_n$. While the generalization of manifolds to orbifolds seems natural - orbifold structures appear in various physical contexts, notably in string theory \cite{dixon+harvey+vafa+witten} - phenomena specific to orbifolds appear, such as the "rank drop" phenomenon for orbifold vector bundles. Indeed, over a singular point of the base orbifold, orbifold sections are constrained by the isotropy and only detect the invariant fibre directions. Concretely, over a point $(x,x) \in M^2$, equivariant sections of a bundle $V_2 \rightarrow M^2$ are constrained to the $\mathfrak{S}_2$-invariant part of the fibre $(V_2)_{(x,x)}$, so that the space of observable quantities over coinciding configurations is smaller than over the generic configuration $(x,y)$. Equivariant vector bundles over powers $M^n$ and constructions involving the quotients $M^n/\mathfrak{S}_n$ also appear independently in algebraic geometry \cite{tevelev+torres} and in conformal field theory on symmetric orbifolds \cite{lunin+mathur}. We postpone the analytic study of distributional sections and renormalization to later work. 

Since the quotient spaces considered here are global quotients by finite groups, we do not need the full general theory of orbifolds; we refer to \cite{adem+leida+ruan} and \cite{caramello} for general background. Instead, we use the equivalent language of equivariant vector bundles. More precisely, vector orbibundles over $M^n/\mathfrak{S}_n$ may be described as $\mathfrak{S}_n$-equivariant vector bundles over $M^n$. This equivariant viewpoint is the framework adopted throughout the paper. 

Our first result is an induced-equivariance construction for equivariant vector bundles, analogous to induced representations in finite group representation theory. This gives a functorial way to pass from $H$-equivariant bundles to $G$-equivariant bundles and yields a Frobenius-type adjunction with restriction. 

We then construct the equivariant analogue of the $2$-monoidal framework of \cite{fkr}. On the category $\mathcal{VB}_{\mathfrak{S}_{\bullet}}(M^{\bullet})$ of equivariant vector bundles over the family of spaces $M^{\bullet} = (M^n)_{n \in \mathbb{N}}$, we define two tensor products: the Hadamard tensor product, given degreewise, and the Cauchy tensor product, obtained from the external tensor product together with induced equivariance. We prove that these two monoidal structures are compatible and form a symmetric $2$-monoidal category. 

Next, we study the associated algebra objects. We construct tensor and symmetric algebra bundles for both monoidal products and show that, for a local vector bundle $V \rightarrow M$, the bundle $\mathbf{S}^{\boxtimes} \big( \mathbf{S}^{\otimes}(V) \big)$ carries a natural structure of commutative $2$-algebra. 

Finally, given a skew kernel $k: V \boxtimes V \rightarrow \mathbf{I}_{\otimes}$, we show that it induces a Poisson bracket on $\mathbf{S}^{\boxtimes} \big( \mathbf{S}^{\otimes}(V) \big)$ compatible with the two monoidal products. This provides the equivariant counterpart of the Poisson $2$-algebra structure previously obtained away from the diagonals. 

The Poisson $2$-algebra bundle $\mathbf{S}^{\boxtimes} \Big( \mathbf{S}^{\otimes}(V) \Big)$ constructed here is the classical algebraic structure underlying multilocal observables on multiconfiguration spaces. It provides the foundation on which a theory of distributional sections, and ultimately a geometric treatment of renormalization, can be developed.

The paper is organized as follows. In Section 1, we recall the basic theory of equivariant vector bundles and prove the induced-equivariance theorem. In Section 2, we introduce the Hadamard and Cauchy tensor products and establish the symmetric $2$-monoidal structure. In Section 3, we construct the corresponding algebra bundles, including the free commutative $2$-algebra generated by a local vector bundle. In Section 4, we define Poisson structures in this setting and prove the extension theorem for the bracket induced by a skew kernel, and illustrate the construction explicitly for trivial bundles over $M = \mathbb{R}$, recovering in particular the canonical Poisson bracket of classical mechanics.

\paragraph{Notations.} We conclude the introduction by fixing some notational conventions used throughout the paper. \begin{itemize}
\item[]If $V$ and $W$ are vector bundles over the same base manifold $X$, we distinguish between $\mathrm{Mor}(V,W)$, the vector space of vector bundle morphisms from $V$ to $W$ covering the identity, and $\mathrm{Hom}(V,W)$, the vector bundle over $X$ defined by $\mathrm{Hom}(V,W) = V^* \otimes W$. Although these objects are different in nature, they are related by the canonical identification $\Gamma \Big( \mathrm{Hom}(V,W) \Big) \cong \mathrm{Mor}(V,W)$. Because of this, the two notations may easily be confused, and in this paper, we will adhere strictly to this distinction. 
\item[] We denote the Hadamard tensor product by $\otimes$ and the Cauchy tensor product by $\boxtimes$; this should not be confused with the usual external tensor product, denoted by $\boxtimes^{\mathrm{ext}}$. 
\item[] If $M$ is a manifold, we write $M^0$ for a singleton whose unique element is denoted by $\emptyset$ and is called the vacuum state. 
\item[] Ordinary vector bundles $V \rightarrow M$, viewed in contrast to bundles $V_n \rightarrow M^n$, will be referred to as local vector bundles.
\item[] Given a permutation $\sigma \in \mathfrak{S}_n$, we use the abbreviated matrix notation $\sigma = \begin{bmatrix} \sigma(1) & \cdots & \sigma(n) \end{bmatrix}$.
\end{itemize}

\paragraph{Acknowledgments.} The author is deeply grateful to Alessandra Frabetti and Leonid Ryvkin for their guidance, support, and constant encouragement throughout this work. Special thanks are due to Anastasis Fotiadis and Sacha Amiel for many stimulating discussions and for constantly keeping the exchange alive. The author also thanks Sebastian Daza, Marco d'Agostino, Olga Kravchenko for helpful discussions and valuable comments.

For the purpose of Open Access, a CC-BY-NC-SA public copyright license has been applied by the author
to the present document and will be applied to all subsequent versions up to the Author Accepted Manuscript arising from this submission.

\input{2_Equivariant_VB}

\input{3_Two-Monoidal_Category}

\input{4_Algebras}

\input{5_Poisson}

\section{Conclusion}

The main outcome of this work is that the algebraic formalism developed previously away from the diagonals extends naturally to the setting where repeated points are allowed. More precisely, by working equivariantly on the spaces $M^n$, or equivalently on the quotient orbifolds $M^n/\mathfrak{S}_n$, we obtained a symmetric $2$-monoidal category of equivariant vector bundles carrying both the Hadamard and Cauchy tensor products. This gives a natural framework for defining tensor, symmetric, and Poisson algebra bundles on multiconfiguration spaces, that is, on unordered configuration spaces with diagonals retained by considering the space of orbits under the action of $\mathfrak{S}_n$. 

In this way, the passage from smooth configuration spaces to the singular quotient spaces $M^n/\mathfrak{S}_n$ does not destroy the algebraic structures underlying multilocal observables, but instead allow them to be formulated in a larger setting. The equivariant point of view is especially well suited to this extension, since it allows one to work on smooth manifolds while keeping track of the orbifold geometry of the quotient. A distinctive feature of this enlarged setting is the rank drop phenomenon: over configurations with coinciding points, the equivariance constraints reduce the effective dimension of the fibres, so that fewer independent sections - and hence fewer independent observables - can be detected. This phenomenon is a direct consequence of the orbifold structure and has no counterpart in the smooth setting of \cite{fkr}.

A natural next step is to study sections of the bundles constructed here. From the equivariant perspective, one is first led to consider the spaces of sections of the corresponding equivariant bundles over $M^n$; passing to the quotient should then amount to restricting to the $\mathfrak{S}_n$-invariant sections, which are precisely the orbifold sections over the multiconfiguration spaces. In this way, the algebraic structures obtained in the present article should induce corresponding algebraic structures at the level of sections, and ultimately at the level of distributional sections relevant for field-theoretic applications.

The analytic theory of sections and distributional sections on these multiconfiguration spaces remains to be developed. In particular, the intrinsic treatment of densities required for integration, the study of singularities along the diagonals and the corresponding renormalization procedures are left for future work. The present article provides the algebraic and categorical foundation for such an extension. 

\newpage

\printbibliography

\end{document}

%% file: 2_Equivariant_VB.tex
\section{Induced Equivariance for Equivariant Vector Bundles}

In this section, we recall the basic definitions and constructions concerning equivariant vector bundles. We then prove an induced-equivariance theorem, inspired by induction in representation theory, together with its adjunction with restriction.

\subsection{Equivariant Vector Bundles} \label{1.1}

We follow the standard definition going back to \cite{segal}.

\begin{Definition} 
Let $G$ be a finite group acting smoothly (on the left) on a smooth manifold $X$. A \textbf{$G$-equivariant vector bundle} over $X$ is a vector bundle $V \rightarrow X$ such that: \begin{enumerate} 
\item $G$ acts on the total space $V$,
\item The projection map $\pi: V \rightarrow X$ is $G$-equivariant,
\item Every group element $g$ induces a linear isomorphism fibrewise: $g: V_x \longrightarrow V_{g \cdot x}$.
\end{enumerate}
A morphism between $G$-equivariant vector bundles $V$ and $W$ is a $G$-equivariant bundle map $V \longrightarrow W$ covering the identity on $X$. The space of $G$-equivariant bundle maps between $V$ and $W$ will be denoted by $\mathrm{Mor}_{G}(V,W)$. We denote by $\mathcal{VB}_G(X)$ the category of $G$-equivariant vector bundles over $X$.
\end{Definition}

Notice that if $x \in X$ is fixed by the $G$-action, then the fibre $V_x$ carries a natural representation of the isotropy group $G_x = \left\{ g \in G, \; g \cdot x = x \right\}$.

An equivalent formulation of the equivariant structure, which is sometimes convenient, is the following.

\begin{Proposition} 
Let $G$ be a finite group, $X$ be a left $G$-space and $\pi: V \rightarrow X$ be a vector bundle. The following data are equivalent: \begin{enumerate}
    \item A $G$-equivariant vector bundle structure on $V$. 
    \item For each $g \in G$, a vector bundle isomorphism $\varphi_g: V \rightarrow g_X^* V$ over $X$, where $g_X : X \rightarrow X$ is the action of $g$ on the base $X$, such that $\varphi_{1} = \mathrm{Id}_V$ and 
    \[ \forall g,h \in G, \qquad \varphi_{gh} = h_X^*(\varphi_g) \circ \varphi_h. \]
\end{enumerate}
\end{Proposition}

\begin{proof} 
Assume first that $V \rightarrow X$ is a $G$-equivariant vector bundle. For each $g \in G$, denote respectively by 
\[ g_V: V \rightarrow V \qquad \text{and} \qquad g_X: X \rightarrow X \] 
the actions of $g$ on the total space $V$ and on the base $X$. Since the projection $\pi$ is $G$-equivariant, the following diagram commutes 
\[ \xymatrix{
V \ar[rr]^{g_V} \ar[d]_{\pi} && V \ar[d]^{\pi} \\
X \ar[rr]_{g_X} && X
} \]
Hence, we have a well-defined bundle map $\varphi_g: V \rightarrow g_X^* V$ defined by 
\[ \forall v \in V, \qquad \varphi_g(v) = \Big( \pi(v), g_V(v) \Big). \]
Fibrewise, this map is $(\varphi_g)_x: V_x \longrightarrow (g_X^* V)_x = V_{g_X(x)}$ and it is linear because the action of $g$ on the total space is fibrewise linear. Since $g_V$ is invertible, $\varphi_g$ is a vector bundle isomorphism. For the identity element $1 \in G$, the map $1_V$ is the identity on $V$, hence $\varphi_1 = \mathrm{Id}_V$. Now let $g,h \in G$. For $v \in V_x$, one has 
\[ \varphi_h(v) = \Big( x, h_V(v) \Big) \in (h_X^* V)_x, \] 
and then
\[ h_X^*(\varphi_g) \big( \varphi_h(v) \big) = \Big( x, g_V \big( h_V(v) \big) \Big) = \Big( x, (gh)_V(v) \Big) = \varphi_{gh}(v). \]
Thus
\[ \varphi_{gh} = h_X^*(\varphi_g) \circ \varphi_h. \]

Conversely, suppose that for each $g \in G$, we are given a vector bundle isomorphism $\varphi_g: V \rightarrow g_X^* V$ satisfying the two stated compatibility conditions. Let $\mathrm{pr}_2: g_X^* V \rightarrow V$ be the second projection, and define 
\[ g_V: \mathrm{pr}_2 \circ \varphi_g: V \longrightarrow V. \] 
Since $\varphi_g$ is a bundle map over $\mathrm{Id}_X$, and since $\mathrm{pr}_2$ is a bundle map covering $g_X$, it follows that $g_V$ is a bundle map covering $g_X$. In particular, one has
\[ \pi \circ g_V = g_X \circ \pi. \] 
Moreover, fibrewise $g_V$ is linear because $\varphi_g$ is fibrewise linear. We now check the group action axioms. Since $\varphi_1 = \mathrm{Id}_V$, the corresponding map $1_V$ is the identity on $V$. Next, using
\[ \varphi_{gh} = h_X^*(\varphi_g) \circ \varphi_h \] 
one sees that both $(gh)_V$ and $g_V \circ h_V$ are obtained by composing 
\[ V \overset{\varphi_h} \longrightarrow h_X^* V \overset{h_X^*(\varphi_g)} \longrightarrow h_X^* g_X^* V \cong (gh)_{X}^* V \] 
with the natural projection onto $V$. Hence
\[ (gh)_V = g_V \circ h_V. \] 
Therefore the maps $g_V$ define an action of $G$ on the total space $V$, and this action makes $\pi: V \rightarrow X$ into a $G$-equivariant vector bundle. The two constructions are inverse to one another, and therefore the two kinds of data are equivalent.
\end{proof}

It is sometimes useful to reinterpret equivariant vector bundles in the language of groupoids. A groupoid ${\mathcal G}$ is a small category in which all arrows are invertible. Recall that a $\mathbb{K}$-linear representation of a groupoid ${\mathcal G}$ is a functor 
\[ {\mathcal F}: {\mathcal G} \longrightarrow \mathbf{Vect}_{\mathbb{K}} \] 
where $\mathbf{Vect}_{\mathbb{K}}$ is the category of $\mathbb{K}$-vector spaces. That is, \begin{itemize} 
\item[$\bullet$] For each object $x$ of ${\mathcal G}$, there is a $\mathbb{K}$-vector space $V_x$,
\item[$\bullet$] Each morphism $g: x \rightarrow y$ in ${\mathcal G}$ induces a linear isomorphism $g: V_x \rightarrow V_y$.
\end{itemize}
Therefore, any $G$-equivariant vector bundle over $X$ can be seen as a linear representation of the action groupoid ${\mathcal G} = G \curvearrowright X$. To include the smooth structure, one can reformulate this as follows: a smooth representation of ${\mathcal G} \rightrightarrows X$ consists of a smooth vector bundle $V \rightarrow X$ together with a smooth Lie groupoid homomorphism 
\[ {\mathcal G} \longrightarrow \mathcal{GL}(V) \]
where $\mathcal{GL}(V)$ is the linear groupoid of $V$ (see \cite{mackenzie}, Section 1.7). 

\begin{Example}[Operations between equivariant vector bundles] Let $V$ and $W$ be $G$-equivariant vector bundles over the same base $X$. The following standard constructions remain equivariant. \begin{enumerate}[(a)]
\item \textbf{The direct sum} $V \oplus W$ is a $G$-equivariant vector bundle with the diagonal action 
\[ \forall g \in G, \quad \forall v \in V, \quad \forall w \in W, \qquad g \cdot (v, w) = (g \cdot v, g \cdot w). \]
The direct sum is the biproduct of the category $\mathcal{VB}_G(X)$.

\item \textbf{The (internal) tensor product} $V \otimes W$ is a $G$-equivariant vector bundle with the diagonal action 
\[ \forall g \in G, \quad \forall v \in V, \quad \forall w \in W, \qquad g \cdot (v \otimes w) = (g \cdot v) \otimes (g \cdot w). \]
The tensor product gives to $\mathcal{VB}_G(X)$ the structure of a symmetric monoidal category whose unit is the trivial line bundle $X \times \mathbb{K}$ with $G$-action given by
\[ \forall g \in G, \quad \forall (x,\lambda) \in X \times \mathbb{K}, \qquad g \cdot (x, \lambda) = (g \cdot x, \lambda). \]

\item \textbf{The dual} $V^*$ is a $G$-equivariant vector bundle with the dual action 
\[ \forall g \in G, \quad \forall \mu \in V^*, \quad \forall v \in V, \qquad (g \cdot \mu)(v) = \mu \big( g^{-1} \cdot v \big). \]

\item \textbf{The internal $\mathrm{Hom}$-bundle} $\mathrm{Hom}(V,W)$ is defined by $\mathrm{Hom}(V,W) = V^* \otimes W$ and inherits a natural $G$-equivariant structure. More precisely, fibrewise, for $x \in X$, $g \in G$ and $T_{g^{-1} \cdot x}: V_{g^{-1} \cdot x} \rightarrow W_{g^{-1} \cdot x}$, one sets 
\[ (g \cdot T)_x = g \circ T_{g^{-1} \cdot x} \circ g^{-1}: V_x \longrightarrow  W_x. \]
Note, however, that the space of sections of $\mathrm{Hom}(V,W)$ identifies with $\mathrm{Mor}(V,W)$, the space of all bundle morphisms from $V$ to $W$. In order to recover the space $\mathrm{Mor}_G(V,W)$ of $G$-equivariant bundle morphisms, one must instead consider the $G$-invariant sections: 
\[ \mathrm{Mor}_G(V,W) = \Gamma \big( \mathrm{Hom}(V,W) \big)^G. \]
In particular, although $\Gamma \big( \mathrm{Hom}(V,W) \big)$ does not directly recover $\mathrm{Mor}_G(V,W)$, the bundle $\mathrm{Hom}(V,W)$ is still the internal Hom object in the category $\mathcal{VB}_G(X)$.

\item Let $V \rightarrow X$ be a $G$-equivariant vector bundle and $W \rightarrow X$ a $G$-equivariant subbundle, i.e. the inclusion $W \hookrightarrow V$ is $G$-equivariant. \textbf{The quotient bundle} $V/W \rightarrow X$ inherits a natural $G$-equivariant structure given by 
\[ \forall g \in G, \quad \forall [v] \in V/W, \qquad g \cdot [v] = \big[ g \cdot v \big]. \]

\item Let $V \rightarrow X$ be a $G$-equivariant vector bundle and $W \rightarrow Y$ be a $H$-equivariant vector bundle. Recall that \textbf{the external tensor product} $V \boxtimes^{\mathrm{ext}} W$ is the vector bundle over $X \times Y$ defined by $\pi_1^* V \otimes \pi_2^* W$ where $\pi_1: X \times Y \rightarrow X$ (resp. $\pi_2: X \times Y \rightarrow Y$) is the projection onto the first factor (resp. onto the second factor). Fibrewise, one has 
\[ \big( V \boxtimes^{\mathrm{ext}} W \big)_{(x,y)} = V_x \otimes W_y. \]
The external tensor product $V \boxtimes^{\mathrm{ext}} W$ is naturally a $(G \times H)$-equivariant vector bundle with action given by 
\[ \forall (g,h) \in G \times H, \quad \forall v \in V, \quad \forall w \in W, \qquad (g,h) \cdot (v \otimes w) = (g \cdot v) \otimes (h \cdot w). \]

\item Let $V \rightarrow X$ be a $G$-equivariant vector bundle and $f: Y \rightarrow X$ be a $G$-equivariant map. \textbf{The pullback} $f^*(V)$ is $G$-equivariant with action 
\[ \forall g \in G, \quad \forall (y,v) \in f^*(V), \qquad g \cdot (y,v) = \big( g \cdot y, g \cdot v \big) \] 
which is well-defined on $f^*(V)$ because $f(g \cdot y) = g \cdot f(y) = g \cdot \pi(v) = \pi(g \cdot v)$.
\end{enumerate}
\end{Example}

\begin{Example}[Trivial bundles and twisted trivial bundles] Let us briefly discuss trivial and twisted trivial equivariant vector bundles. Consider a trivial vector bundle $X \times V$, where $V$ is a finite-dimensional vector space. Any $G$-equivariant structure on $X \times V$ has the form 
\[ g \cdot (x,v) = \big( g \cdot x, \rho_g(x)(v) \big) \] 
where, for each $g \in G$, the map $\rho_g: X \longrightarrow GL(V)$ is smooth and satisfies the cocycle relation 
\[ \rho_{gh}(x) = \rho_g(h \cdot x) \rho_h(x). \]
When the fibre action is independent of $x$, that is, when each $\rho_g: X \rightarrow GL(V)$ is a constant map, the family $(\rho_g)_{g \in G}$ defines a linear representation $\rho: G \rightarrow GL(V)$. In that case, we denote by $X \times_{\rho} V$ the corresponding equivariant bundle and call it a \textbf{twisted trivial bundle}. The \textbf{trivial equivariant bundle} corresponds to the case where $\rho$ is the trivial representation; it will simply be denoted by $X \times V$. For instance, the monoidal unit of $\big( \mathcal{VB}_G(X), \otimes \big)$ is the trivial equivariant line bundle $X \times \mathbb{K}$.
\end{Example}

\subsection{Induced Equivariance and Restricted Equivariance} \label{1.3}

Throughout this subsection, $G$ is a finite group and $H$ is a subgroup of $G$ (not necessarily normal in $G$). When $X = \left\{ * \right\}$, an $H$-equivariant vector bundle over $X$ is simply a linear representation of $H$. Then it is a standard construction of representation theory to produce a linear representation of $G$ called the induced representation $\mathrm{Ind}_H^G(V)$ (see \cite[Chapter 3]{fulton+harris}). We mimic this construction in the smooth case for equivariant vector bundles.

\begin{Definition}[Induced equivariance]
Let $V \rightarrow X$ be an $H$-equivariant vector bundle and suppose that $X$ is a $G$-space. Denote by ${\mathcal R}$ a set of representatives of the right cosets $H \setminus G$. The \textbf{induced equivariant vector bundle} is defined by 
\[ \mathrm{Ind}_H^G(V) = \bigoplus \limits_{\omega \in {\mathcal R}} \omega^* V. \]
\end{Definition}

\begin{Theorem}[Induced equivariance] \label{induced}
Using the previous notations, $\mathrm{Ind}_H^G(V)$ is a $G$-equivariant vector bundle over $X$. Moreover, this construction is independent of the choice of ${\mathcal R}$ up to canonical isomorphism, and therefore defines a functor
\[ \mathrm{Ind}_H^G: \mathcal{VB}_H(X) \longrightarrow \mathcal{VB}_G(X). \]
\end{Theorem}

\begin{proof} 
Since ${\mathcal R}$ is a set of representatives of right cosets $H \setminus G$, the map $\omega \in {\mathcal R} \longmapsto H \omega \in H \setminus G$ is a bijection. The group $G$ acts (on the left) on $H \setminus G$ by 
\[ \gamma \cdot (Hg) = H g \gamma^{-1}, \] 
and transporting this action through the above bijection yields an action $G \curvearrowright {\mathcal R}$. Equivalently, for any $\omega \in {\mathcal R}$ and $g \in G$, there exist unique elements $h \in H$ and $\omega' \in {\mathcal R}$ such that 
\[ \omega g^{-1} = h \omega'. \]

We now define a $G$-action on $\mathrm{Ind}_{H}^G(V)$, which we denote by $W$ for simplicity. An element of $W$ has the form $\big( x, (v_{\omega})_{\omega \in {\mathcal R}} \big)$, where $v_{\omega}$ is an element of $(\omega^* V)_x = V_{\omega(x)}$. For $g \in G$, we set
\[ g \cdot \big( x, (v_{\omega})_{\omega \in {\mathcal R}} \big) = \big( gx, \;  (v'_{\omega'})_{\omega' \in {\mathcal R}} \big), \qquad \text{where} \qquad v'_{\omega'} = h^{-1} \cdot v_{\omega}, \]
and where $h \in H$ and $\omega' \in {\mathcal R}$ are uniquely determined by the relation $\omega g^{-1} = h \omega'$. Note that this action is smooth because $G$ is finite (hence discrete) and it only involves linear transformations by sending an element from a summand to another one.

We now verify that this defines a $G$-equivariant vector bundle structure. Let $g_1, g_2 \in G$. Write  
\[ g_2 \cdot \big( x, (v_{\omega})_{\omega \in {\mathcal R}} \big) = \Big( g_2 \cdot x, (v_{\omega_2}^{(2)})_{\omega_2 \in {\mathcal R}} \Big), \qquad \text{where} \qquad v_{\omega_2}^{(2)} = h_2^{-1} \cdot v_{\omega} \quad \text{and} \quad \omega g_2^{-1} = h_2 \omega_2, \]
and then 
\[ g_1 \cdot \Big( g_2 \cdot x, (v_{\omega_2}^{(2)})_{\omega_2 \in {\mathcal R}} \Big) = \Big( g_1 \cdot (g_2 \cdot x), (v_{\omega_1}^{(1)})_{\omega_1 \in {\mathcal R}} \Big), \qquad \text{where} \qquad v_{\omega_1}^{(1)} = h_1^{-1} \cdot v_{\omega_2}^{(2)} \quad \text{and} \quad \omega_2 g_1^{-1} = h_1 \omega_1. \]
We have 
\[ v_{\omega_1}^{(1)} = h_1^{-1} \cdot v_{\omega_2}^{(2)} = h_1^{-1} \cdot h_2^{-1} \cdot v_{\omega} = (h_2 h_1)^{-1} \cdot v_{\omega} \]
and 
\[ \omega (g_1 g_2)^{-1} = \omega g_2^{-1} g_1^{-1} = h_2 \omega_2 g_1^{-1} = h_2 h_1 \omega_1. \]
Therefore, since the decomposition $G \longrightarrow H \times {\mathcal R}$ is unique, it follows that
\[ (g_1 g_2) \cdot \Big( x, (v_{\omega})_{\omega \in {\mathcal R}} \Big) = \Big( (g_1 g_2) \cdot x, (v'_{\omega'})_{\omega' \in {\mathcal R}} \Big)\qquad \text{where} \qquad v'_{\omega'} = (h_2 h_1)^{-1} \cdot v_{\omega} = v_{\omega_1}^{(1)}. \]
Thus the action is associative. The identity element acts trivially, since for every $\omega \in {\mathcal R}$ one has 
\[ \omega = \mathrm{Id} \cdot \omega \] 
which is already the required decomposition. Finally, the projection $W \rightarrow X$ is $G$-equivariant by construction, and the action is fibrewise linear because the $H$-action on each summand $\omega^* V$ is fibrewise linear. Hence $W$ is a $G$-equivariant vector bundle.

It remains to show that the construction is independent, up to canonical isomorphism, of the choice of representatives. Let ${\mathcal R}$ and ${\mathcal R}'$ be two systems of representatives of $H \setminus G$. For every $\omega' \in {\mathcal R}'$ there exist unique $h \in H$ and $\omega \in {\mathcal R}$ such that $\omega' = h \omega$. The $H$-equivariant structure of $V$ then induces canonical bundle isomorphisms
\[ (\omega')^{*} V \cong (h \omega)^* V \cong \omega^* V. \]
By assembling these isomorphisms over all representatives, one obtains a canonical isomorphism between the corresponding direct sums. Therefore the bundle $\mathrm{Ind}_H^G(V)$ is well defined up to canonical isomorphism, and the functor $\mathrm{Ind}_H^G: \mathcal{VB}_H(X) \rightarrow \mathcal{VB}_G(X)$ is well defined.
\end{proof}
 
In linear representation theory, the induced representation functor is left-adjoint to the restriction functor. This result, also known as Frobenius reciprocity, admits the following bundle-theoretic analogue.

\begin{Definition}[Restricted equivariance]
Let $W \rightarrow X$ be a $G$-equivariant vector bundle. The \textbf{restricted $H$-equivariant bundle} is the same underlying vector bundle $W \rightarrow X$, endowed with the action obtained by restricting the $G$-action to $H$. We denote it by $\mathrm{Res}_H^G(W)$. This defines a functor 
\[ \mathrm{Res}_H^G: \mathcal{VB}_G(X) \longrightarrow \mathcal{VB}_H(X). \]
\end{Definition}

The following statement means concretely that a $G$-equivariant map out of $\mathrm{Ind}_H^G(V)$ is equivalent to an $H$-equivariant map out of $V$.

\begin{Corollary}[Frobenius reciprocity]
The functor $\mathrm{Ind}_H^G$ is left adjoint to $\mathrm{Res}_H^G$. In other words, there is a bijection
\[ \mathrm{Mor}_G \Big( \mathrm{Ind}_H^G(V), W \Big) \cong \mathrm{Mor}_H \Big( V, \mathrm{Res}_H^G(W) \Big). \]
\end{Corollary}

\begin{proof} 
Define a map $\Phi: \mathrm{Mor}_G \Big( \mathrm{Ind}_H^G(V), W \Big) \longrightarrow \mathrm{Mor}_H \Big( V, \mathrm{Res}_H^G(W) \Big)$ by
\[ \forall f \in \mathrm{Mor}_G \Big( \mathrm{Ind}_H^G(V), W \Big), \qquad \Phi(f) = f \circ j \in \mathrm{Mor}_H \Big( V, \mathrm{Res}_H^G(W) \Big) \]
where $j: V \hooklongrightarrow \mathrm{Ind}_H^G(V)$ is the inclusion of the summand corresponding to the identity representative. Since $j$ is $H$-equivariant and every $G$-equivariant map is in particular $H$-equivariant, the composite $\Phi(f) = f \circ j$ is $H$-equivariant. 

Conversely, let $\psi: V \longrightarrow W$ be an $H$-equivariant bundle map. For each $\omega \in {\mathcal R}$, define a bundle map $\widetilde{\psi}_{\omega}: \omega^* V \longrightarrow W$ fibrewise by
\[ \big( \widetilde{\psi}_{\omega} \big)_x: v \in (\omega^* V)_x = V_{\omega \cdot x} \longmapsto \omega^{-1} \cdot \psi_{\omega \cdot x}(v) \in W_x. \]
These maps assemble into a bundle map $\widetilde{\psi} = \bigoplus \limits_{\omega \in {\mathcal R}} \widetilde{\psi}_{\omega}: \mathrm{Ind}_H^G(V) \rightarrow W$. Since ${\mathcal R}$ is finite, this direct sum is finite, so $\widetilde{\psi}$ is well defined, smooth, and fibrewise linear. It remains to check that $\widetilde{\psi}$ is $G$-equivariant. Recall that 
\[ g \cdot \Big( x, (v_{\omega})_{\omega \in {\mathcal R}} \Big) = \Big( g \cdot x, (v'_{\omega'})_{\omega' \in {\mathcal R}} \Big) \qquad \text{where} \qquad v'_{\omega'} = h^{-1} \cdot v_{\omega} \quad \text{and} \quad \omega g^{-1} = h \omega'. \]
We have 
\[ \begin{array}{lll} \widetilde{\psi}_{gx} \Big( g \cdot \big( x, (v_{\omega})_{\omega \in {\mathcal R}} \big) \Big) &=& \sum \limits_{\omega' \in {\mathcal R}} (\omega')^{-1} \cdot \psi_{\omega' g x}(v'_{\omega'}) \\
&=& \sum \limits_{\omega' \in {\mathcal R}} (\omega')^{-1} \cdot \psi_{\omega' g x}(h^{-1} \cdot v_{\omega}) \\
&=& \sum \limits_{\omega' \in {\mathcal R}} (\omega')^{-1} \cdot h^{-1} \cdot \psi_{\omega g x}(v_{\omega}) \\
&=& \sum \limits_{\omega' \in {\mathcal R}} (h \omega')^{-1} \cdot \psi_{\omega g x}(v_{\omega}) \\
&=& \sum \limits_{\omega =  h \omega' g} (g \omega^{-1}) \cdot \psi_{\omega x}(v_{\omega}) \\ 
&=& g \cdot \widetilde{\psi}_x \big( (v_{\omega})_{\omega \in {\mathcal R}} \big) \end{array} \]
where we have used the $H$-equivariance of $\psi$. Hence, $\widetilde{\psi}$ is $G$-equivariant. This defines a map 
\[ \Psi: \psi \in \mathrm{Mor}_H \Big( V, \mathrm{Res}_H^G(W) \Big) \longmapsto \widetilde{\psi} \in \mathrm{Mor}_G \Big( \mathrm{Ind}_H^G(V), W \Big). \]
It is immediate from the constructions that the maps $\Phi$ and $\Psi$ are inverses of each other.
\end{proof}

The previous construction admits a geometric description analogous to the classical associated-bundle model for induced representations in linear representation theory.

\begin{Proposition} 
There is a canonical isomorphism of $G$-equivariant vector bundles 
\[ \mathrm{Ind}_H^G(V) \cong (G \times V)/H \]
where the right action of $H$ on $G \times V$ is given by 
\[ (g,v) \cdot h = \big( gh, h^{-1} \cdot v \big). \]
\end{Proposition}

\begin{proof} 
The group $G$ acts on $(G \times V)/H$ by 
\[ \gamma \cdot \big[ (g,v) \big] = \big[ (\gamma g, v) \big]. \] 
This is well defined, since 
\[ \gamma \cdot \big( gh, h^{-1} \cdot v \big) = \big( \gamma g h, h^{-1} \cdot v \big) = (\gamma g, v) \cdot h. \]
The projection 
\[ \widetilde{\pi}: (G \times V)/H \longrightarrow X, \qquad \widetilde{\pi} \big( [g,v] \big) = g \cdot \pi(v), \]
where $\pi: V \rightarrow X$, is $G$-equivariant. Thus $E := (G \times V)/H \rightarrow X$ is a $G$-equivariant vector bundle. For a given $x \in X$, the fibre over $x$ is 
\[ E_x = \left\{ \big[(g,v) \big], \; g \cdot \pi(v) = x \right\} = \left\{ \big[ (g,v) \big], \; v \in V_{g^{-1} \cdot x} \right\}. \]
Now let $\big[ (g,v) \big] \in E_x$. Let ${\mathcal R}'$ be a set of representatives of the left cosets $G/H$; note that inversion induces a bijection ${\mathcal R}' \rightarrow {\mathcal R}$. We may write 
\[ g = \omega' h \qquad \text{with} \qquad \omega' \in {\mathcal R}' \quad \text{and} \quad h \in H. \]
Since $V$ is $H$-equivariant, the map 
\[ h: V_{h^{-1} \cdot (\omega')^{-1} \cdot x} \longrightarrow V_{(\omega')^{-1} \cdot x} \]
is an isomorphism. Define $\widetilde{v} = h \cdot v \in V_{(\omega')^{-1} \cdot x}$. Then $v = h^{-1} \cdot \widetilde{v}$, and therefore 
\[ \big[ (g,v) \big] = \big[ (\omega' h, h^{-1} \cdot \widetilde{v}) \big] = \big[ (\omega', \widetilde{v}) \big]. \]
Hence every class in $E_x$ admits a representative of the form $\big[ (\omega', \widetilde{v}) \big]$ with $\omega' \in {\mathcal R}'$ and $\widetilde{v} \in V_{(\omega')^{-1} \cdot x}$. This representative is unique. Indeed, suppose $\big[ (\omega', \widetilde{v}) \big] = \big[ (\omega'', \widetilde{v}') \big]$. Then there exists $h \in H$ such that 
\[ (\omega'', \widetilde{v}') = (\omega', \widetilde{v}) \cdot h. \] 
Thus $\omega'' = \omega' h$. Since $\omega'$ and $\omega''$ are both representatives of $G/H$, this forces $h = \mathrm{Id}$, and hence 
\[ (\omega'', \widetilde{v}') = (\omega', \widetilde{v}). \]
It follows that 
\[ E_x \cong \bigoplus \limits_{\omega' \in {\mathcal R}'} V_{(\omega')^{-1} \cdot x} \cong \bigoplus \limits_{\omega \in {\mathcal R}} V_{\omega \cdot x} = \bigoplus \limits_{\omega \in {\mathcal R}} (\omega^* V)_x. \] 
Since these identifications are natural in $x$, they assemble into an isomorphism of $G$-equivariant vector bundles 
\[ (G \times V)/H \cong \mathrm{Ind}_H^G(V). \]
\end{proof}

The induction functor is also compatible with the tensor product in an oplax sense.

\begin{Proposition} 
The functor $\mathrm{Ind}_H^G$ is an oplax monoidal functor between $\big( \mathcal{VB}_H(X), \otimes \big)$ and $\big( \mathcal{VB}_G(X), \otimes \big)$. More precisely, for all $A,B \in \mathcal{VB}_H(X)$, there is a canonical map 
\[ \mathrm{Ind}_H^G(A \otimes B) \longrightarrow \mathrm{Ind}_H^G(A) \otimes \mathrm{Ind}_H^G(B). \]
Moreover, this map is injective.
\end{Proposition}

\begin{proof} 
The restriction functor $\mathrm{Res}_H^G$ is strongly monoidal, that is $\mathrm{Res}_H^G(A \otimes B) \cong \mathrm{Res}_H^G(A) \otimes \mathrm{Res}_H^G(B)$. Therefore, there is an $H$-equivariant map 
\[ A \otimes B \longrightarrow \mathrm{Res}_H^G \Big( \mathrm{Ind}_H^G(A) \otimes \mathrm{Ind}_H^G(B) \Big). \]
By adjunction, this corresponds to a unique $G$-equivariant map 
\[ \mathrm{Ind}_H^G(A \otimes B) \longrightarrow \mathrm{Ind}_H^G(A) \otimes \mathrm{Ind}_H^G(B). \]
To see that this map is injective, note that it is given by the diagonal inclusion: 
\[ \mathrm{Ind}_H^G(A \otimes B) = \bigoplus \limits_{\omega \in {\mathcal R}} \omega^* A \otimes \omega^* B \hooklongrightarrow \mathrm{Ind}_H^G(A) \otimes \mathrm{Ind}_H^G(B) = \bigoplus \limits_{\omega_1, \omega_2 \in {\mathcal R}} \omega_1^* A \otimes \omega_2^* B. \]
Hence it is injective.
\end{proof}

\begin{Proposition} \label{prop6}
Let $K \subset H \subset G$ be subgroups. Then there is a canonical natural isomorphism of functors 
\[ \mathrm{Ind}_K^G \cong \mathrm{Ind}_H^G \circ \mathrm{Ind}_K^H: \mathcal{VB}_K(X) \longrightarrow \mathcal{VB}_H(X) \longrightarrow \mathcal{VB}_G(X). \]
\end{Proposition}

\begin{proof} 
For every $V \in \mathcal{VB}_K(X)$ and every $W \in \mathcal{VB}_G(X)$, one has 
\[ \mathrm{Mor}_G \Big( \mathrm{Ind}_H^G \circ \mathrm{Ind}_K^H(V), W \Big) \cong \mathrm{Mor}_{H} \Big( \mathrm{Ind}_K^H(V), \mathrm{Res}_{H}^G(W) \Big) \]
and 
\[ \mathrm{Mor}_H \Big( \mathrm{Ind}_K^H(V), \mathrm{Res}_H^G(W) \Big) \cong \mathrm{Mor}_K \Big( V, \mathrm{Res}_{K}^{H} \circ \mathrm{Res}_H^G(W) \Big) = \mathrm{Mor}_K \Big( V, \mathrm{Res}_K^G(W) \Big). \]
Therefore, the functor $\mathrm{Ind}_H^G \circ \mathrm{Ind}_K^H$ is left-adjoint to $\mathrm{Res}_K^G$. Since $\mathrm{Ind}_K^G$ is also left adjoint to $\mathrm{Res}_K^G$, the uniqueness of the left adjoints up to canonical natural isomorphism yields 
\[ \mathrm{Ind}_H^G \circ \mathrm{Ind}_K^H \cong \mathrm{Ind}_K^G. \]
\end{proof}

%% file: 3_Two-Monoidal_Category.tex
\section{The $2$-Monoidal Category of Equivariant Vector Bundles over Configuration Spaces}

In this section, we introduce the category of $\mathfrak{S}_{\bullet}$-equivariant vector bundles over the family of spaces $M^{\bullet} = (M^n)_{n \in \mathbb{N}}$, together with the two tensor products that will be used throughout the paper. We first define the ambient category $\mathcal{VB}_{\mathfrak{S}_{\bullet}}(M^{\bullet})$ and the Hadamard tensor product, then construct the Cauchy tensor product by induced equivariance, and finally prove that these two monoidal structures form a symmetric $2$-monoidal category.

\subsection{Equivariant Vector Bundles over Configuration Spaces} \label{2.1}

Let $M$ be a manifold. For each $n \in \mathbb{N}$, the permutation group $\mathfrak{S}_n$ acts smoothly on $M^n$ on the left by
\[ \forall \sigma \in \mathfrak{S}_n, \quad \forall (x_1, \ldots, x_n) \in M^n, \qquad \sigma \cdot (x_1, \ldots, x_n) = \big( x_{\sigma^{-1}(1)}, \ldots, x_{\sigma^{-1}(n)} \big). \]
As a convention, $\mathfrak{S}_0$ denotes the trivial group and $M^0 = \left\{ \emptyset \right\}$ is the one-point manifold whose single element is called the vacuum state. We denote $M^{\bullet}$ the family of spaces $(M^n)_{n \in \mathbb{N}}$.

We now introduce a geometric analogue of species, in the sense of Joyal \cite{joyal}, or equivalently of $\mathbb{S}$-modules appearing in the operadic litterature \cite{loday+vallette}. Recall that an $\mathbb{S}$-module is a collection $(V_n)_{n \in \mathbb{N}}$ in which each $V_n$ is a representation of the symmetric group $\mathfrak{S}_n$. In the geometric setting, we replace each representation $V_n$ by an $\mathfrak{S}_n$-equivariant vector bundle over $M^n$. We do not pursue this analogy further in the present paper.

\begin{Definition} 
An \textbf{$\mathfrak{S}_{\bullet}$-equivariant vector bundle} on $M^{\bullet}$ is a collection $\mathbf{V} = \big( V_n \rightarrow M^n \big)_{n \in \mathbb{N}}$ such that for each $n \in \mathbb{N}$, $V_n \rightarrow M^n$ is an $\mathfrak{S}_n$-equivariant vector bundle over $M^n$. 
We denote by $\mathcal{VB}_{\mathfrak{S}_{\bullet}}(M^{\bullet})$ the category whose objects are $\mathfrak{S}_{\bullet}$-equivariant vector bundles on $M$, and whose morphisms $f: \mathbf{V} \rightarrow \mathbf{W}$ are collections of $\mathfrak{S}_n$-equivariant bundle maps $f_n: V_n \rightarrow W_n$ covering the identity on $M^n$, for every $n \in \mathbb{N}$.
\end{Definition}

The inclusion of $M$ as the $1$-component of $M^{\bullet}$ gives rise to two natural functors.
\[ \iota^*: \mathbf{V} \in \mathcal{VB}_{\mathfrak{S}_{\bullet}}(M^{\bullet}) \longmapsto V_1 \in \mathcal{VB}(M) \]
called the \textbf{pullback} functor and 
\[ \iota_*: \mathcal{VB}(M) \longrightarrow \mathcal{VB}_{\mathfrak{S}_{\bullet}}(M^{\bullet}) \]
called the \textbf{pushforward} functor which sends a vector bundle $V \rightarrow M$ to the $\mathfrak{S}$-equivariant vector bundle concentrated in degree $1$, equal to $V$ in degree $1$ and to the zero vector bundle over $M^n$ in every degree $\neq 1$. Note that the pullback is a left inverse of the pushforward, i.e. $\iota^* \iota_* V = V$ for any $V \in \mathcal{VB}(M)$. However, $\iota_* \iota^* \neq \mathrm{Id}$. \newline
Thus $\mathcal{VB}(M)$ identifies with a full subcategory of $\mathcal{VB}_{\mathfrak{S}_{\bullet}}(M^{\bullet})$. As opposed to general $\mathfrak{S}_{\bullet}$-equivariant vector bundles, the ones coming from $\mathcal{VB}(M)$ will be called \textbf{local vector bundles}.

\begin{Definition}
For $\mathbf{V}, \mathbf{W} \in \mathcal{VB}_{\mathfrak{S}_{\bullet}}(M^{\bullet})$, their \textbf{Hadamard tensor product} $\mathbf{V} \otimes \mathbf{W}$ is defined degreewise by 
\[ \big( \mathbf{V} \otimes \mathbf{W} \big)_{n \in \mathbb{N}} = V_n \otimes W_n. \]
Its unit is given by the trivial equivariant line bundle $\mathbf{I}_{\otimes} = \big( M^n \times \mathbb{K} \rightarrow M^n \big)_{n \in \mathbb{N}}$ with trivial $\mathfrak{S}_n$-action on the fibre $\mathbb{K}$. For each $n \in \mathbb{N}$, we denote by $\mathds{1}_n^{\otimes}$ the constant section equal to $1$ of the trivial line bundle $(\mathbf{I}_{\otimes})_n$.
\end{Definition} 
The usual tensor product of vector bundles immediately yields the following.
\begin{Proposition} 
The Hadamard tensor product $\otimes$ gives $\mathcal{VB}_{\mathfrak{S}_{\bullet}}(M^{\bullet})$ a structure of a symmetric monoidal category, with braiding $\beta_{\mathbf{V}, \mathbf{W}}^{\otimes}: \mathbf{V} \otimes \mathbf{W} \rightarrow \mathbf{W} \otimes \mathbf{V}$ given degreewise by the canonical symmetry $V_n \otimes W_n \cong W_n \otimes V_n$.
\end{Proposition}

\subsection{Cauchy Tensor Product of Equivariant Vector Bundles over Configuration Spaces} \label{2.2}

A natural first attempt to combine vector bundles living over different powers $M^p$ and $M^q$ is to use the external tensor product $\boxtimes^{\mathrm{ext}}$. In the category of vector bundles over varying bases, this construction is symmetric monoidal: its braiding is induced by the block-swap diffeomorphism of the base  
\[ \tau_{p,q}: (x_1, \ldots, x_p, y_1, \ldots, y_q) \in M^p \times M^q \longmapsto (y_1, \ldots, y_q, x_1, \ldots, x_p) \in M^{q} \times M^p, \]
together with the canonical symmetry of tensor products. In the present setting, however, this procedure does not define a braiding in $\mathcal{VB}_{\mathfrak{S}_{\bullet}}(M^{\bullet})$, because morphisms in $\mathcal{VB}_{\mathfrak{S}_{\bullet}}(M^{\bullet})$ are required to cover the identity on the base $M^{p+q}$, whereas the map induced by $\tau_{p,q}$ covers the block permutation of $M^{p+q}$ rather than the identity. For this reason, we introduce the Cauchy tensor product $\boxtimes$ as a symmetrized version of the external tensor product, yielding a symmetric monoidal structure on $\mathcal{VB}_{\mathfrak{S}_{\bullet}}(M^{\bullet})$. 

\begin{Definition}
For $\mathbf{V}, \mathbf{W} \in \mathcal{VB}_{\mathfrak{S}_{\bullet}}(M^{\bullet})$, their \textbf{Cauchy tensor product} $\mathbf{V} \boxtimes \mathbf{W}$ is defined degreewise by 
\[ \forall n \in \mathbb{N}, \qquad (\mathbf{V} \boxtimes \mathbf{W})_n = \bigoplus \limits_{p+q=n} \mathrm{Ind}_{\mathfrak{S}_p \times \mathfrak{S}_q}^{\mathfrak{S}_n} \big( V_p \boxtimes^{\mathrm{ext}} W_q \big). \]
Its unit object, denoted by $\mathbf{I}_{\boxtimes}$, is the trivial line bundle concentrated in degree $0$, that is, 
\[ (\mathbf{I}_{\boxtimes})_0 = \left\{ \emptyset \right\} \times \mathbb{K} \qquad \text{and} \qquad (\mathbf{I}_{\boxtimes})_n = M^n \times \left\{0 \right\} \; \text{for all} \; n \geqslant 1. \]
We denote by $\mathds{1}_0^{\boxtimes}$ the canonical generator of the one-dimensional vector space $(\mathbf{I}_{\boxtimes})_0$.
\end{Definition}

\begin{Theorem} 
The Cauchy tensor product $\boxtimes$ gives $\mathcal{VB}_{\mathfrak{S}_{\bullet}}(M^{\bullet})$ the structure of a symmetric monoidal category.
\end{Theorem}

\begin{proof} 
The definition of $\boxtimes$ is functorial, since it is obtained degreewise from the external tensor product, direct sums, and the induction functors. 

We first identify the unit object. By definition, $(\mathbf{I}_{\boxtimes})_p = 0$ for every $p \geqslant 1$, while $(\mathbf{I}_{\boxtimes})_0$ is the trivial line bundle. Hence, for any $\mathbf{V} \in \mathcal{VB}_{\mathfrak{S}_{\bullet}}(M^{\bullet})$ and any $n \in \mathbb{N}$, only the summand with $p=0$ and $q=n$ contributes to $\big( \mathbf{I}_{\boxtimes} \boxtimes \mathbf{V} \big)_n$, and therefore 
\[ \big( \mathbf{I}_{\boxtimes} \boxtimes \mathbf{V} \big)_n = \bigoplus \limits_{p+q=n} \mathrm{Ind}_{\mathfrak{S}_p \times \mathfrak{S}_q}^{\mathfrak{S}_n} \Big( (\mathbf{I}_{\boxtimes})_p \boxtimes^{\mathrm{ext}} V_q \Big) \cong (\mathbf{I}_{\boxtimes})_0 \boxtimes^{\mathrm{ext}} V_n \cong V_n. \]
These isomorphisms are natural in $\mathbf{V}$ and define a left unit isomorphism $\mathbf{I}_{\boxtimes} \boxtimes \mathbf{V} \overset{\cong} \longrightarrow \mathbf{V}$. The same argument applied to $\mathbf{V} \boxtimes \mathbf{I}_{\boxtimes}$ gives a right unit isomorphism $\mathbf{V} \boxtimes \mathbf{I}_{\boxtimes} \overset{\cong} \longrightarrow \mathbf{V}$. Since both are induced degreewise from the ordinary unit isomorphisms for the external tensor product, their compatibility is inherited from the corresponding compatibility for $\boxtimes^{\mathrm{ext}}$. 

We next construct the associator. Let $\mathbf{U}, \mathbf{V}, \mathbf{W} \in \mathcal{VB}_{\mathfrak{S}_{\bullet}}(M^{\bullet})$. Expanding the definition twice, one finds 
\[ \Big( \big( \mathbf{U} \boxtimes \mathbf{V} \big) \boxtimes \mathbf{W} \Big)_n = \bigoplus \limits_{r+s=n} \mathrm{Ind}_{\mathfrak{S}_r \times \mathfrak{S}_s}^{\mathfrak{S}_n} \Big( \big( \mathbf{U} \boxtimes \mathbf{V} \big)_r \boxtimes^{\mathrm{ext}} W_s \Big) \cong \bigoplus \limits_{p+q+s=n} \mathrm{Ind}_{\mathfrak{S}_p \times \mathfrak{S}_q \times \mathfrak{S}_s}^{\mathfrak{S}_n} \Big( \big( U_p \boxtimes^{\mathrm{ext}} V_q \big) \boxtimes^{\mathrm{ext}} W_s \Big). \]
Similarly, 
\[ \Big( \mathbf{U} \boxtimes \big( \mathbf{V} \boxtimes \mathbf{W} \big) \Big)_n = \bigoplus \limits_{p+q+s=n} \mathrm{Ind}_{\mathfrak{S}_p \times \mathfrak{S}_q \times \mathfrak{S}_s}^{\mathfrak{S}_n} \Big( U_p \boxtimes^{\mathrm{ext}} \big( V_q \boxtimes^{\mathrm{ext}} W_s \big) \Big). \]
The associator of the external tensor product therefore induces, for each triple $(p,q,s)$, a natural isomorphism between the corresponding summands, and these assemble into a natural isomorphism 
\[ \alpha_{\mathbf{U}, \mathbf{V}, \mathbf{W}}^{\boxtimes}: \big( \mathbf{U} \boxtimes \mathbf{V} \big) \boxtimes \mathbf{W} \overset{\cong} \longrightarrow \mathbf{U} \boxtimes \big( \mathbf{V} \boxtimes \mathbf{W} \big). \]
The only point to check is that the repeated induction functors occurring in these identifications are compatible. This follows from Proposition \ref{prop6}, which identifies iterated induction with induction from the product subgroup $\mathfrak{S}_p \times \mathfrak{S}_q \times \mathfrak{S}_s$ to $\mathfrak{S}_n$. Once this is done, the compatibility of the associators is reduced componentwise to the usual associativity of the external tensor product. 

Finally, we construct the braiding. For $p,q \in \mathbb{N}$, let $\tau_{p,q}$ denote the permutation exchanging the first block of length $p$ with the second block of length $q$ 
\[ \tau_{p,q}: (x_1, \ldots, x_p, y_1, \ldots, y_q) \in M^p \times M^q \longmapsto (y_1, \ldots, y_q, x_1, \ldots, x_p) \in M^{q} \times M^p. \]
The corresponding block permutation of the base induces an isomorphism 
\[ V_p \boxtimes^{\mathrm{ext}} W_q \cong \tau_{p,q}^* \big( W_q \boxtimes^{\mathrm{ext}} V_p \big) \]
covering the identity of $M^{p+q}$. Moreover, conjugation by $\tau_{p,q}$ identifies the subgroup $\mathfrak{S}_p \times \mathfrak{S}_q$ with $\mathfrak{S}_q \times \mathfrak{S}_p$ inside $\mathfrak{S}_{p+q}$. By functoriality of induction, we obtain a natural isomorphism 
\[ \mathrm{Ind}_{\mathfrak{S}_p \times \mathfrak{S}_q}^{\mathfrak{S}_{p+q}}(V_p \boxtimes^{\mathrm{ext}} W_q) \cong \mathrm{Ind}_{\mathfrak{S}_q \times \mathfrak{S}_p}^{\mathfrak{S}_{p+q}}(W_q \boxtimes^{\mathrm{ext}} V_p). \]
Summing over all decompositions $n=p+q$ gives a natural isomorphism 
\[ \beta_{\mathbf{V}, \mathbf{W}}^{\boxtimes}: \mathbf{V} \boxtimes \mathbf{W} \overset{\cong} \longrightarrow \mathbf{W} \boxtimes \mathbf{V} \]
Its compatibility with the associator is again reduced degreewise to the corresponding compatibility for the external tensor product, and the relation 
\[ \beta_{\mathbf{W}, \mathbf{V}}^{\boxtimes} \circ \beta_{\mathbf{V}, \mathbf{W}}^{\boxtimes} = \mathrm{Id} \]
follows from the equality $\tau_{p,q} \circ \tau_{q,p} = \mathrm{Id}_{M^{p+q}}$.

We have therefore constructed a unit object, unit isomorphisms, an associator, and a symmetry satisfying the required compatibility conditions. Hence $\big( \mathcal{VB}_{\mathfrak{S}_{\bullet}}(M^{\bullet}), \boxtimes, \mathbf{I}_{\boxtimes} \big)$ is a symmetric monoidal category.
\end{proof}

\begin{Remark} 
The braiding described above corresponds to the case where the fibres are not graded. If they are, then the same construction must be modified by the usual Koszul sign rule: exchanging two homogeneous elements $a$ and $b$ contributes a factor $(-1)^{|a| |b|}$ in addition to the block permutation of the base.
\end{Remark}

\begin{Remark}[Explicit formula for the Cauchy tensor product] A convenient choice of representatives for the quotient $(\mathfrak{S}_p \times \mathfrak{S}_q) \setminus \mathfrak{S}_{p+q}$ is given by unshuffle permutations: 
\[ \mathrm{Ush}(p,q) = \left\{ \sigma \in \mathfrak{S}_{p+q}, \; \sigma^{-1}(1) < \ldots < \sigma^{-1}(p) \; \text{and} \; \sigma^{-1}(p+1) < \ldots < \sigma^{-1}(p+q) \right\}. \]
Hence,
\[ (\mathbf{V} \boxtimes \mathbf{W})_n = \bigoplus \limits_{\substack{p+q=n \\ \omega \in \mathrm{Ush}(p,q)}} \omega^* \big( V_p \boxtimes^{\mathrm{ext}} W_q \big). \]
For fixed $p$ and $q$, the number of such summands is $\mathrm{Card} \big( \mathrm{Ush}(p,q) \big) = \binom{p+q}{p}$, and summing over all decompositions $p+q=n$ yields a total of $2^n$ summands.

For small values of $n$, this gives the following fibrewise decompositions: 
\[ (\mathbf{V} \boxtimes \mathbf{W})_{(x_1, x_2)} = \big( V_{\emptyset} \otimes W_{(x_1, x_2)} \big) \oplus \big( V_{(x_1,x_2)} \otimes W_{\emptyset} \big) \oplus \big( V_{x_1} \otimes W_{x_2} \big) \oplus \big( V_{x_2} \otimes W_{x_1} \big), \] 
\[ \begin{array}{lll} (\mathbf{V} \boxtimes \mathbf{W})_{(x_1,x_2,x_3)} &=& \big( V_{\emptyset} \otimes W_{(x_1, x_2, x_3)} \big) \oplus \big( V_{(x_1, x_2, x_3)} \otimes W_{\emptyset} \big) \\ 
&& \oplus \big( V_{x_1} \otimes W_{(x_2,x_3)} \big) \oplus \big( V_{x_2} \otimes W_{(x_1, x_3)} \big) \oplus \big( V_{x_3} \otimes W_{(x_1, x_2)} \big) \\ 
&& \oplus \big( V_{(x_1, x_2)} \otimes W_{x_3} \big) \oplus \big( V_{(x_1, x_3)} \otimes W_{x_2} \big) \oplus \big( V_{(x_2, x_3)} \otimes W_{x_1} \big). \end{array} \]
The $\mathfrak{S}_n$-action may be described explicitly on this unshuffle decomposition as follows. Fix a summand indexed by $\omega \in \mathrm{Ush}(p,q)$ and let $\sigma \in \mathfrak{S}_n$. Write 
\[ \omega \sigma^{-1} = (\alpha \sqcup \beta) \omega' \] 
where $\alpha \sqcup \beta \in \mathfrak{S}_p \times \mathfrak{S}_q$ and $\omega' \in \mathrm{Ush}(p,q)$. Then $\sigma$ sends the $\omega$-summand to the $\omega'$-summand, and the induced map is obtained by applying $(\alpha \cup \beta)^{-1}$ through the $(\mathfrak{S}_p \times \mathfrak{S}_q)$-equivariance of $V_p \boxtimes^{\mathrm{ext}} W_q$. Thus the $\mathfrak{S}_n$-action permutes the summands indexed by unshuffles, together with the corresponding equivariance correction inside each summand. 

For example, in degree $3$ and for the decomposition $1+2$, this mechanism can be visualized schematically as follows (where we omit the $x_i$'s).
\begin{center}
\centering
\begin{minipage}{14cm}
\begin{tikzpicture}[
  every node/.style={font=\small},
  circ/.style={draw, rounded corners, inner sep=2pt, minimum height=10pt},
  redbox/.style={draw=red, rounded corners, inner sep=2pt},
  greenbox/.style={draw=green!60!black, rounded corners, inner sep=2pt},
  bluebox/.style={draw=blue!60!black, rounded corners, inner sep=2pt},
  arrow/.style={-{Stealth[scale=1]}, thick}
]

\node at (-1.25,-1) {$\sigma = \begin{bmatrix} 2&3&1 \end{bmatrix}^{-1}$};

% Ligne 1 (du haut)
\node (t1) at (0,0) {$\mathrm{Ind}_{\mathfrak{S_1} \times \mathfrak{S}_2}^{\mathfrak{S}_3}(V_1 \boxtimes^{\mathrm{ext}} W_2)_{123}$};
\node (eq1) at (2,0) {$=$};

\node (t2) at (4,0) {$V_1 \otimes W_{23}$};
\node (eq2) at (6,0) {$\oplus$};
\node (t3) at (8,0) {$V_2 \otimes W_{13}$};
\node (eq3) at (10,0) {$\oplus$};
\node (t4) at (12,0) {$V_3 \otimes W_{12}$};

% Ligne 2 (en dessous)
\node (t1bis) at (0,-2) {$\mathrm{Ind}_{\mathfrak{S_1} \times \mathfrak{S}_2}^{\mathfrak{S}_3}(V_1 \boxtimes^{\mathrm{ext}} W_2)_{231}$};
\node (eq1bis) at (2,-2) {$=$};

\node (t2bis) at (4,-2) {$V_2 \otimes W_{31}$};
\node (eq2bis) at (6,-2) {$\oplus$};
\node (t3bis) at (8,-2) {$V_3 \otimes W_{21}$};
\node (eq3bis) at (10,-2) {$\oplus$};
\node (t4bis) at (12,-2) {$V_1 \otimes W_{23}$};

% Flèches de correspondance
\draw[arrow] (t1.south) to[out=-90,in=90] (t1bis.north);
\draw[arrow, blue!70!black] (t2.south) to[out=-90,in=90] (t4bis.north);
\draw[arrow, green!70!black] (t3.south) to[out=-90,in=90] (t2bis.north);
\draw[arrow, red!70!black] (t4.south) to[out=-90,in=90] (t3bis.north);

\end{tikzpicture}
\end{minipage} \end{center}
\end{Remark}

\subsection{The $2$-Monoidal Structure on Equivariant Vector Bundles over Configuration Spaces} \label{2.3}

The Hadamard and Cauchy tensor products $\otimes$ and $\boxtimes$ are compatible in the sense that they endow $\mathcal{VB}_{\mathfrak{S}_{\bullet}}(M^{\bullet})$ with the structure of a symmetric $2$-monoidal category, also known as a symmetric duoidal category. The terminology 'duoidal' for this structure was introduced in \cite{batanin+markl}; we follow the terminology and conventions of \cite{aguiar+mahajan}, in order to avoid confusion with monoidal $2$-categories, which are different in nature.

We briefly recall the definition in order to fix notation. A $2$-monoidal category is a tuple $({\mathcal C}, \boxtimes, \mathbf{I}_{\boxtimes}, \otimes, \mathbf{I}_{\otimes})$, where \begin{itemize} 
\item[$\bullet$] $({\mathcal C}, \boxtimes, \mathbf{I}_{\boxtimes})$ and $({\mathcal C}, \otimes, \mathbf{I}_{\otimes})$ are monoidal categories,
\item[$\bullet$] There is a natural interchange morphism 
\[ \zeta_{A,B,C,D}: (A \otimes B) \boxtimes (C \otimes D) \longrightarrow (A \boxtimes C) \otimes (B \boxtimes D) \]
satisfying the usual associativity compatibility conditions,
\item[$\bullet$] There are three morphisms 
\[ \Delta: \mathbf{I}_{\boxtimes} \longrightarrow \mathbf{I}_{\boxtimes} \otimes \mathbf{I}_{\boxtimes}, \qquad \mu: \mathbf{I}_{\otimes} \boxtimes \mathbf{I}_{\otimes} \longrightarrow \mathbf{I}_{\otimes} \qquad \text{and} \qquad \nu: \mathbf{I}_{\boxtimes} \longrightarrow \mathbf{I}_{\otimes}, \]
such that $(\mathbf{I}_{\boxtimes}, \Delta, \nu)$ is a comonoid in $({\mathcal C}, \otimes, \mathbf{I}_{\otimes})$ and $(\mathbf{I}_{\otimes}, \mu, \nu)$ is a monoid in $({\mathcal C}, \boxtimes, \mathbf{I}_{\boxtimes})$.
\end{itemize}

Our goal is now to construct these data for $\mathcal{VB}_{\mathfrak{S}_{\bullet}}(M^{\bullet})$ with the Hadamard and Cauchy tensor products. The construction is induced by the compatibility between the Hadamard tensor product and the external tensor product: 
\[ (A \otimes B) \boxtimes^{\mathrm{ext}} (C \otimes D) \cong (A \boxtimes^{\mathrm{ext}} C) \otimes (B \boxtimes^{\mathrm{ext}} D), \] 
together with the oplax monoidality of the induction functors.

\begin{Theorem} 
The category $\Big( \mathcal{VB}_{\mathfrak{S}_{\bullet}}(M^{\bullet}), \boxtimes, \mathbf{I}_{\boxtimes}, \otimes, \mathbf{I}_{\otimes} \Big)$ is a symmetric $2$-monoidal category.
\end{Theorem}

\begin{proof}
We first construct the structural morphisms of the $2$-monoidal structure. For $\mathbf{A}, \mathbf{B}, \mathbf{C}, \mathbf{D} \in \mathcal{VB}_{\mathfrak{S}_{\bullet}}(M^{\bullet})$, the interchange map 
\[ \zeta_{\mathbf{A}, \mathbf{B}, \mathbf{C},\mathbf{D}}: (\mathbf{A} \otimes \mathbf{B}) \boxtimes (\mathbf{C} \otimes \mathbf{D}) \longrightarrow (\mathbf{A} \boxtimes \mathbf{C}) \otimes (\mathbf{B} \boxtimes \mathbf{D}) \]
is defined degreewise. In degree $n$, one has, using the canonical compatibility for the external tensor product
\[ \Big[ (\mathbf{A} \otimes \mathbf{B}) \boxtimes (\mathbf{C} \otimes \mathbf{D}) \Big]_n = \bigoplus \limits_{p+q=n} \mathrm{Ind}_{\mathfrak{S}_p \times \mathfrak{S}_q}^{\mathfrak{S}_n} \Big( \big( A \otimes B \big)_p \boxtimes^{\mathrm{ext}} \big( C \otimes D \big)_q \Big) \cong \bigoplus \limits_{p+q=n} \mathrm{Ind}_{\mathfrak{S}_p \times \mathfrak{S}_q}^{\mathfrak{S}_n} \Big( \big( A_p \boxtimes^{\mathrm{ext}} C_q \big) \otimes \big( B_p \boxtimes^{\mathrm{ext}} D_q \big) \Big), \]
then by the oplax monoidality of the induction functors, one obtains a natural map for each summand 
\[ \mathrm{Ind}_{\mathfrak{S}_p \times \mathfrak{S}_q}^{\mathfrak{S}_n} \Big( \big( A_p \boxtimes^{\mathrm{ext}} C_q \big) \otimes \big( B_p \boxtimes^{\mathrm{ext}} D_q \big) \Big) \longrightarrow \mathrm{Ind}_{\mathfrak{S}_p \times \mathfrak{S}_q}^{\mathfrak{S}_n} \big( A_p \boxtimes^{\mathrm{ext}} C_q \big) \otimes \mathrm{Ind}_{\mathfrak{S}_p \times \mathfrak{S}_q}^{\mathfrak{S}_n} \big( B_p \boxtimes^{\mathrm{ext}} D_q \big). \] 
Summing over all decompositions $p+q=n$ gives the desired morphism $\zeta_{\mathbf{A}, \mathbf{B}, \mathbf{C}, \mathbf{D}}$.

The morphism $\Delta: \mathbf{I}_{\boxtimes} \rightarrow \mathbf{I}_{\boxtimes} \otimes \mathbf{I}_{\boxtimes}$ is the obvious degreewise map: in degree $0$, it sends $\mathds{1}_{0}^{\boxtimes}$ to $\mathds{1}_0^{\boxtimes} \otimes \mathds{1}_0^{\boxtimes}$, and in positive degree it is zero. Similarly, 
\[ \nu: \mathbf{I}_{\boxtimes} \rightarrow \mathbf{I}_{\otimes} \] 
is the canonical degreewise inclusion, so that in degree $0$ one has $\nu \big( \mathds{1}_0^{\boxtimes} \big) = \mathds{1}_0^{\otimes}$. Finally, 
\[ \mu: \mathbf{I}_{\otimes} \boxtimes \mathbf{I}_{\otimes} \rightarrow \mathbf{I}_{\otimes} \]
is defined degreewise from the maps 
\[ \mu_{p,q}: \mathrm{Ind}_{\mathfrak{S}_p \times \mathfrak{S}_q}^{\mathfrak{S}_{p+q}} \Big( (\mathbf{I}_{\otimes})_p \boxtimes^{\mathrm{ext}} (\mathbf{I}_{\otimes})_q \Big) \longrightarrow (\mathbf{I}_{\otimes})_{p+q} \]
sending the canonical unit $\big( \mathds{1}_p^{\otimes}, \mathds{1}_{q}^{\otimes} \big)$ to $\mathds{1}_{p+q}^{\otimes}$. \newline

It remains to verify the axioms. The key point is that all structural morphisms introduced above are defined degreewise from the corresponding morphisms for the ordinary tensor product and the external tensor product, together with iterated induction. Thus every required diagram may be checked for each degree $n$.

For the interchange morphism, the associativity compatibility reduces degreewise to the following two facts. First, the canonical isomorphism 
\[ (A_p \otimes B_p) \boxtimes^{\mathrm{ext}} (C_q \otimes D_q) \cong (A_p \boxtimes^{\mathrm{ext}} C_q) \otimes (B_p \boxtimes^{\mathrm{ext}} D_q) \]
is natural and compatible with the associativity of both $\otimes$ and $\boxtimes^{\mathrm{ext}}$. Second, the oplax monoidality of the induction functor is itself natural and associative. Since $\zeta$ is obtained by composing these morphisms and then summing over all decompositions $p+q=n$, the associativity diagram for $\zeta$ commutes in each degree.

The unit compatibilities are equally straightforward. Because $\mathbf{I}_{\boxtimes}$ is concentrated in degree $0$, every Cauchy product involving $\mathbf{I}_{\boxtimes}$ has a unique nonzero summand, namely the one corresponding to the decomposition $0+n=n$ or $n+0=n$. Under the resulting identifications, the maps involving $\Delta$ reduce in degree $0$ to the identity 
\[ \mathds{1}_0^{\boxtimes} \longmapsto \mathds{1}_0^{\boxtimes} \otimes \mathds{1}_0^{\boxtimes} \]
and they vanish in every positive degree. This is exactly what is needed for the unit diagrams involving $\mathbf{I}_{\boxtimes}$ to commute. Similarly, the diagrams involving $\mu$ reduce degreewise to the fact that the multiplication on $\mathbf{I}_{\otimes}$ identifies the canonical section 
\[ \mathds{1}_p^{\otimes} \boxtimes^{\mathrm{ext}} \mathds{1}_q^{\otimes} \qquad \text{with} \qquad \mathds{1}_{p+q}^{\otimes}. \]

Finally, the compatibility of the unit data is checked directly. The object $\mathbf{I}_{\otimes}$ is a monoid for $\boxtimes$ because the multiplication $\mu$ is associative: for all $p,q,r \in \mathbb{N}$, one has 
\[ \mu_{p,q+r} \Big( \mathds{1}_p^{\otimes}, \mu_{q,r} \big( \mathds{1}_q^{\otimes}, \mathds{1}_r^{\otimes} \big) \Big) = \mathds{1}_{p+q+r}^{\otimes} = \mu_{p+q,r} \Big( \mu_{p,q} \big( \mathds{1}_p^{\otimes}, \mathds{1}_q^{\otimes} \big), \mathds{1}_{r}^{\otimes} \Big). \]
Its left and right unit laws follow from the fact that $\nu \big( \mathds{1}_0^{\boxtimes} \big) = \mathds{1}_0^{\otimes}$ and that $\nu$ vanishes in every positive degree. Dually, $\mathbf{I}_{\boxtimes}$ is a comonoid for $\otimes$ because $\Delta$ is coassociative and counital degreewise: in degree $0$, everything reduces to the identity $\mathds{1}_0^{\boxtimes} \mapsto \mathds{1}_0^{\boxtimes} \otimes \mathds{1}_0^{\boxtimes} \otimes \mathds{1}_0^{\boxtimes}$, and in positive degree all maps vanish.

Thus the data $(\zeta, \Delta, \mu, \nu)$ satisfy the axioms of a symmetric $2$-monoidal category, and therefore the category $\big( \mathcal{VB}_{\mathfrak{S}_{\bullet}}(M^{\bullet}), \boxtimes, \mathbf{I}_{\boxtimes}, \otimes, \mathbf{I}_{\otimes} \big)$ is symmetric $2$-monoidal. 
\end{proof}

\begin{Remark} \begin{itemize}
\item[$\bullet$] The object $\mathbf{I}_{\otimes}$ is in fact a commutative monoid in $(\boxtimes, \mathbf{I}_{\boxtimes})$. Indeed, for every $p,q \in \mathbb{N}$, one has
\[ \mu_{p,q} \big( \mathds{1}_p^{\otimes}, \mathds{1}_q^{\otimes} \big) = \mathds{1}_{p+q}^{\otimes} = \mu_{q,p} \big( \mathds{1}_q^{\otimes}, \mathds{1}_p^{\otimes} \big). \]
Thus the multiplication is invariant under exchanging the two factors, and therefore commutative.
\item[$\bullet$] Similarly, $\mathbf{I}_{\boxtimes}$ is a cocommutative comonoid in $(\otimes, \mathbf{I}_{\otimes})$. Indeed, the comultiplication $\Delta: \mathbf{I}_{\boxtimes} \rightarrow \mathbf{I}_{\boxtimes} \otimes \mathbf{I}_{\boxtimes}$ is symmetric with respect to the $\otimes$-braiding: in degree $0$, one has 
\[ \Delta(\mathds{1}_0^{\boxtimes}) = \mathds{1}_0^{\boxtimes} \otimes \mathds{1}_0^{\boxtimes} \]
which is fixed by the symmetry of the tensor product, while in positive degree all components vanish. Hence $\Delta$ is cocommutative. 
\item[$\bullet$] In degree $0$, the unit morphism $\nu: \mathbf{I}_{\boxtimes} \rightarrow \mathbf{I}_{\otimes}$ is an isomorphism. More precisely, $\nu_0: (\mathbf{I}_{\boxtimes})_0 \overset{\cong} \longrightarrow (\mathbf{I}_{\otimes})_0$ identifies the generator $\mathds{1}_0^{\boxtimes}$ and $\mathds{1}_0^{\otimes}$. In particular, when no confusion is possible, we may regard these two elements as canonically identified.
\end{itemize}
\end{Remark}

\begin{Remark}[Alternative choice of interchange map] One may also define another interchange map 
\[ \xi_{\mathbf{A}, \mathbf{B}, \mathbf{C}, \mathbf{D}}: (\mathbf{A} \boxtimes \mathbf{B}) \otimes (\mathbf{C} \boxtimes \mathbf{D}) \longrightarrow (\mathbf{A} \otimes \mathbf{C}) \boxtimes (\mathbf{B} \otimes \mathbf{D}) \]
obtained by projecting onto the diagonal summands. More precisely, in degree $n$ one has
\[ \begin{array}{lll} \Big( (\mathbf{A} \boxtimes \mathbf{B}) \otimes (\mathbf{C} \boxtimes \mathbf{D}) \Big)_n &=& \bigoplus \limits_{\substack{i+j=n \\ \omega \in \mathrm{Ush}(i,j)}} \bigoplus \limits_{\substack{k+\ell = n \\ \widetilde{\omega} \in \mathrm{Ush}(k, \ell)}} \Big[ \omega^* \big( \mathbf{A}_i \boxtimes^{\mathrm{ext}} \mathbf{B}_j \big) \otimes \widetilde{\omega}^* \big( \mathbf{C}_k \boxtimes^{\mathrm{ext}} \mathbf{D}_{\ell} \big) \Big] \\ 
&\twoheadrightarrow& \bigoplus \limits_{\substack{i+j=n \\ \omega \in \mathrm{Ush}(i,j)}} \Big[ \omega^*  \big( \mathbf{A}_i \boxtimes^{\mathrm{ext}} \mathbf{B}_j \big) \otimes \omega^* \big( \mathbf{C}_i \boxtimes^{\mathrm{ext}} \mathbf{D}_j \big) \Big] \\ 
&\cong& \bigoplus \limits_{\substack{i+j=n \\ \omega \in \mathrm{Ush}(i,j)}} \omega^* \Big( \big( \mathbf{A}_i \otimes \mathbf{C}_i \big) \boxtimes^{\mathrm{ext}} \big( \mathbf{B}_j \otimes \mathbf{D}_j \big) \Big) \\ 
&=& \Big( \big( \mathbf{A} \otimes \mathbf{C} \big) \boxtimes \big( \mathbf{B} \otimes \mathbf{D} \big) \Big)_n. \end{array} \]
Here we have used the canonical compatibility between pullback and tensor product, together with the usual interchange between tensor product and external tensor product. Equivalently, $\xi$ is obtained by sending all off-diagonal summands, namely those for which $(i,j,\omega) \neq (k,\ell, \widetilde{\omega})$ to zero.

This yields a second $2$-monoidal structure $\Big( \mathcal{VB}_{\mathfrak{S}_{\bullet}}(M^{\bullet}), \otimes, \mathbf{I}_{\otimes}, \boxtimes, \mathbf{I}_{\boxtimes} \Big)$, where the order of the two tensor products indicates which interchange map is being used. In this alternative setting: \begin{enumerate} 
\item The morphism $\widetilde{\Delta}: \mathbf{I}_{\otimes} \rightarrow \mathbf{I}_{\otimes} \boxtimes \mathbf{I}_{\otimes}$ is the diagonal map. Indeed, 
\[ \Big( \mathbf{I}_{\otimes} \boxtimes \mathbf{I}_{\otimes} \Big)_n \cong \bigoplus \limits_{\substack{i+j=n \\ \omega \in \mathrm{Ush}(i,j)}} \big( \mathbf{I}_{\otimes} \big)_n, \]
so $\widetilde{\Delta}_n: (\mathbf{I}_{\otimes})_n \rightarrow \big(\mathbf{I}_{\otimes} \boxtimes \mathbf{I}_{\otimes} \big)_n$ is naturally defined by the diagonal embedding. Under the above decomposition, $\widetilde{\Delta}_n(\mathds{1}_n^{\otimes})$ is the diagonal element whose entries are all equal to $\mathds{1}_n^{\otimes}$.
\item The morphism $\widetilde{\mu}: \mathbf{I}_{\boxtimes} \otimes \mathbf{I}_{\boxtimes} \longrightarrow \mathbf{I}_{\boxtimes}$ is simply the inverse of the comultiplication $\Delta$ constructed for the previous interchange morphism. In degree $0$, it is therefore characterized by 
\[ \widetilde{\mu} \Big( \mathds{1}_0^{\boxtimes} \otimes \mathds{1}_0^{\boxtimes} \Big) = \mathds{1}_0^{\boxtimes}, \]
and it vanishes in every positive degree.
\item The unit morphism $\widetilde{\nu}: \mathbf{I}_{\otimes} \longrightarrow \mathbf{I}_{\boxtimes}$ is the projection onto the degree $0$ component. Equivalently 
\[\widetilde{\nu}_0(\mathds{1}_0^{\otimes}) = \mathds{1}_0^{\boxtimes}, \qquad \widetilde{\nu}_n(\mathds{1}_n^{\otimes}) = 0 \; \text{for all } n \geqslant 1. \]
\end{enumerate}
In this alternative setting, $\mathbf{I}_{\otimes}$ is a cocommutative comonoid in $(\boxtimes, \mathbf{I}_{\boxtimes})$ and $\mathbf{I}_{\boxtimes}$ is a commutative monoid in $(\otimes, \mathbf{I}_{\otimes})$.
These two interchange maps are related by
\[ \xi_{\mathbf{A}, \mathbf{C}, \mathbf{B}, \mathbf{D}} \circ \zeta_{\mathbf{A}, \mathbf{B}, \mathbf{C}, \mathbf{D}} = \mathrm{Id}_{(\mathbf{A} \otimes \mathbf{B}) \boxtimes (\mathbf{C} \otimes \mathbf{D})}. \]
The two $2$-monoidal structures give rise to different notions of monoid and comonoid objects, and in this sense they are dual to one another. In the present paper, we work with the interchange map $\zeta$, since it yields the largest class of monoid objects.
\end{Remark}

%% file: 4_Algebras.tex
\section{Equivariant Algebra Bundles over Configuration Spaces} 

Let $({\mathcal C}, \otimes, I)$ be a symmetric monoidal category, recall that a monoid object in ${\mathcal C}$ is an object $M \in {\mathcal C}$ with a multiplication map $M \otimes M \rightarrow M$ and a unit map $I \rightarrow M$ satisfying the usual associative and unit axioms. Classical examples include monoids in $(\mathbf{Set}, \times)$ and unital associative algebras in $(\mathbf{Vect}_k, \otimes)$.

The purpose of this section is to study monoid objects associated with the two monoidal structures introduced above, namely the Hadamard tensor product and the Cauchy tensor product, in the setting of equivariant vector bundles over configuration spaces. We describe the corresponding algebra objects and construct their associated free algebras. \newline

We begin with the Hadamard tensor product, whose behavior is close to the standard tensor product of vector bundles and serves as a guiding example. We then turn to the Cauchy tensor product, where the presence of induced equivariance and underlying combinatorics gives rise to a richer class of algebra objects. 

Finally, we introduce the notion of a $2$-algebra, namely an object endowed with both a Hadamard and a Cauchy algebra structure, together with compatibility conditions encoded by an interchange law between the two tensor products. 

\subsection{Equivariant Hadamard Algebra Bundles over Configuration Spaces} \label{3.1}

\begin{Definition} 
An \textbf{equivariant Hadamard algebra} (or \textbf{equivariant $\otimes$-algebra}) is a monoid object in the monoidal category $\big( \mathcal{VB}_{\mathfrak{S}_{\bullet}}(M^{\bullet}), \otimes, \mathbf{I}_{\otimes} \big)$. 
\end{Definition}

Equivalently, an equivariant $\otimes$-algebra is an $\mathfrak{S}$-equivariant vector bundle $\mathbf{A}$ together with two maps 
\[ m^{\otimes}: \mathbf{A} \otimes \mathbf{A} \longrightarrow \mathbf{A} \qquad \text{and} \qquad u^{\otimes}: \mathbf{I}_{\otimes} \longrightarrow \mathbf{A} \]
such that $m^{\otimes}$ and $u^{\otimes}$ satisfy the usual associativity and unit axioms. Equivalently, each $A_n$ is an $\mathfrak{S}_n$-equivariant algebra bundle in the sense that
\[ \forall a,b \in A_n, \quad \forall \sigma \in \mathfrak{S}_n, \qquad m_n^{\otimes}(\sigma \cdot a, \sigma \cdot b) = \sigma \cdot m_n^{\otimes}(a,b) \] 
and the unit $\mathds{1}_n^{\otimes}$ of $A_n$ is fixed by the action, i.e. 
\[ \forall \sigma \in \mathfrak{S}_n, \qquad \sigma \cdot \mathds{1}_n^{\otimes} = \mathds{1}_n^{\otimes}. \] 
We will say that $\mathbf{A}$ is commutative if $m^{\otimes} \circ \beta^{\otimes} = m^{\otimes}$, where $\beta^{\otimes}$ is the $\otimes$-braiding of $\mathcal{VB}_{\mathfrak{S}_{\bullet}}(M^{\bullet})$; equivalently, for all $n \in \mathbb{N}$ and all $a,b \in A_n$, one has $m_n^{\otimes}(a,b) = m_n^{\otimes}(b,a)$. One may also consider graded-commutative $\otimes$-algebras where the braiding is twisted by the Koszul sign rule, i.e. $\beta_{\mathrm{sgn}}^{\otimes}(a \otimes b) = (-1)^{|a| |b|} b \otimes a$. Thus the condition becomes 
\[ m_n^{\otimes}(a,b) = (-1)^{|a| |b|} m_n^{\otimes}(b,a). \]

A $\otimes$-morphism of equivariant $\otimes$-algebras is an equivariant bundle map $f: \mathbf{A} \longrightarrow \mathbf{B}$ that is compatible with the $\otimes$-multiplication and the $\otimes$-unit in the usual sense; equivalently, it is given degreewise by maps $f_n: A_n \longrightarrow B_n$ satisfying
\[ f_n(\sigma \cdot a) = \sigma \cdot f_n(a) \qquad \text{and} \qquad  
f_n \big( m_n^{\otimes}(a,b) \big) = m_n^{\otimes} \big( f_n(a), f_n(b) \big). \]
We will denote by: \begin{itemize} 
\item[$\bullet$] $\mathbf{Alg}^{\otimes} \big( \mathcal{VB}_{\mathfrak{S}_{\bullet}}(M^{\bullet}) \big)$ the category of equivariant $\otimes$-algebras.
\item[$\bullet$] $\mathbf{CAlg}^{\otimes} \big( \mathcal{VB}_{\mathfrak{S}_{\bullet}}(M^{\bullet}) \big)$ the full subcategory of commutative ones.
\item[$\bullet$] $\mathbf{CAlg}_{\mathrm{sgn}}^{\otimes} \big( \mathcal{VB}_{\mathfrak{S}_{\bullet}}(M^{\bullet}) \big)$ the full subcategory of graded-commutative ones.
\end{itemize}
Since the Hadamard tensor product is defined degreewise, a $\otimes$-algebra structure on $\mathbf{A}$ is equivalently given by independent $\mathfrak{S}_n$-equivariant algebra bundle structures on each $A_n$, with no interaction between different degrees.

\begin{Example} If $(\mathbf{A}_1, m_1^{\otimes}, u_1^{\otimes})$ and $(\mathbf{A}_2, m_2^{\otimes}, u_2^{\otimes})$ are commutative $\otimes$-algebras, then $\mathbf{A}_1 \otimes \mathbf{A}_2$ is naturally a commutative $\otimes$-algebra with multiplication $m : (\mathbf{A}_1 \otimes \mathbf{A}_2) \otimes (\mathbf{A}_1 \otimes \mathbf{A}_2) \longrightarrow \mathbf{A}_1 \otimes \mathbf{A}_2$ given by
\[ \forall a_1, a_2 \in (\mathbf{A}_1)_n, \quad \forall b_1, b_2 \in (\mathbf{A}_2)_n, \qquad m_n \Big( (a_1 \otimes a_2) \otimes (b_1 \otimes b_2) \Big) = m_{1;n}^{\otimes}(a_1, b_1) \otimes m_{2;n}^{\otimes}(a_2, b_2) \] 
and unit $u : \mathbf{I}_{\otimes} \longrightarrow \mathbf{A}_1 \otimes \mathbf{A}_2$ given by 
\[ u_n \big( \mathds{1}_n^{\otimes} \big) = u_{1;n}^{\otimes}(\mathds{1}_n^{\otimes}) \otimes u_{2;n}^{\otimes}(\mathds{1}_n^{\otimes}) \]
In fact, by using the componentwise description, one easily obtains that $\otimes$ is the coproduct of the category $\mathbf{CAlg} \big( \mathcal{VB}_{\mathfrak{S}_{\bullet}}(M^{\bullet}) \big)$.
\end{Example}

\subsubsection{Equivariant Hadamard Tensor Algebra Bundles} \label{3.1.1}

\begin{Definition} 
The \textbf{equivariant Hadamard tensor algebra bundle} generated by $\mathbf{V} \in \mathcal{VB}_{\mathfrak{S}_{\bullet}}(M^{\bullet})$ is defined by
\[ \mathbf{T}^{\otimes}(\mathbf{V}) = \bigoplus \limits_{n \in \mathbb{N}} \mathbf{V}^{\otimes n} \qquad \text{with} \qquad \mathbf{V}^{\otimes 0} = \mathbf{I}_{\otimes}. \]
\end{Definition}

There is an obvious injection map $\mathbf{V} \hookrightarrow \mathbf{T}^{\otimes}(\mathbf{V})$, and one has
\[ \big[ \mathbf{T}^{\otimes}(\mathbf{V}) \big]_k = \bigoplus \limits_{n \in \mathbb{N}} V_k^{\otimes n} = \mathbf{T}^{\otimes}(V_k). \]
Thus the Hadamard tensor algebra is constructed degreewise: each bundle $V_k \rightarrow M^k$ generates the usual tensor algebra bundle $\mathbf{T}^{\otimes}(V_k)$ over $M^k$. 

Its multiplication $\mathbf{T}^{\otimes}(\mathbf{V}) \otimes \mathbf{T}^{\otimes}(\mathbf{V}) \longrightarrow \mathbf{T}^{\otimes}(\mathbf{V})$ is given by elementary maps 
\[ m_{p,q}^{\otimes}: \Big( (v_1 \otimes \ldots \otimes v_p), (w_1 \otimes \ldots \otimes w_q) \Big) \in  \mathbf{V}^{\otimes p} \times \mathbf{V}^{\otimes q} \longmapsto v_1 \otimes \ldots \otimes v_p \otimes w_1 \otimes \ldots \otimes w_q \in \mathbf{V}^{\otimes (p+q)}. \]
As usual, we will write $m^{\otimes}(a,b)$ simply as $a \otimes b$.

\begin{Proposition} 
The equivariant Hadamard tensor algebra $\mathbf{T}^{\otimes}(\mathbf{V})$ is the free equivariant $\otimes$-algebra bundle generated by $\mathbf{V} \in \mathcal{VB}_{\mathfrak{S}_{\bullet}}(M^{\bullet})$. More precisely, for every equivariant $\otimes$-algebra $\mathbf{A} \in \mathbf{Alg}^{\otimes} \big( \mathcal{VB}_{\mathfrak{S}_{\bullet}}(M^{\bullet}) \big)$ and every equivariant linear bundle map $f: \mathbf{V} \longrightarrow \mathbf{A}$, there exists a unique $\otimes$-algebra bundle morphism $\widetilde{f}: \mathbf{T}^{\otimes}(\mathbf{V}) \longrightarrow \mathbf{A}$ such that the following diagram commutes 
\[
\xymatrix{
\mathbf{V} \ar[rr]^{\iota} \ar[rrd]_{f} && \mathbf{T}^{\otimes}(\mathbf{V}) \ar@{.>}[d]^{\exists! \widetilde{f}} \\ 
&& \mathbf{A}
}
\]
\end{Proposition}

Equivalently, the functor $\mathbf{T}^{\otimes}: \mathcal{VB}_{\mathfrak{S}_{\bullet}}(M^{\bullet}) \longrightarrow \mathbf{Alg}^{\otimes} \big( \mathcal{VB}_{\mathfrak{S}_{\bullet}}(M^{\bullet}) \big)$ is left adjoint to the forgetful functor $\mathbf{Alg}^{\otimes} \big( \mathcal{VB}_{\mathfrak{S}_{\bullet}}(M^{\bullet}) \big) \rightarrow \mathcal{VB}_{\mathfrak{S}_{\bullet}}(M^{\bullet})$.

\subsubsection{Equivariant Hadamard Symmetric Algebra Bundles} \label{3.1.2}

In the context of linear algebra, the symmetric algebra generated by a vector space $V$ is constructed as follows: for each $n \in \mathbb{N}$, the space $V^{\otimes n}$ carries a natural $\mathfrak{S}_n$-action by permutation of the factors. The space $S^{\otimes n}(V)$ is then the space of coinvariants for this action, namely, the quotient of $V^{\otimes n}$ by the subspace spanned by vectors of the form $v - \sigma \cdot v$ for $v \in V^{\otimes n}$ and $\sigma \in \mathfrak{S}_n$. The symmetric algebra is then given by the direct sum 
\[ S^{\otimes}(V) = \bigoplus \limits_{n \in \mathbb{N}} S^{\otimes n}(V). \]
We follow this standard construction in the setting of equivariant vector bundles.

Fix $\mathbf{V} \in \mathcal{VB}_{\mathfrak{S}_{\bullet}}(M^{\bullet})$. Since $V_k$ is a $\mathfrak{S}_k$-equivariant vector bundle over $M^k$, the tensor power $V_k^{\otimes n}$ is also a $\mathfrak{S}_k$-equivariant vector bundle over $M^k$. Note that $\mathfrak{S}_n$ acts on $V_k^{\otimes n}$ by permuting the factors
\[ \tau \cdot (v_1 \otimes \ldots \otimes v_n) = v_{\tau^{-1}(1)} \otimes \ldots \otimes v_{\tau^{-1}(n)} \]
and that this action commutes with the $\mathfrak{S}_k$-equivariant action. In particular, the subbundle 
\[ W_{k,n}^{\otimes} = \mathrm{Span} \left\{ \tau \cdot v - v, \; v \in V_k^{\otimes n}, \; \tau \in \mathfrak{S}_n \right\} \subset V_k^{\otimes n} \]
is $\mathfrak{S}_k$-equivariant and hence the quotient 
\[ \mathbf{S}^{\otimes n}(V_k) = V_k^{\otimes n}/W_{k,n}^{\otimes} \]
is a well-defined $\mathfrak{S}_k$-equivariant vector bundle. Thus the Hadamard symmetric algebra is obtained degreewise: each bundle $V_k \rightarrow M^k$ generates the usual symmetric algebra bundle $\mathbf{S}^{\otimes}(V_k)$ over $M^k$.

\begin{Definition} 
The \textbf{equivariant Hadamard symmetric algebra bundle} generated by $\mathbf{V} \in \mathcal{VB}_{\mathfrak{S}_{\bullet}}(M^{\bullet})$ is defined by
\[ \mathbf{S}^{\otimes}(\mathbf{V}) = \bigoplus \limits_{n \in \mathbb{N}} \mathbf{S}^{\otimes n}(\mathbf{V}) \qquad \text{where} \qquad \big[ \mathbf{S}^{\otimes n}(\mathbf{V}) \big]_k = \mathbf{S}^{\otimes n}(V_k). \]
\end{Definition}

The concatenation of tensors $\mathbf{V}^{\otimes p} \times \mathbf{V}^{\otimes q} \rightarrow \mathbf{V}^{\otimes (p+q)}$ is $(\mathfrak{S}_p \times \mathfrak{S}_q)$-equivariant, where the action is given by permutation of the factors. Hence it induces a map 
\[ \mathbf{S}^{\otimes p}(\mathbf{V}) \times \mathbf{S}^{\otimes q}(\mathbf{V}) \longrightarrow \mathbf{S}^{\otimes(p+q)}(\mathbf{V}) \]
which is the multiplication map of $\mathbf{S}^{\otimes}(\mathbf{V})$. If $[v] \in \mathbf{S}^{\otimes p}(\mathbf{V})$ and $[w] \in \mathbf{S}^{\otimes q}(\mathbf{V})$, then we will denote by $[v] \odot [w]$ the equivalence class of $v \otimes w$. Whenever no confusion can arise, we will omit the bracket notation and simply write products like $v \odot w$. Moreover, the tensors $v \otimes w$ and $w \otimes v$ differ by a permutation exchanging the two blocks of sizes $p$ and $q$, which belongs to $\mathfrak{S}_{p+q}$. Therefore they induce the same element in $\mathbf{S}^{\otimes}(\mathbf{V})$ and hence the multiplication on $\mathbf{S}^{\otimes}(\mathbf{V})$ is commutative.

Since the Hadamard symmetric algebra is constructed degreewise, the following holds
\begin{Proposition} 
The equivariant Hadamard symmetric algebra $\mathbf{S}^{\otimes}(\mathbf{V})$ is the free equivariant commutative $\otimes$-algebra generated by $\mathbf{V}$. More precisely, for every equivariant commutative $\otimes$-algebra bundle $\mathbf{A} \in \mathbf{CAlg}^{\otimes} \big( \mathcal{VB}_{\mathfrak{S}_{\bullet}}(M^{\bullet}) \big)$ and every equivariant bundle map $f: \mathbf{V} \longrightarrow \mathbf{A}$, there exists a unique $\otimes$-algebra bundle morphism $\widetilde{f}: \mathbf{S}^{\otimes}(\mathbf{V}) \longrightarrow \mathbf{A}$ such that the following diagram commutes 
\[ \xymatrix{
\mathbf{V} \ar[rr]^{\iota} \ar[rrd]_{f} && \mathbf{S}^{\otimes}(\mathbf{V}) \ar@{.>}[d]^{\exists! \widetilde{f}} \\ 
&& \mathbf{A}
}
\]
\end{Proposition}

Equivalently, the functor $\mathbf{S}^{\otimes}: \mathcal{VB}_{\mathfrak{S}_{\bullet}}(M^{\bullet}) \longrightarrow \mathbf{CAlg}^{\otimes} \big( \mathcal{VB}_{\mathfrak{S}_{\bullet}}(M^{\bullet}) \big)$ is left adjoint to the forgetful functor $\mathbf{CAlg}^{\otimes} \big( \mathcal{VB}_{\mathfrak{S}_{\bullet}}(M^{\bullet}) \big) \rightarrow \mathcal{VB}_{\mathfrak{S}_{\bullet}}(M^{\bullet})$.

\begin{Remark}[Hadamard exterior algebra] Similarly, one may construct the \textbf{equivariant Hadamard exterior algebra bundle} degreewise: each bundle $V_k \rightarrow M^k$ generates the usual exterior algebra bundle $\mathbf{\Lambda}^{\otimes}(V_k)$ over $M^k$. It is defined by
\[ \mathbf{\Lambda}^{\otimes}(\mathbf{V}) = \bigoplus \limits_{n \in \mathbb{N}} \mathbf{\Lambda}^{\otimes n}(\mathbf{V}) \qquad \text{where} \qquad \big[ \mathbf{\Lambda}^{\otimes n}(\mathbf{V}) \big]_k = \mathbf{\Lambda}^{\otimes n}(V_k) \]
and it is the free equivariant graded-commutative $\otimes$-algebra generated by $\mathbf{V}$.
\end{Remark}

\subsection{Equivariant Cauchy Algebra Bundles over Configuration Spaces} \label{3.2}

\begin{Definition}
An \textbf{equivariant Cauchy algebra} (or \textbf{equivariant $\boxtimes$-algebra}) is a monoid object in the monoidal category $\big( \mathcal{VB}_{\mathfrak{S}_{\bullet}}(M^{\bullet}), \boxtimes, \mathbf{I}_{\boxtimes} \big)$. \end{Definition}

Equivalently, an equivariant $\boxtimes$-algebra bundle is an $\mathfrak{S}$-equivariant vector bundle $\mathbf{A}$ together with two maps 
\[ m^{\boxtimes}: \mathbf{A} \boxtimes \mathbf{A} \longrightarrow \mathbf{A} \qquad \text{and} \qquad u^{\boxtimes}: \mathbf{I}_{\boxtimes} \longrightarrow \mathbf{A} \]
such that $m^{\boxtimes}$ and $u^{\boxtimes}$ are associative and satisfy the unit axioms in the usual categorical sense. Since
\[ \big[ \mathbf{A} \boxtimes \mathbf{A} \big]_n = \bigoplus \limits_{p+q=n} \mathrm{Ind}_{\mathfrak{S}_p \times \mathfrak{S}_q}^{\mathfrak{S}_n} \big( {A}_p \boxtimes^{\mathrm{ext}} {A}_q \big), \]
a multiplication map $m^{\boxtimes}: \mathbf{A} \boxtimes \mathbf{A} \rightarrow \mathbf{A}$ is determined by maps $\mathrm{Ind}_{\mathfrak{S}_p \times \mathfrak{S}_q}^{\mathfrak{S}_n} \big( A_p \boxtimes^{\mathrm{ext}} A_q \big) \longrightarrow {A}_{p+q}$ and, by the adjunction property of the induced equivariance functor, equivalently by $(\mathfrak{S}_p \times \mathfrak{S}_q)$-equivariant bundle maps 
\[ m_{p,q}^{\boxtimes}: A_p \boxtimes^{\mathrm{ext}} A_q \longrightarrow A_{p+q}. \]
These may be seen as ordinary two-slot multiplication maps: an element $v$ above a configuration $(x_1, \ldots, x_p)$ and an element $w$ above a configuration $(y_1, \ldots, y_q)$ are multiplied and placed above the concatenated configuration $(x_1, \ldots, x_p, y_1, \ldots, y_q)$. The unit is then $\mathds{1}_0^{\boxtimes}$ which lies only over the vacuum state $\emptyset$. Thus, the $\boxtimes$-multiplication combines not only the algebra elements but also the underlying configurations of points. Its associativity is expressed by 
\[ m^{\boxtimes}_{p+q, r} \circ \Big( m_{p,q}^{\boxtimes} \boxtimes^{\mathrm{ext}} \mathrm{Id}_{A_r} \Big) = m_{p,q+r}^{\boxtimes} \circ \Big( \mathrm{Id}_{A_p} \boxtimes^{\mathrm{ext}} m_{q,r}^{\boxtimes} \Big) \]
as an $(\mathfrak{S}_p \times \mathfrak{S}_q \times \mathfrak{S}_r)$-equivariant map covering the concatenation map $M^p \times M^q \times M^r \rightarrow M^{p+q+r}$. 

We say that $\mathbf{A}$ is commutative if $m^{\boxtimes} \circ \beta^{\boxtimes} = m^{\boxtimes}$, where $\beta^{\boxtimes}$ is the $\boxtimes$-braiding of $\mathcal{VB}_{\mathfrak{S}_{\bullet}}(M^{\bullet})$. Equivalently, if $a$ lies above the configuration $(x_1, \ldots, x_p)$ and $b$ lies above the configuration $(y_1, \ldots, y_q)$, we have
\[ m_{p,q}^{\boxtimes}(a,b) = \tau_{p,q} \cdot m_{q,p}^{\boxtimes}(b,a) \]
where $\tau_{p,q}$ is the block permutation $\begin{bmatrix} q+1 & \ldots & q+p &|& 1 & \ldots & q \end{bmatrix}$ sending the configuration $(y_1, \ldots, y_q, x_1, \ldots, x_p)$ to $(x_1, \ldots, x_p, y_1, \ldots, y_q)$. One may likewise consider graded-commutative $\boxtimes$-algebras, for which the braiding is replaced by $\beta^{\boxtimes}_{\mathrm{sign}}$, so that 
\[ m_{p,q}^{\boxtimes}(a,b) = (-1)^{|a| |b|} \tau_{p,q} \cdot m_{q,p}^{\boxtimes}(b,a). \]
A $\boxtimes$-morphism of equivariant $\boxtimes$-algebras is an equivariant bundle map $f: \mathbf{A} \longrightarrow \mathbf{B}$ that is compatible with the $\boxtimes$-multiplication and the $\boxtimes$-unit in the usual sense. Equivalently, it is given degreewise by equivariant $\boxtimes$-algebra bundle maps $f_n: A_n \longrightarrow B_n$ such that 
\[ f_n(\sigma \cdot a) = \sigma \cdot f_n(a) \qquad \text{and} \qquad f_{p+q} \big( m_{p,q}^{\boxtimes}(a,b) \big) = m_{p,q}^{\boxtimes} \big( f_p(a), f_q(b) \big). \]
We will denote by: \begin{itemize} 
\item[$\bullet$] $\mathbf{Alg}^{\boxtimes} \big( \mathcal{VB}_{\mathfrak{S}_{\bullet}}(M^{\bullet}) \big)$ the category of equivariant $\boxtimes$-algebras.
\item[$\bullet$] $\mathbf{CAlg}^{\boxtimes} \big( \mathcal{VB}_{\mathfrak{S}_{\bullet}}(M^{\bullet}) \big)$ the full subcategory of commutative ones.
\item[$\bullet$] $\mathbf{CAlg}_{\mathrm{sgn}}^{\boxtimes} \big( \mathcal{VB}_{\mathfrak{S}_{\bullet}}(M^{\bullet}) \big)$ the full subcategory of graded-commutative ones.
\end{itemize}
In contrast with the Hadamard product, the Cauchy product combines both the algebra elements and the underlying configurations of points in $M$.

\begin{Example}
Let $(\mathbf{A}_1, m_1^{\boxtimes}, u_1^{\boxtimes})$ and $(\mathbf{A}_2, m_2^{\boxtimes}, u_2^{\boxtimes})$ be commutative $\boxtimes$-algebras. \begin{itemize}
\item[$\bullet$] The object $\mathbf{A}_1 \boxtimes \mathbf{A}_2 =: \mathbf{B}$ is naturally a commutative $\boxtimes$-algebra with multiplication $m : (\mathbf{A}_1 \boxtimes \mathbf{A}_2) \boxtimes (\mathbf{A}_1 \boxtimes \mathbf{A}_2) \longrightarrow \mathbf{A}_1 \boxtimes \mathbf{A}_2$ obtained as follows:
\[ m_{p,q}^{\boxtimes} : \mathbf{B}_p \boxtimes^{\mathrm{ext}} \mathbf{B}_q \longrightarrow \mathbf{B}_{p+q}. \]
These maps themselves are obtained by the maps 
\[ m_{p_1, p_2; q_1, q_2} : \Big( (\mathbf{A}_1)_{p_1} \boxtimes^{\mathrm{ext}} (\mathbf{A}_2)_{p_2} \Big) \boxtimes^{\mathrm{ext}} \Big( (\mathbf{A}_1)_{q_1} \boxtimes^{\mathrm{ext}} (\mathbf{A}_2)_{q_2} \Big) \longrightarrow (\mathbf{A}_1)_{p_1 + q_1} \boxtimes^{\mathrm{ext}} (\mathbf{A}_2)_{p_2 + q_2} \] 
with $p = p_1 + p_2$ and $q = q_1 + q_2$, defined by 
\[ m_{p_1, p_2; q_1, q_2} \big( (a_1 \boxtimes a_2) \boxtimes (b_1 \boxtimes b_2) \big) = m_{1;p_1,q_1}^{\boxtimes}(a_1, b_1) \boxtimes m_{2;p_2,q_2}^{\boxtimes}(a_2, b_2) \]
and unit $u : \mathbf{I}_{\boxtimes} \longrightarrow \mathbf{A}_1 \boxtimes \mathbf{A}_2$ given by
\[ u_0 \big( \mathds{1}_{0}^{\boxtimes} \big) = u_{1;0}^{\boxtimes} \big( \mathds{1}_0^{\boxtimes} \big) \boxtimes u_{2;0}^{\boxtimes} \big( \mathds{1}_0^{\boxtimes} \big). \]
In fact, $\boxtimes$ is the coproduct of the category $\mathbf{CAlg}^{\boxtimes} \big( \mathcal{VB}_{\mathfrak{S}_{\bullet}}(M^{\bullet}) \big)$, since in any symmetric monoidal category, the tensor product is the coproduct in the category of commutative monoid objects (see \cite{johnstone}, C1.1).
\item[$\bullet$] The object $\mathbf{A}_1 \otimes \mathbf{A}_2$ has a natural $\boxtimes$-algebra structure, whose $\boxtimes$-multiplication is defined using the interchange map $\zeta$ by
\[ m : (\mathbf{A}_1 \otimes \mathbf{A}_2) \boxtimes (\mathbf{A}_1 \otimes \mathbf{A}_2) \overset{\zeta} \longrightarrow (\mathbf{A}_1 \boxtimes \mathbf{A}_1) \otimes (\mathbf{A}_2 \boxtimes \mathbf{A}_2) \overset{m_{1}^{\boxtimes} \otimes m_{2}^{\boxtimes}} \longrightarrow \mathbf{A}_1 \otimes \mathbf{A}_2. \] 
Concretely, for each $p,q \in \mathbb{N}$, the elementary multiplication maps
\[ m_{p,q} : \Big( (\mathbf{A}_1)_p \otimes (\mathbf{A}_2)_p \Big) \boxtimes^{\mathrm{ext}} \Big( (\mathbf{A}_1)_q \otimes (\mathbf{A}_2)_q \Big) \longrightarrow (\mathbf{A}_1)_{p+q} \otimes (\mathbf{A}_2)_{p+q} \]
defined by 
\[ m_{p,q} \big( (a_1 \otimes a_2) \boxtimes (b_1 \otimes b_2) \big) = m_{1;p,q}^{\boxtimes}(a_1, b_1) \otimes m_{2;p,q}^{\boxtimes}(a_2, b_2) \]
and unit $u : \mathbf{I}_{\boxtimes} \overset{\Delta} \longrightarrow \mathbf{I}_{\boxtimes} \otimes \mathbf{I}_{\boxtimes} \overset{u_1^{\boxtimes} \otimes u_2^{\boxtimes}} \longrightarrow \mathbf{A}_1 \otimes \mathbf{A}_2$ given by 
\[ u \big( \mathds{1}^{\boxtimes} \big) = u_1^{\boxtimes}(\mathds{1}^{\boxtimes}) \otimes u_2^{\boxtimes}(\mathds{1}^{\boxtimes}). \]
In fact, the bifunctor $\otimes$ lifts from $\mathcal{VB}_{\mathfrak{S}_{\bullet}}(M^{\bullet})$ to $\mathbf{Alg}^{\boxtimes} \big( \mathcal{VB}_{\mathfrak{S}_{\bullet}}(M^{\bullet}) \big)$ endowing the latter with a symmetric monoidal structure.
\end{itemize}
\end{Example}

\subsubsection{Equivariant Cauchy Tensor Algebra Bundles} \label{3.2.1}

\begin{Definition} 
The \textbf{equivariant Cauchy tensor algebra bundle} generated by $\mathbf{V} \in \mathcal{VB}_{\mathfrak{S}_{\bullet}}(M^{\bullet})$ is defined by 
\[ \mathbf{T}^{\boxtimes}(\mathbf{V}) = \bigoplus \limits_{n \in \mathbb{N}} \mathbf{V}^{\boxtimes n} \qquad \text{with} \qquad \mathbf{V}^{\boxtimes 0} = \mathbf{I}_{\boxtimes}. \]
\end{Definition}

There is an obvious injection map $\mathbf{V} \hookrightarrow \mathbf{T}^{\boxtimes}(\mathbf{V})$, and one has 
\[ \big[ \mathbf{T}^{\boxtimes}(\mathbf{V}) \big]_k = \bigoplus \limits_{n \in \mathbb{N}} (\mathbf{V}^{\boxtimes n})_k \qquad \text{where} \qquad (\mathbf{V}^{\boxtimes n})_k =\bigoplus \limits_{p_1 + \ldots + p_n = k} \mathrm{Ind}_{\mathfrak{S}_{p_1} \times \ldots \times \mathfrak{S}_{p_n}}^{\mathfrak{S}_k} \Big( {V}_{p_1} \boxtimes^{\mathrm{ext}} \ldots \boxtimes^{\mathrm{ext}} {V}_{p_n} \Big). \]
For instance, over a point $(x,y) \in M^2$, the fibre of $\mathbf{V}^{\boxtimes 3}$ is 
\[ \begin{array}{lll} (\mathbf{V}^{\boxtimes 3})_{xy} &=& (V_{\emptyset} \otimes V_x \otimes V_y) \oplus (V_{\emptyset} \otimes V_y \otimes V_x) \oplus (V_{\emptyset} \otimes V_{xy} \otimes V_{\emptyset}) \oplus (V_{\emptyset} \otimes V_{\emptyset} \otimes V_{xy}) \\ 
&& \oplus (V_x \otimes V_{\emptyset} \otimes V_y) \oplus (V_y \otimes V_{\emptyset} \otimes V_x) \oplus (V_x \otimes V_y \otimes V_{\emptyset}) \oplus (V_y \otimes V_x \otimes V_{\emptyset}) \\ 
&& \oplus (V_{xy} \otimes V_{\emptyset} \otimes V_{\emptyset}). \end{array} \]
In general, $(\mathbf{V}^{\boxtimes n})_k$ contains $n^k$ terms.

The multiplication on $\mathbf{T}^{\boxtimes}(\mathbf{V})$ is described as follows. Fix base configurations $(x_1, \ldots, x_p) \in M^p$ and $(y_1, \ldots, y_q) \in M^q$. An element of $\mathbf{T}^{\boxtimes}(\mathbf{V})_{(x_1, \ldots, x_p)}$ has the form $v = v_1 \otimes \ldots \otimes v_r$, where each $v_i$ lives over some subtuple of $(x_1, \ldots, x_p)$ of length $n_i$. Similarly, an element of $\mathbf{T}^{\boxtimes}(\mathbf{V})_{(y_1, \ldots, y_q)}$ has the form $w = w_1 \otimes \ldots \otimes w_s$, where each $w_j$ lives over some subtuple of $(y_1, \ldots, y_q)$ of length $m_j$. Then,
\[ m_{p,q}^{\boxtimes}(v,w) = v_1 \otimes \ldots \otimes v_r \otimes w_1 \otimes \ldots \otimes w_s \]
lying above the concatenated configuration $(x_1, \ldots, x_p, y_1, \ldots, y_q)$ in the block pattern $(n_1 | \ldots | n_r | m_1 | \ldots | m_s)$. Thus, the $\boxtimes$-multiplication of $\mathbf{T}^{\boxtimes}(\mathbf{V})$ concatenates both vectors and base points. As before, we denote $m^{\boxtimes}(v,w)$ simply by $v \boxtimes w$.

\begin{Theorem} 
The equivariant Cauchy tensor algebra $\mathbf{T}^{\boxtimes}(\mathbf{V})$ is the free equivariant $\boxtimes$-algebra generated by $\mathbf{V}$. More precisely, for every equivariant $\boxtimes$-algebra bundle $\mathbf{A} \in \mathbf{Alg}^{\boxtimes} \big( \mathcal{VB}_{\mathfrak{S}_{\bullet}}(M^{\bullet}) \big)$ and every equivariant bundle map $f: \mathbf{V} \rightarrow \mathbf{A}$, there exists a unique equivariant $\boxtimes$-algebra bundle morphism $\widetilde{f}: \mathbf{T}^{\boxtimes}(\mathbf{V}) \longrightarrow \mathbf{A}$ such that the following diagram commutes 
\[
\xymatrix{
\mathbf{V} \ar[rr]^{\iota} \ar[rrd]_{f} && \mathbf{T}^{\boxtimes}(\mathbf{V}) \ar@{.>}[d]^{\exists ! \widetilde{f}} \\ 
&& \mathbf{A}
}
\]
\end{Theorem}

\begin{proof} 
This is the standard construction of free monoids in a symmetric monoidal category described in \cite{saunders}.
\end{proof}

The functor $\mathbf{T}^{\boxtimes}$ is canonically oplax monoidal with respect to $\otimes$. More precisely, it is an oplax monoidal functor 
\[ \mathbf{T}^{\boxtimes}: \big( \mathcal{VB}_{\mathfrak{S}_{\bullet}}(M^{\bullet}), \otimes, \mathbf{I}_{\otimes} \big) \longrightarrow \Big( \mathbf{Alg}^{\boxtimes} \big( \mathcal{VB}_{\mathfrak{S}_{\bullet}}(M^{\bullet}) \big), \otimes, \mathbf{I}_{\otimes} \Big). \]
This means that there are natural morphisms
\[ \mathbf{T}^{\boxtimes}(\mathbf{V} \otimes \mathbf{W}) \longrightarrow \mathbf{T}^{\boxtimes}(\mathbf{V}) \otimes \mathbf{T}^{\boxtimes}(\mathbf{W}) \qquad \text{and} \qquad \mathbf{T}^{\boxtimes}(\mathbf{I}_{\otimes}) \longrightarrow \mathbf{I}_{\otimes}. \]
Indeed, the object $\mathbf{T}^{\boxtimes}(\mathbf{V}) \otimes \mathbf{T}^{\boxtimes}(\mathbf{W})$ is an equivariant $\boxtimes$-algebra, and there is a canonical bundle map given by the inclusions
\[ \mathbf{V} \otimes \mathbf{W} \longrightarrow \mathbf{T}^{\boxtimes}(\mathbf{V}) \otimes \mathbf{T}^{\boxtimes}(\mathbf{W}). \]
Since $\mathbf{T}^{\boxtimes}(\mathbf{V} \otimes \mathbf{W})$ is the free equivariant $\boxtimes$-algebra generated by $\mathbf{V} \otimes \mathbf{W}$, this map extends uniquely to a $\boxtimes$-algebra morphism 
\[ \mathbf{T}^{\boxtimes}(\mathbf{V} \otimes \mathbf{W}) \longrightarrow \mathbf{T}^{\boxtimes}(\mathbf{V}) \otimes \mathbf{T}^{\boxtimes}(\mathbf{W}). \]
This morphism is generally not an isomorphism, even for local vector bundles.

When $\mathbf{V}$ is local, the description simplifies considerably. If $\mathbf{V}$ is a local vector bundle, so that $V_k = 0$ unless $k=1$, and we write $V := V_1$, thereby identifying $\mathbf{V}$ with $V$, we therefore use the non-bold notation $V$ for local vector bundles. Then $\mathbf{T}^{\boxtimes}(V)_k$ is non-zero only for the $n=k$-summand. That is,
\[ \big[ \mathbf{T}^{\boxtimes}(V) \big]_k = (V^{\boxtimes k})_k = \mathrm{Ind}_{\mathfrak{S}_1 \times \ldots \times \mathfrak{S}_1}^{\mathfrak{S}_k} \big( V \boxtimes^{\mathrm{ext}} \ldots \boxtimes^{\mathrm{ext}} V \big) = \bigoplus \limits_{\sigma \in \mathfrak{S}_k} \sigma^* \Big( V^{\boxtimes^{\mathrm{ext}} k} \Big). \]
In particular, the fibre over a point $(x_1, \ldots, x_k) \in M^k$ is 
\[ \Big[ \mathbf{T}^{\boxtimes}(V) \Big]_{x_1, \ldots, x_k} = \bigoplus \limits_{\sigma \in \mathfrak{S}_k} V_{x_{\sigma(1)}} \otimes_{\mathbb{K}} \ldots \otimes_{\mathbb{K}} V_{x_{\sigma(k)}}. \]

\subsubsection{Equivariant Cauchy Symmetric Algebra Bundles} \label{3.2.2}

Similarly to the construction of the Hadamard symmetric algebra through coinvariants in Section \ref{3.1.2}, we carry out the same type of construction for the Cauchy symmetric algebra bundle. For a given $\mathbf{V} \in \mathcal{VB}_{\mathfrak{S}_{\bullet}}(M^{\bullet})$, recall that $\big[ \mathbf{T}^{\boxtimes}(\mathbf{V}) \big]_k$ is a graded $\mathfrak{S}_k$-equivariant bundle. As in the Hadamard case, we want to define a $\mathfrak{S}_n$-action on $(\mathbf{V}^{\boxtimes n})_k$ that permutes the factors. This is achieved via the following combinatorial construction. 

To this end, let us divide $\llbracket 1, k \rrbracket$ into $n$ packs of lengths $i_1, \ldots, i_n$. One may think of this as partitioning a deck of $k$ cards into $n$ ordered piles of cards. 
\begin{center}
\begin{tikzpicture}[every node/.style={font=\normalsize},baseline]
\node (1) at (0,0) {$1$};
\node (dots1) at (1,0) {$\cdots$};
\node (i1) at (2,0) {$i_1$};

\node (bar1) at (3.5,0) {$\;\Big|\;$};

\node (2) at (5,0) {$i_1 + 1$};
\node (dots2) at (6.5,0) {$\cdots$};
\node (i2) at (8,0) {$i_1 + i_2$}; 

\node (bar2) at (10,0) {$\;\Big|\;$};

\node (dots3) at (11,0) {$\cdots$};

\node (bar3) at (12,0) {$\;\Big|\;$}; 

\node (ik-1) at (14,0) {$i_1 + i_2 + \ldots + i_{n-1} + 1$};
\node (bark) at (17,0) {$\cdots$};
\node (ik) at (18,0) {$k$};

% Cercles au-dessus
\node[circle,draw,minimum size=3mm] at (1,1) {$1$};
\node[circle,draw,minimum size=3mm] at (6.6,1) {$2$};
\node[circle,draw,minimum size=3mm] at (11,1) {$\cdots$};
\node[circle,draw,minimum size=3mm] at (17,1) {$n$};
\end{tikzpicture}
\end{center}
For a given permutation $\sigma \in \mathfrak{S}_n$, we would like to consider a permutation of the piles without changing the order within each pile. This leads to the permutation
\begin{center} 
\begin{tikzpicture}[every node/.style={font=\normalsize},baseline]
\node (tau) at (0,0) {$\tau_{\sigma}^{(i_1, \ldots, i_n)}$};
\node (equal) at (1,0) {$=$};
\node[circle,draw,minimum size = 3mm] at (2,0) {$\sigma(1)$};
\node (bar) at (3,0) {$\; \Big| \;$};
\node[circle,draw,minimum size = 3mm] at (4,0) {$\sigma(2)$};
\node (bar2) at (5,0) {$\; \Big| \;$};
\node (dots) at (6,0) {$\cdots$};
\node (bar3) at (7,0) {$\; \Big| \;$};
\node[circle,draw,minimum size = 3mm] at (8,0) {$\sigma(n)$};
\end{tikzpicture}
\end{center} 
These block permutations also appear in the standard description of the symmetric group action on tensor products in colored operad theory; see \cite[Chapter 11]{yau} for details. Some examples, where we denote $\sigma = \begin{bmatrix} \sigma(1) & \ldots& \sigma(n) \end{bmatrix}$ in matrix notation, are as follows:
\[ \begin{array}{lll} \sigma = \begin{bmatrix} 3 & 2 & 4 & 1 \end{bmatrix} \in \mathfrak{S}_4, &\qquad& \tau_{\sigma}^{(1,1,2,2)} = \left[ \begin{array}{cc|c|cc|c} 3 & 4 & 2 & 5 & 6 & 1 \end{array} \right] \in \mathfrak{S}_6; \\
\sigma = \begin{bmatrix} 2 & 3 & 1 \end{bmatrix} \in \mathfrak{S}_3, & \qquad & \tau_{\sigma}^{(1,2,3)} = \left[ \begin{array}{cc|ccc|c} 2& 3 & 4 & 5 & 6 & 1 \end{array} \right] \in \mathfrak{S}_6; \\ 
\sigma = \begin{bmatrix} 2&3&1 \end{bmatrix} \in \mathfrak{S}_3, & \qquad & \tau_{\sigma}^{(2,3,1)} = \left[ \begin{array}{ccc|c|cc} 3&4&5 & 6 & 1 & 2 \end{array} \right] \in \mathfrak{S}_6; \\ 
\sigma = \begin{bmatrix} 2&1 \end{bmatrix} \in \mathfrak{S}_2, &\qquad& \tau_{\sigma}^{(4,2)} = \left[ \begin{array}{cc|cccc} 5 & 6 & 1 & 2 & 3 & 4 \end{array} \right] \in \mathfrak{S}_6. \end{array} \]
In the extreme case where the $k$ cards are divided into $2$ decks $(p,q)$, we recover the earlier permutation $\tau_{p,q}$ used for the $\boxtimes$-braiding. At the opposite extreme, when the $k$ cards are divided into exactly $k$ decks, so that necessarily $(i_1, \ldots, i_k) = (1, \ldots, 1)$, one has $\tau_{\sigma}^{(1, \ldots, 1)} = \sigma$ for all $\sigma \in \mathfrak{S}_k$.

Given indices $i_1, \ldots, i_n$ such that $i_1 + \ldots + i_n = k$, we denote by $\mathrm{Sh}(i_1, \ldots, i_n)$ the set of $(i_1, \ldots, i_n)$-shuffles, that is, permutations $\alpha \in \mathfrak{S}_k$ that preserve the relative order within each consecutive block $\left\{ i_1 + \ldots + i_{r-1} + 1, \ldots, i_1 + \ldots + i_r \right\}$. The set of $(i_1, \ldots, i_n)$-shuffles is a set of full representatives of the left quotient $\mathfrak{S}_K/ \big( \mathfrak{S}_{i_1} \times \ldots \times \mathfrak{S}_{i_n} \big)$.

\begin{Proposition} 
Given indices $i_1, \ldots, i_n$ such that $i_1 + \ldots + i_n = k$ and a permutation $\sigma \in \mathfrak{S}_n$, the map 
\[ \alpha \in \mathrm{Sh}(i_1, \ldots, i_n) \longmapsto \alpha \circ \tau_{\sigma}^{(i_1, \ldots, i_n)} \in \mathrm{Sh} \big( i_{\sigma(1)}, \ldots, i_{\sigma(n)} \big) \]
is a bijection of sets. Moreover, we have the relation 
\[ \tau_{\sigma}^{(i_1, \ldots, i_n)} \circ \tau_{\sigma'}^{(i_{\sigma(1)}, \ldots, i_{\sigma(n)})} = \tau_{\sigma \circ \sigma'}^{(i_1, \ldots, i_n)}. \]
In particular, we have 
\[ \big( \tau_{\sigma}^{(i_1, \ldots, i_n)} \big)^{-1} = \tau_{\sigma^{-1}}^{(i_{\sigma(1)}, \ldots, i_{\sigma(n)})}. \]
\end{Proposition}

\begin{proof} 
Let us partition $\llbracket 1, k \rrbracket$ into consecutive blocks 
\[ B_1 = \left\{1, \ldots, i_1 \right\}, \qquad B_2 = \left\{ i_1 + 1, \ldots, i_1 + i_2 \right\}, \qquad \ldots, \qquad B_n = \left\{ i_1 + \ldots + i_{n-1} + 1, \ldots, k \right\}. \]
Each block $B_r$ has length $i_r$ and an element of $\mathrm{Sh}(i_1, \ldots, i_n)$ is precisely a permutation of $\mathfrak{S}_k$ that is increasing on each block $B_r$. Let $\alpha \in \mathrm{Sh}(i_1, \ldots, i_n)$ and $\sigma \in \mathfrak{S}_n$. By construction, $\tau_{\sigma}^{(i_1, \ldots, i_n)}$ reorders the blocks according to $\sigma$ while preserving the order inside each block: it moves the block $B_r$ to position $\sigma(r)$ while preserving the order inside the block. It follows that $\alpha \circ \tau_{\sigma}^{(i_1, \ldots, i_n)}$ is increasing on the reordered blocks, hence belongs to $\mathrm{Sh} \big( i_{\sigma(1)}, \ldots, i_{\sigma(n)} \big)$.

Now let $\sigma, \sigma' \in \mathfrak{S}_n$. The permutation $\tau_{\sigma'}^{\big( i_{\sigma(1)}, \ldots, i_{\sigma(n)} \big)}$ first reorders the blocks according to $\sigma'$, and then $\tau_{\sigma}^{(i_1, \ldots, i_n)}$ reorders them according to $\sigma$. Therefore the overall effect is to reorder the blocks according to $\sigma \circ \sigma'$. Since the order inside each block is preserved at every step, the resulting permutation is exactly the deck permutation $\tau_{\sigma \circ \sigma'}^{(i_1, \ldots, i_n)}$. Taking $\sigma' = \sigma^{-1}$ yields the inverse formula, and bijectivity is immediate.
\end{proof}

Since permuting the $n$ factors sends the decomposition indexed by $(i_1, \ldots, i_n)$ to the decomposition indexed by $\big( i_{\sigma(1)}, \ldots, i_{\sigma(n)} \big)$, the permutation must be accompanied by a permutation of the underlying $k$ points. This is precisely achieved by the deck permutation $\tau_{\sigma}^{(i_1, \ldots, i_n)}$. Accordingly, we define an action of $\mathfrak{S}_n$ on $(\mathbf{V}^{\boxtimes n})_k = \bigoplus \limits_{i_1 + \ldots + i_n = k} \mathrm{Ind}_{\mathfrak{S}_{i_1} \times \ldots \times \mathfrak{S}_{i_n}}^{\mathfrak{S}_k} \Big( V_{i_1} \boxtimes^{\mathrm{ext}} \ldots \boxtimes^{\mathrm{ext}} V_{i_n} \Big)$ by sending each summand indexed by $(i_1, \ldots, i_n)$ to the summand indexed by $\big( i_{\sigma(1)}, \ldots, i_{\sigma(n)} \big)$ through the map
\[ \tau_{\sigma}^{(i_1, \ldots, i_n)}: \mathrm{Ind}_{\mathfrak{S}_{i_1} \times \ldots \times \mathfrak{S}_{i_n}}^{\mathfrak{S}_k} \Big( V_{i_1} \boxtimes^{\mathrm{ext}} \ldots \boxtimes^{\mathrm{ext}} V_{i_n} \Big) \longrightarrow \mathrm{Ind}_{\mathfrak{S}_{i_{\sigma(1)}} \times \ldots \times \mathfrak{S}_{i_{\sigma(n)}}}^{\mathfrak{S}_k} \Big( V_{i_{\sigma(1)}} \boxtimes^{\mathrm{ext}} \ldots \boxtimes^{\mathrm{ext}} V_{i_{\sigma(n)}} \Big). \]
These maps glue together to a morphism
\[ \sigma: (\mathbf{V}^{\boxtimes n})_k \longrightarrow (\mathbf{V}^{\boxtimes n})_k \]
and the previous proposition shows that this indeed defines a $\mathfrak{S}_n$-action on $(\mathbf{V}^{\boxtimes n})_k$. Since $\mathfrak{S}_k$ permutes the points of the configuration through the induced equivariant structure, whereas $\mathfrak{S}_n$ permutes the labels of the factors, we obtain the following.

\begin{Lemma} 
The $\mathfrak{S}_n$ and $\mathfrak{S}_k$ actions on $(\mathbf{V}^{\boxtimes n})_k$ commute.
\end{Lemma}

Using this lemma, the subbundle ${W}_{k,n}^{\boxtimes}$ is $\mathfrak{S}_k$-equivariant and therefore the quotient 
\[ \big[ \mathbf{S}^{\boxtimes n}(\mathbf{V}) \big]_k := (\mathbf{V}^{\boxtimes n})_k \big/ W_{k,n}^{\boxtimes} \] 
is well-defined and $\mathfrak{S}_k$-equivariant. 

\begin{Definition} 
The \textbf{equivariant Cauchy symmetric algebra bundle} generated by $\mathbf{V} \in \mathcal{VB}_{\mathfrak{S}_{\bullet}}(M^{\bullet})$ is defined by 
\[ \mathbf{S}^{\boxtimes}(\mathbf{V}) = \bigoplus \limits_{n \in \mathbb{N}} \mathbf{S}^{\boxtimes n}(\mathbf{V}) \qquad \text{where} \qquad \big[ \mathbf{S}^{\boxtimes n}(\mathbf{V}) \big]_k = (\mathbf{V}^{\boxtimes n})_k \big/ W_{k,n}^{\boxtimes}. \]
\end{Definition}

As in the Hadamard case, the concatenation of tensors $\mathbf{V}^{\boxtimes p} \times \mathbf{V}^{\boxtimes q} \rightarrow \mathbf{V}^{\boxtimes (p+q)}$ induces a map 
\[ \mathbf{S}^{\boxtimes p}(\mathbf{V}) \times \mathbf{S}^{\boxtimes q}(\mathbf{V}) \longrightarrow \mathbf{S}^{\boxtimes (p+q)}(\mathbf{V}) \]
which is the multiplication map on $\mathbf{S}^{\boxtimes}(\mathbf{V})$. If $[v] \in \mathbf{S}^{\boxtimes p}(\mathbf{V})$ and $[w] \in \mathbf{S}^{\boxtimes q}(\mathbf{V})$, then we denote by $[v] \boxdot [w]$ the equivalence class of $[v \boxtimes w]$. This multiplication $\boxdot$ is commutative.

\begin{Theorem} 
The equivariant Cauchy symmetric algebra bundle $\mathbf{S}^{\boxtimes}(\mathbf{V})$ is the free equivariant commutative $\boxtimes$-algebra generated by $\mathbf{V} \in \mathcal{VB}_{\mathfrak{S}_{\bullet}}(M^{\bullet})$. More precisely, for every equivariant commutative $\boxtimes$-algebra bundle $\mathbf{A} \in \mathbf{CAlg}^{\boxtimes} \big( \mathcal{VB}_{\mathfrak{S}_{\bullet}}(M^{\bullet}) \big)$ and every equivariant bundle map $f: \mathbf{V} \rightarrow \mathbf{A}$, there exists a unique $\boxtimes$-algebra bundle morphism $\widetilde{f}: \mathbf{S}^{\boxtimes}(\mathbf{V}) \longrightarrow \mathbf{A}$ such that the following diagram commutes 
\[
\xymatrix{
\mathbf{V} \ar[rr]^{\iota} \ar[rrd]_{f} && \mathbf{S}^{\boxtimes}(\mathbf{V}) \ar@{.>}[d]^{\exists ! \widetilde{f}} \\ 
&& \mathbf{A}
}
\]
\end{Theorem}

\begin{proof} 
Let $\mathbf{A}$ be an equivariant commutative $\boxtimes$-algebra with multiplication $\boxtimes_{\mathbf{A}}$ and let $f: \mathbf{V} \rightarrow \mathbf{A}$ be an equivariant bundle morphism. Since $\mathbf{T}^{\boxtimes}(\mathbf{V})$ is the free $\boxtimes$-algebra, $f$ extends to a unique $\boxtimes$-algebra map 
\[ F: \mathbf{T}^{\boxtimes}(\mathbf{V}) \longrightarrow \mathbf{A}. \]
For each $n \in \mathbb{N}$, denote by $F_n$ the restriction $F|_{\mathbf{V}^{\boxtimes n}}$. We claim that each map $F_n$ is $\mathfrak{S}_n$-invariant with respect to the permutation of factors. Indeed, on a pure tensor $v_1 \boxtimes \ldots \boxtimes v_n \in \mathbf{V}^{\boxtimes n}$, one has
\[ F_n(v_1 \boxtimes \ldots \boxtimes v_n) = f(v_1) \boxtimes_{\mathbf{A}} \ldots \boxtimes_{\mathbf{A}} f(v_n). \]  
If $\sigma \in \mathfrak{S}_n$, then $\sigma \cdot (v_1 \boxtimes \ldots \boxtimes v_n)$ is obtained by permuting the $n$ tensor blocks and simultaneously transporting the underlying configuration by the associated deck permutation $\tau_{\sigma}^{(i_1, \ldots, i_n)}$. Since $F$ is equivariant and multiplicative, we obtain 
\[ F_n \Big( \sigma \cdot (v_1 \boxtimes \ldots \boxtimes v_n) \Big) = f \big( v_{\sigma(1)} \big) \boxtimes_{\mathbf{A}} \ldots \boxtimes_{\mathbf{A}} f \big( v_{\sigma(n)} \big). \]
Since $\mathbf{A}$ is commutative, the product is invariant under permutations of the factors, and therefore
\[ F_n \Big( \sigma \cdot (v_1 \boxtimes \ldots \boxtimes v_n) \Big) = f(v_1) \boxtimes_{\mathbf{A}} \ldots \boxtimes_{\mathbf{A}} f(v_n) = F_n(v_1 \boxtimes \ldots \boxtimes v_n). \]
It follows that $F_n$ vanishes on the subbundle $W_{n}^{\boxtimes}$ defined above (for each $k$, the degree-$k$ component of $F_n$ vanishes on $W_{k,n}^{\boxtimes}$), and hence induces a well-defined equivariant bundle morphism 
\[ f_n: \mathbf{S}^{\boxtimes n}(\mathbf{V}) \longrightarrow \mathbf{A}. \]
Taking the direct sum over all $n \in \mathbb{N}$ yields the equivariant bundle morphism 
\[ \widetilde{f}: \mathbf{S}^{\boxtimes}(\mathbf{V}) \longrightarrow \mathbf{A}. \]
We now show that $\widetilde{f}$ is a $\boxtimes$-algebra morphism. Let $[u] \in \mathbf{S}^{\boxtimes p}(\mathbf{V})$ and $[v] \in \mathbf{S}^{\boxtimes q}(\mathbf{V})$. Since $[u] \boxdot [v] = [u \boxtimes v]$, one has
\[ \widetilde{f} \Big( [u] \boxdot [v] \Big) = \widetilde{f} \Big( [u \boxtimes v] \Big) = F(u \boxtimes v) = F(u) \boxtimes_{\mathbf{A}} F(v) = \widetilde{f} \big( [u] \big) \boxtimes_{\mathbf{A}} \widetilde{f} \big( [v] \big). \]
Therefore $\widetilde{f}: \mathbf{S}^{\boxtimes}(\mathbf{V}) \rightarrow \mathbf{A}$ is multiplicative. Moreover, since $F: \mathbf{T}^{\boxtimes}(\mathbf{V}) \rightarrow \mathbf{A}$ is a $\boxtimes$-algebra map, it sends the unit $\mathbf{I}_{\boxtimes}$ to the unit of $\mathbf{A}$. Since the group action on $\mathbf{I}_{\boxtimes}$ is trivial, the unit of $\mathbf{S}^{\boxtimes}(\mathbf{V})$ is also $\mathbf{I}_{\boxtimes}$, and $\widetilde{f}$ sends it to the unit of $\mathbf{A}$.

Finally, let $g: \mathbf{S}^{\boxtimes}(\mathbf{V}) \rightarrow \mathbf{A}$ be a $\boxtimes$-algebra morphism extending $f$. Then the composition 
\[ \mathbf{T}^{\boxtimes}(\mathbf{V}) \longrightarrow \mathbf{S}^{\boxtimes}(\mathbf{V}) \overset{g} \longrightarrow \mathbf{A} \]
is a $\boxtimes$-algebra map extending $f: \mathbf{V} \rightarrow \mathbf{A}$. By the universal property of $\mathbf{T}^{\boxtimes}(\mathbf{V})$, this composition must be $F$. Since the projection $\mathbf{T}^{\boxtimes}(\mathbf{V}) \longrightarrow \mathbf{S}^{\boxtimes}(\mathbf{V})$ is surjective, it follows that $g = \widetilde{f}$.
\end{proof}

Similarly to $\mathbf{T}^{\boxtimes}$, the functor $\mathbf{S}^{\boxtimes}$ is canonically oplax monoidal with respect to $\otimes$, so that there are natural morphisms 
\[ \mathbf{S}^{\boxtimes}(\mathbf{V} \otimes \mathbf{W}) \longrightarrow \mathbf{S}^{\boxtimes}(\mathbf{V}) \otimes \mathbf{S}^{\boxtimes}(\mathbf{W}) \qquad \text{and} \qquad \mathbf{S}^{\boxtimes}(\mathbf{I}_{\otimes}) \longrightarrow \mathbf{I}_{\otimes}. \]
In general, these morphisms are not isomorphisms. However, they become isomorphisms in the special case of local vector bundles. If $V$ is a local vector bundle, the combinatorics simplifies drastically, and one has, as vector bundles,
\[ \big[ \mathbf{S}^{\boxtimes}(V) \big]_k \cong V^{\boxtimes^{\mathrm{ext}} k} = V \boxtimes^{\mathrm{ext}} \ldots \boxtimes^{\mathrm{ext}} V. \]

\begin{Proposition}
The functor $\mathbf{S}^{\boxtimes}$ is strongly monoidal with respect to $\otimes$ on the subcategory of local vector bundles, that is, there are isomorphisms of commutative $\boxtimes$-algebras 
\[ \mathbf{S}^{\boxtimes}(V \otimes W) \cong \mathbf{S}^{\boxtimes}(V) \otimes \mathbf{S}^{\boxtimes}(W) \qquad \text{and} \qquad \mathbf{S}^{\boxtimes}(\mathbf{I}_{\otimes}) \cong \mathbf{I}_{\otimes} \]
when $V$ and $W$ are both local vector bundles.
\end{Proposition}

\begin{proof} 
We have the isomorphism $\mathbf{S}^{\boxtimes}(V)_n \cong V^{\boxtimes^{\mathrm{ext}} n}$, as vector bundles and we use the interchange law between $\boxtimes^{\mathrm{ext}}$ and $\otimes$: 
\[ (A \otimes B) \boxtimes^{\mathrm{ext}} (C \otimes D) \cong (A \boxtimes^{\mathrm{ext}} C) \otimes (B \boxtimes^{\mathrm{ext}} D). \]
Therefore, 
\[ \mathbf{S}^{\boxtimes}(V \otimes W)_n \cong (V \otimes W)^{\boxtimes^{\mathrm{ext}}n} \cong V^{\boxtimes^{\mathrm{ext}}n} \otimes W^{\boxtimes^{\mathrm{ext}}n} \cong \mathbf{S}^{\boxtimes}(V)_n \otimes \mathbf{S}^{\boxtimes}(W)_n. \]
Moreover, since $\mathbf{I}_{\otimes}$ is already a commutative $\boxtimes$-algebra, we have $\mathbf{S}^{\boxtimes}(\mathbf{I}_{\otimes}) \cong \mathbf{I}_{\otimes}$.
\end{proof}

The Cauchy tensor product recovers the Hadamard tensor product along the diagonals in the sense that:
\begin{Proposition} 
Denote by $\Delta_k: x \in M \mapsto (x,\ldots,x) \in M^k$ the diagonal inclusion. If $V \rightarrow M$ is a local vector bundle, then there is a canonical isomorphism
\[ \Delta_k^* \Big( \mathbf{S}^{\boxtimes}(V)_k \Big) \cong V^{\otimes k} \]
as vector bundles over $M$.
\end{Proposition}

\begin{proof} 
Since $V$ is local, we have an isomorphism 
\[ \mathbf{S}^{\boxtimes}(V)_k \cong V^{\boxtimes^{\mathrm{ext}} k} \] 
as vector bundles over $M^k$. Pulling back along $\Delta_k$, one gets 
\[ \Delta_k^* \Big( \mathbf{S}^{\boxtimes}(V)_k \Big) \cong \Delta_k^* \big( V^{\boxtimes^{\mathrm{ext}}k} \big). \]
But restriction of an external tensor product to the full diagonal is canonically the internal tensor product, hence 
\[ \Delta_k^* \Big( \mathbf{S}^{\boxtimes}(V)_k \Big) \cong V^{\otimes k}. \]
\end{proof}

\begin{Remark} 
Similarly, by introducing a sign twist in the action, one constructs the \textbf{Cauchy exterior algebra bundle} $\mathbf{\Lambda}^{\boxtimes}(\mathbf{V})$. As expected, $\mathbf{\Lambda}^{\boxtimes}(\mathbf{V})$ is the free equivariant graded-commutative $\boxtimes$-algebra generated by $\mathbf{V}$. Likewise, there are natural morphisms
\[ \mathbf{\Lambda}^{\boxtimes}(\mathbf{V} \otimes \mathbf{W}) \longrightarrow \mathbf{\Lambda}^{\boxtimes}(\mathbf{V}) \otimes \mathbf{\Lambda}^{\boxtimes}(\mathbf{W}) \qquad \text{and} \qquad \mathbf{\Lambda}^{\boxtimes}(\mathbf{I}_{\otimes}) \longrightarrow \mathbf{I}_{\otimes} \]
and these are isomorphisms when $\mathbf{V}$ and $\mathbf{W}$ are both local vector bundles.
\end{Remark}

\subsection{Equivariant Cauchy-Hadamard $2$-Algebra Bundles over Configuration Spaces} \label{3.3}

\begin{Definition}
An \textbf{equivariant $2$-algebra bundle} is a double monoid object in the symmetric $2$-monoidal category $\Big( \mathcal{VB}_{\mathfrak{S}_{\bullet}}(M^{\bullet}), \boxtimes, \mathbf{I}_{\boxtimes}, \otimes, \mathbf{I}_{\otimes} \Big)$.
\end{Definition}

Equivalently, an equivariant $2$-algebra bundle is a $\mathfrak{S}_{\bullet}$-equivariant vector bundle $\mathbf{A}$ together with four maps 
\[ \begin{array}{lll} m^{\otimes}: \mathbf{A} \otimes \mathbf{A} \longrightarrow \mathbf{A} &\qquad& m^{\boxtimes}: \mathbf{A} \boxtimes \mathbf{A} \longrightarrow \mathbf{A} \\
u^{\otimes}: \mathbf{I}_{\otimes} \longrightarrow \mathbf{A} &\qquad& u^{\boxtimes}: \mathbf{I}_{\boxtimes} \longrightarrow \mathbf{A} \end{array} \]
such that the following conditions hold: \begin{enumerate}
    \item $\mathbf{A}$ is an equivariant $\otimes$-algebra, i.e. $\big( \mathbf{A}, m^{\otimes}, u^{\otimes} \big)$ is a monoid object in $\big( \mathcal{VB}_{\mathfrak{S}_{\bullet}}(M^{\bullet}), \otimes \big)$.
    \item $\mathbf{A}$ is an equivariant $\boxtimes$-algebra, i.e. $\big( \mathbf{A}, m^{\boxtimes}, u^{\boxtimes} \big)$ is a monoid object in $\big( \mathcal{VB}_{\mathfrak{S}_{\bullet}}(M^{\bullet}), \boxtimes \big)$.
    \item The interchange law for $m^{\otimes}$ and $m^{\boxtimes}$ holds, i.e. the following diagram commutes 
    \[
\xymatrix{
(\mathbf{A} \otimes \mathbf{A}) \boxtimes (\mathbf{A} \otimes \mathbf{A}) \ar[d]_{m^{\otimes} \boxtimes m^{\otimes}} \ar[rr]^{\zeta_{\mathbf{A}, \mathbf{A}, \mathbf{A}, \mathbf{A}}} && (\mathbf{A} \boxtimes \mathbf{A}) \otimes (\mathbf{A} \boxtimes \mathbf{A}) \ar[d]^{m^{\boxtimes} \otimes m^{\boxtimes}} \\ 
\mathbf{A} \boxtimes \mathbf{A} \ar[rd]_{m^{\boxtimes}} && \mathbf{A} \otimes \mathbf{A} \ar[ld]^{m^{\otimes}} \\ 
& \mathbf{A} &   
}
\]
where $\zeta$ is the interchange map of the symmetric $2$-monoidal category $\big( \mathcal{VB}_{\mathfrak{S}_{\bullet}}(M^{\bullet}), \boxtimes, \otimes \big)$.
\item The units $u^{\otimes}$ and $u^{\boxtimes}$ are compatible, i.e. the three diagrams commute: 
    \[
\xymatrix{
\mathbf{I}_{\otimes} \boxtimes \mathbf{I}_{\otimes} \ar[d]_{\mu} \ar[rr]^{u^{\otimes} \boxtimes u^{\otimes}} && \mathbf{A} \boxtimes \mathbf{A} \ar[d]^{m^{\boxtimes}} \\ 
\mathbf{I}_{\otimes} \ar[rr]_{u^{\otimes}} && \mathbf{A}}
\qquad 
\xymatrix{
\mathbf{I}_{\boxtimes} \otimes \mathbf{I}_{\boxtimes}  \ar[rr]^{u^{\boxtimes} \otimes u^{\boxtimes}} && \mathbf{A} \otimes \mathbf{A} \ar[d]^{m^{\otimes}} \\ 
\mathbf{I}_{\boxtimes} \ar[rr]_{u^{\boxtimes}} \ar[u]^{\Delta} && \mathbf{A}}
\qquad
\xymatrix{
\mathbf{I}_{\boxtimes} \ar[rd]_{u^{\boxtimes}} \ar[rr]^{\nu} && \mathbf{I}_{\otimes} \ar[ld]^{u^{\otimes}} \\ 
&\mathbf{A}&
}
\]
where $\mu$, $\Delta$ and $\nu$ are the three structure morphisms of the symmetric $2$-monoidal category $\big( \mathcal{VB}_{\mathfrak{S}_{\bullet}}(M^{\bullet}), \boxtimes, \otimes \big)$.
\end{enumerate} 

We say that $\mathbf{A}$ is commutative if both equivariant algebras $(\mathbf{A}, m^{\otimes}, u^{\otimes})$ and $(\mathbf{A}, m^{\boxtimes}, u^{\boxtimes})$ are commutative, i.e. $m^{\otimes} \circ \beta^{\otimes} = m^{\otimes}$ and $m^{\boxtimes} \circ \beta^{\boxtimes} = m^{\boxtimes}$, where $\beta^{\otimes}$ and $\beta^{\boxtimes}$ are the respective braidings of the monoidal structures $\otimes$ and $\boxtimes$. The graded-commutativity is defined in a similar way using the graded braidings.

A morphism between equivariant $2$-algebra bundles is an equivariant bundle map that is both a $\otimes$-algebra morphism and a $\boxtimes$-algebra morphism. 

We will denote by \begin{itemize}
\item[$\bullet$] $\mathbf{Alg}^{\boxtimes, \otimes} \big( \mathcal{VB}_{\mathfrak{S}_{\bullet}}(M^{\bullet}) \big)$ the category of equivariant $2$-algebra bundles.
\item[$\bullet$] $\mathbf{CAlg}^{\boxtimes, \otimes} \big( \mathcal{VB}_{\mathfrak{S}_{\bullet}}(M^{\bullet}) \big)$ the full subcategory of commutative ones.
\end{itemize}

\begin{Remark}[Reformulation due to Aguiar \& Mahajan] 
As described in the previous section, if $\mathbf{A}$ is an equivariant $\boxtimes$-algebra, then $\mathbf{A} \otimes \mathbf{A}$ is naturally a $\boxtimes$-algebra with $\boxtimes$-multiplication 
\[ (\mathbf{A} \otimes \mathbf{A}) \boxtimes (\mathbf{A} \otimes \mathbf{A}) \overset{\zeta} \longrightarrow (\mathbf{A} \boxtimes \mathbf{A}) \otimes (\mathbf{A} \boxtimes \mathbf{A}) \overset{m_{\mathbf{A}}^{\boxtimes} \otimes m_{\mathbf{A}}^{\boxtimes}} \longrightarrow \mathbf{A} \otimes \mathbf{A}. \]
Hence conditions $3.$ and $4.$ can be reformulated by saying that
\[ m^{\otimes}: \mathbf{A} \otimes \mathbf{A} \longrightarrow \mathbf{A} \qquad \text{and} \qquad u^{\otimes}: \mathbf{I}_{\otimes} \longrightarrow \mathbf{A} \qquad \text{are } \boxtimes\text{-algebra morphisms.} \]
In other words, an equivariant $2$-algebra bundle is a $\otimes$-monoid in the monoidal category $\Big( \mathbf{Alg}^{\boxtimes} \big( \mathcal{VB}_{\mathfrak{S}_{\bullet}}(M^{\bullet}) \big), \otimes \Big)$. More concretely, over $M^p \times M^q$, this means that the diagram 
\[ \xymatrix{ 
(A_p \boxtimes^{\mathrm{ext}} A_q) \otimes (A_p \boxtimes^{\mathrm{ext}} A_q) \ar[rrrr]^{m_p^{\otimes} \boxtimes^{\mathrm{ext}} m_q^{\otimes}} \ar[d]_{m_{p,q}^{\boxtimes} \otimes m_{p,q}^{\boxtimes}} &&&& A_p \boxtimes^{\mathrm{ext}} A_q \ar[d]^{m_{p,q}^{\boxtimes}} \\
A_{p+q} \otimes A_{p+q} \ar[rrrr]_{m_{p+q}^{\otimes}} &&&& A_{p+q} 
} \]
commutes. Elementwise, if $a_1, a_2$ lie over $(x_1, \ldots, x_p) \in M^p$ and $b_1, b_2$ lie over $(y_1, \ldots, y_q) \in M^{q}$, then over the configuration $(x_1, \ldots, x_p, y_1, \ldots, y_q)$, we have 
\[ m_{p+q}^{\otimes} \Big( m_{p,q}^{\boxtimes}(a_1, b_1), m_{p,q}^{\boxtimes}(a_2, b_2) \Big) = m_{p,q}^{\boxtimes} \Big( m_p^{\otimes}(a_1, a_2), m_q^{\otimes}(b_1, b_2) \Big). \]
Equivalently,
\[ m_{p,q}^{\boxtimes}: (A_p, m_p^{\otimes}) \times (A_q, m_q^{\otimes}) \longrightarrow (A_{p+q}, m_{p+q}^{\otimes}) \] 
is a $\otimes$-algebra morphism. For the units, this amounts to the following two concrete conditions: \begin{itemize} 
\item[$\bullet$] the units coincide in degree $0$, i.e. $\mathds{1}_0^{\otimes} = \mathds{1}_0^{\boxtimes}$ in $A_0$.
\item[$\bullet$] the units multiply along $\boxtimes$, i.e. 
\[ \forall p,q \in \mathbb{N}, \qquad m_{p,q}^{\boxtimes} \big( \mathds{1}_p^{\otimes}, \mathds{1}_q^{\otimes} \big) = \mathds{1}_{p+q}^{\otimes}. \]
\end{itemize}
\end{Remark}

For a local vector bundle $V \in \mathcal{VB}(M)$, the equivariant Cauchy-Hadamard bundle of $V$ is $\mathbf{S}^{\boxtimes} \big( S^{\otimes} V \big) \in \mathcal{VB}_{\mathfrak{S}_{\bullet}}(M^{\bullet})$.

\begin{Theorem} 
The equivariant Cauchy-Hadamard bundle $\mathbf{S}^{\boxtimes} \big( S^{\otimes}(V) \big)$ is the free equivariant commutative $2$-algebra bundle generated by $V$.
\end{Theorem}

\begin{proof}
We first construct the $2$-algebra structure on $\mathbf{S}^{\boxtimes} \big( \mathbf{S}^{\otimes}(V) \big)$. The $\boxtimes$-multiplication is the free commutative multiplication of $\mathbf{S}^{\boxtimes}$, and we denote it by 
\[ (a,b) \in \mathbf{S}^{\boxtimes} \big( \mathbf{S}^{\otimes}(V) \big) \times \mathbf{S}^{\boxtimes} \big( \mathbf{S}^{\otimes} V \big) \longmapsto a \boxdot b \in \mathbf{S}^{\boxtimes} \big( \mathbf{S}^{\otimes} V \big). \]
The $\otimes$-multiplication is induced by the $\otimes$-multiplication $\mathbf{S}^{\otimes}(V) \otimes \mathbf{S}^{\otimes}(V) \rightarrow \mathbf{S}^{\otimes}(V)$ and the strong monoidality of $\mathbf{S}^{\boxtimes}$ with respect to $\otimes$ on local vector bundles: 
\[ \mathbf{S}^{\boxtimes} \big( \mathbf{S}^{\otimes} V \big) \otimes \mathbf{S}^{\boxtimes} \big( \mathbf{S}^{\otimes} V \big) \cong \mathbf{S}^{\boxtimes} \Big( \mathbf{S}^{\otimes}(V) \otimes \mathbf{S}^{\otimes}(V) \Big) \longrightarrow \mathbf{S}^{\boxtimes} \big( \mathbf{S}^{\otimes} V \big). \]
More concretely, over a configuration $(x_1, \ldots, x_n) \in M^n$, if $a = a_1 \boxdot \ldots \boxdot a_n$ and $b = b_1 \boxdot \ldots \boxdot b_n$ are pure tensors of $\mathbf{S}^{\boxtimes} \big( \mathbf{S}^{\otimes} V \big)$, where $a_i, b_i \in \mathbf{S}^{\otimes}(V)_{x_i}$, then 
\[ a \odot b = (a_1 \odot b_1) \boxdot \ldots \boxdot (a_n \odot b_n) \in \mathbf{S}^{\boxtimes} \big( \mathbf{S}^{\otimes} V \big). \]
The $\boxtimes$-unit is the unit of the free commutative $\boxtimes$-algebra $\mathbf{S}^{\boxtimes} \big( \mathbf{S}^{\otimes} V \big)$. The $\otimes$-unit is induced by the $\otimes$-unit $\mathbf{I}_{\otimes} \rightarrow \mathbf{S}^{\otimes}(V)$ and the strong monoidality of $\mathbf{S}^{\boxtimes}$ with respect to $\otimes$ on local vector bundles :
\[ u^{\otimes}: \mathbf{I}_{\otimes} \cong \mathbf{S}^{\boxtimes} \big( \mathbf{I}_{\otimes} \big) \longrightarrow \mathbf{S}^{\boxtimes} \big( \mathbf{S}^{\otimes} V \big). \]
By a slight abuse of notation, we denote by the same symbol the $\otimes$-units of the $\otimes$-algebras $\mathbf{S}^{\otimes}(V)$ and $\mathbf{S}^{\boxtimes} \big( \mathbf{S}^{\otimes}(V) \big)$. More concretely, over a configuration $(x_1, \ldots, x_n) \in M^n$, the unit is 
\[ \mathds{1}_{x_1, \ldots, x_n}^{\otimes} = \mathds{1}_{x_1} \boxdot \ldots \boxdot \mathds{1}_{x_n}. \]
To verify the interchange law, take $a_1, a_2 \in \mathbf{S}^{\boxtimes} \big( \mathbf{S}^{\otimes} V \big)$ above a configuration $(x_1, \ldots, x_p)$ and $b_1, b_2 \in \mathbf{S}^{\boxtimes} \big( \mathbf{S}^{\otimes} V \big)$ above a configuration $(y_1, \ldots, y_q)$. Write 
\[ a_i = a_i^{(1)} \boxdot \ldots \boxdot a_i^{(p)} \qquad \text{and} \qquad b_i = b_i^{(1)} \boxdot \ldots \boxdot b_i^{(q)}. \]
Then
\[ \begin{array}{lll} (a_1 \boxdot b_1) \odot (a_2 \boxdot b_2) &=& \big( a_{1}^{(1)} \boxdot \ldots \boxdot a_1^{(p)} \boxdot b_1^{(1)} \boxdot \ldots \boxdot b_1^{(q)} \big) \odot \big( a_{2}^{(1)} \boxdot \ldots \boxdot a_2^{(p)} \boxdot b_2^{(1)} \boxdot \ldots \boxdot b_2^{(q)} \big) \\ 
    &=& (a_1^{(1)} \odot a_2^{(1)}) \boxdot \ldots \boxdot (a_1^{(p)} \odot a_2^{(p)}) \boxdot (b_1^{(1)} \odot b_2^{(1)}) \boxdot \ldots \boxdot (b_1^{(q)} \odot b_2^{(q)}) \\ 
    &=& (a_1 \odot a_2) \boxdot (b_1 \odot b_2). \end{array} \] 
The compatibility of units is immediate from the construction above. Since both multiplications are commutative, $\mathbf{S}^{\boxtimes} \big( \mathbf{S}^{\otimes} V \big)$ is therefore a commutative $2$-algebra.

We now prove the universal property. Let $\mathbf{A}$ be a commutative $2$-algebra and $f: V \longrightarrow \mathbf{A}$ be an equivariant bundle morphism. Since $\mathbf{A}$ is in particular a $\otimes$-algebra, $f$ extends to a $\otimes$-algebra morphism:
\[ \mathbf{S}^{\otimes}(f): \mathbf{S}^{\otimes}(V) \longrightarrow \mathbf{A}. \]
Since $\mathbf{A}$ is also a $\boxtimes$-algebra, this further extends to a $\boxtimes$-algebra morphism 
\[ \mathbf{S}^{\boxtimes} \big( \mathbf{S}^{\otimes}(f) \big) =: \widetilde{f}: \mathbf{S}^{\boxtimes} \big( \mathbf{S}^{\otimes}(V) \big) \longrightarrow \mathbf{A}. \]
It remains to show that $\widetilde{f}$ is still a $\otimes$-algebra morphism. Consider the two morphisms 
\[ \mathbf{S}^{\boxtimes} \big( \mathbf{S}^{\otimes} V \big) \otimes \mathbf{S}^{\boxtimes} \big( \mathbf{S}^{\otimes} V \big) \cong \mathbf{S}^{\boxtimes} \big( \mathbf{S}^{\otimes} V \otimes \mathbf{S}^{\otimes} V \big) \longrightarrow \mathbf{A} \]
induced by 
\[ \alpha(v,w) = \widetilde{f}(v \odot w) \qquad \text{and} \qquad \beta(v,w) = m_{\mathbf{A}}^{\otimes} \big( \widetilde{f}(v),  \widetilde{f}(w) \big). \]
Since $m_{\mathbf{A}}^{\otimes}$ is a $\boxtimes$-algebra morphism, so is $\beta$, and by construction, $\alpha$ is also a $\boxtimes$-algebra morphism. Since both are defined on the free $\boxtimes$-algebra $\mathbf{S}^{\boxtimes} \big( \mathbf{S}^{\otimes} V \otimes \mathbf{S}^{\otimes} V \big)$, it suffices to check that $\alpha = \beta$ on generators, namely on $\mathbf{S}^{\otimes}(V) \otimes \mathbf{S}^{\otimes}(V)$. On $\mathbf{S}^{\otimes}(V) \otimes \mathbf{S}^{\otimes}(V)$, one has
\[ \alpha|_{\mathbf{S}^{\otimes}(V) \otimes \mathbf{S}^{\otimes}(V)} = \mathbf{S}^{\otimes}(f) \circ m_{\mathbf{S}^{\otimes} V}^{\otimes} \qquad \text{and} \qquad \beta|_{\mathbf{S}^{\otimes}(V) \otimes \mathbf{S}^{\otimes}(V)} = m_{\mathbf{A}}^{\otimes} \circ \big( \mathbf{S}^{\otimes}(f) \otimes \mathbf{S}^{\otimes}(f) \big). \]
Since $\mathbf{S}^{\otimes}(f)$ is a $\otimes$-algebra morphism, these two restrictions coincide, hence $\alpha = \beta$. Therefore $\widetilde{f}$ is a $\otimes$-algebra morphism.

Finally, let $g: \mathbf{S}^{\boxtimes} \big( \mathbf{S}^{\otimes}(V) \big) \rightarrow \mathbf{A}$ be a $2$-algebra morphism extending $f$. Restricting $g$ along the inclusion $\mathbf{S}^{\otimes}(V) \rightarrow \mathbf{S}^{\boxtimes} \big( \mathbf{S}^{\otimes} V \big)$ gives a $\otimes$-algebra morphism $\mathbf{S}^{\otimes} V \rightarrow \mathbf{A}$ extending $f$. By the universal property of $\mathbf{S}^{\otimes}$, one has
\[ g|_{\mathbf{S}^{\otimes}(V)} = \mathbf{S}^{\otimes}(f). \]
Since $g$ is also a $\boxtimes$-algebra extending $\mathbf{S}^{\otimes}(f)$, the universal property of $\mathbf{S}^{\boxtimes}$ implies that $g=\widetilde{f}$. This proves the uniqueness, and hence the theorem.
\end{proof}

\begin{Remark}
If $V$ is not local, then the functor $\mathbf{S}^{\boxtimes}$ is no longer strongly monoidal. This property is crucial in the previous theorem, both for defining the $\otimes$-multiplication and for proving the universal property: the isomorphism allows one to check that the extended map $\widetilde{f}$ is a $\otimes$-algebra morphism by restricting to generators. 

In the general case, the oplax monoidality of $\mathbf{S}^{\boxtimes}$ only provides a morphism 
\[ \mathbf{S}^{\boxtimes} \Big( \mathbf{S}^{\otimes}(\mathbf{V}) \otimes \mathbf{S}^{\otimes}(\mathbf{V}) \Big) \longrightarrow \mathbf{S}^{\boxtimes} \big( \mathbf{S}^{\otimes}(\mathbf{V}) \big) \otimes \mathbf{S}^{\boxtimes} \big( \mathbf{S}^{\otimes}(\mathbf{V}) \big) \]
which goes in the opposite direction to what is required to define a $\otimes$-multiplication on $\mathbf{S}^{\boxtimes} \big( \mathbf{S}^{\otimes}(\mathbf{V}) \big)$.
\end{Remark}

%% file: 5_Poisson.tex
\section{Equivariant Poisson $2$-Algebra Structures} 

Now that a $2$-algebra structure has been introduced, we consider a Poisson structure compatible with the two types of multiplications. We begin by describing the standard notion of a Poisson structure with respect to the $\boxtimes$-structure, which will serve as the model for the later compatibility with the $\otimes$-structure.

\subsection{Equivariant Poisson-Cauchy Algebra Bundles over Configuration Spaces} \label{4.1}

\begin{Definition}
An \textbf{equivariant Poisson $\boxtimes$-algebra} is an equivariant commutative $\boxtimes$-algebra bundle $\big( \mathbf{P}, m^{\boxtimes}, u^{\boxtimes} \big)$ with a Poisson bracket 
\[ B: \mathbf{P} \boxtimes \mathbf{P} \longrightarrow \mathbf{P} \]
such that the following conditions hold: \begin{enumerate}
    \item $B$ is skew-symmetric, i.e. $B \circ \beta^{\boxtimes} = - B$ where $\beta^{\boxtimes}$ is the $\boxtimes$-braiding.
    \item $B$ satisfies the Jacobi identity, where the circular braiding is defined by 
    \[ \sigma^{\boxtimes} = \big( \mathrm{Id}_{\mathbf{P}} \boxtimes \beta^{\boxtimes} \big) \circ \big( \beta^{\boxtimes} \boxtimes \mathrm{Id}_{\mathbf{P}} \big): \mathbf{P} \boxtimes \mathbf{P} \boxtimes \mathbf{P} \longrightarrow \mathbf{P} \boxtimes \mathbf{P} \boxtimes \mathbf{P} \]
    and the identity reads
    \[ B \circ \big( B \boxtimes \mathrm{Id}_{\mathbf{P}} \big) \circ \big( \mathrm{Id} + \sigma^{\boxtimes} + \sigma^{\boxtimes} \circ \sigma^{\boxtimes} \big) = 0 \] 
    \item $B$ satisfies the $\boxtimes$-Leibniz rule, i.e. 
    \[ B \circ \big( \mathrm{Id}_{\mathbf{P}} \boxtimes m^{\boxtimes} \big) = m^{\boxtimes} \circ \big( B \boxtimes \mathrm{Id}_{\mathbf{P}} \big) + m^{\boxtimes} \circ \big( \mathrm{Id}_{\mathbf{P}} \boxtimes B \big) \circ \big( \beta^{\boxtimes} \boxtimes \mathrm{Id}_{\mathbf{P}} \big). \]
\end{enumerate}
\end{Definition}

Equivalently, a Poisson bracket $B: \mathbf{P} \boxtimes \mathbf{P} \rightarrow \mathbf{P}$ is induced by "elementary Poisson brackets", namely $\mathfrak{S}_p \times \mathfrak{S}_q$-equivariant bilinear maps :
\[ \left\{ \cdot, \; \cdot \right\}_{p,q}: P_p \times P_q \longrightarrow P_{p+q} \]
which send a vector over $(x_1, \ldots, x_p)$ and a vector over $(y_1, \ldots, y_q)$ to a vector over $(x_1, \ldots, x_p, y_1, \ldots, y_q)$. In terms of these elementary brackets, the Poisson bracket axioms become the following: \begin{enumerate} 
\item \textbf{Skew-symmetry} is given by 
\[ \left\{ a,b \right\}_{p,q} = - \tau_{p,q} \cdot \left\{ b,a \right\}_{q,p} \]
where $\tau_{p,q}$ is the block permutation $\begin{bmatrix} q+1 & \ldots & q+p & | & 1 & \ldots & q \end{bmatrix}$ sending the configuration $(y_1, \ldots, y_q, x_1, \ldots, x_p)$ to $(x_1, \ldots, x_p, y_1, \ldots, y_q)$.

\item \textbf{Jacobi identity.} Consider the block permutation 
\[ \sigma_{p,q,r} = \left[ \begin{array}{ccc|ccc|ccc} p+1 & \ldots & p+q & p+q+1 & \ldots & p+q+r & 1 & \ldots & p \end{array} \right] \in \mathfrak{S}_{p+q+r} \]
and
\[ \sigma_{p,q,r}^{(2)} = \left[ \begin{array}{ccc|ccc|ccc} p+q+1 & \ldots & p+q+r & 1 & \ldots & p & p+1 & \ldots & p+q \end{array} \right] \in \mathfrak{S}_{p+q+r} \] 
These permutations cyclically permute the three blocks of sizes $p,q,r$, namely
\[ \begin{array}{lll} \sigma_{p,q,r}(X,Y,Z) &=& (Y,Z,X) \\
\sigma_{p,q,r}^{(2)}(X,Y,Z) &=& (Z,X,Y) \end{array} \qquad \text{where} \qquad \left\{ \begin{array}{lll} X &=& (x_1, \ldots, x_p) \in M^p \\ 
Y &=& (y_1, \ldots, y_q) \in M^q \\ 
Z &=& (z_1, \ldots, z_r) \in M^r \end{array} \right. . \]
Note that $\sigma_{p,q,r}^{(2)} \neq \sigma_{p,q,r} \circ \sigma_{p,q,r}$, since the naive square composition does not account for the change of block sizes. However, one has $\sigma_{p,q,r}^{(2)} = \sigma_{q,r,p} \circ \sigma_{p,q,r}$. Thus the Jacobi identity reads 
\[ \left\{ \left\{ a,b \right\}_{p,q}, c \right\}_{p+q,r} + \sigma_{p,q,r} \cdot \left\{ \left\{ b,c \right\}_{q,r}, a \right\}_{q+r, p} + \sigma_{p,q,r}^{(2)} \cdot \left\{ \left\{ c,a \right\}_{r,p}, b \right\}_{r+p, q} = 0. \]
\item \textbf{$\boxtimes$-Leibniz rule} is given by 
\[ \left\{ a, m_{q,r}^{\boxtimes}(b,c) \right\}_{p,q+r} = m_{p+q, r}^{\boxtimes} \Big( \left\{ a,b \right\}_{p,q}, c \Big) + \tau_{p,q,r} \cdot m_{q,p+r}^{\boxtimes} \Big( b, \left\{ a,c \right\}_{p,r} \Big) \]
where $\tau_{p,q,r}$ is the block permutation $\begin{bmatrix} \tau_{p,q} &|& p+q+1 & \ldots & p+q+r \end{bmatrix}$ permuting only the first two blocks and fixing the last one, so that 
\[ \tau_{p,q,r} (X,Y,Z) = (Y,X,Z) \qquad \text{where} \qquad \left\{ \begin{array}{lll} X &=& (x_1, \ldots, x_p) \in M^p \\ 
Y &=& (y_1, \ldots, y_q) \in M^q \\ 
Z &=& (z_1, \ldots, z_r) \in M^r \end{array} \right. . \]
\end{enumerate}

A morphism between equivariant Poisson $\boxtimes$-algebras is a $\boxtimes$-algebra bundle morphism $f: (\mathbf{P}_1, B_1) \longrightarrow (\mathbf{P}_2, B_2)$ such that the following diagram commutes 
\[ \xymatrix{
\mathbf{P}_1 \boxtimes \mathbf{P}_1 \ar[rr]^{f \boxtimes f} \ar[d]_{B_1} && \mathbf{P}_2 \boxtimes \mathbf{P}_2 \ar[d]^{B_2} \\ 
\mathbf{P}_1 \ar[rr]_{f} && \mathbf{P}_2
} \]
that is,
\[ \forall (a,b) \in (\mathbf{P}_1)_p \times (\mathbf{P}_1)_q, \qquad f_{p+q} \big( \left\{ a,b \right\}_{p,q}^{(1)} \big) = \big\{ f_p(a), f_q(b) \big\}_{p,q}^{(2)}. \]
Therefore, equivariant Poisson $\boxtimes$-algebras form a category denoted $\mathbf{Pois}^{\boxtimes} \big( \mathcal{VB}_{\mathfrak{S}_{\bullet}}(M^{\bullet}) \big)$.

\begin{Proposition}
Let $(\mathbf{P}_1, m_1^{\boxtimes}, u_1^{\boxtimes}, B_1)$ and $(\mathbf{P}_2, m_2^{\boxtimes}, u_2^{\boxtimes}, B_2^{\boxtimes})$ be equivariant Poisson $\boxtimes$-algebra bundles. Then, $\mathbf{P}_1 \boxtimes \mathbf{P}_2$ has a natural Poisson structure given by: 
\[ B: (\mathbf{P}_1 \boxtimes \mathbf{P}_2) \boxtimes (\mathbf{P}_1 \boxtimes \mathbf{P}_2) \cong (\mathbf{P}_1 \boxtimes \mathbf{P}_1) \boxtimes (\mathbf{P}_2 \boxtimes \mathbf{P}_2) \overset{B_1 \boxtimes m_2 + m_1 \boxtimes B_2} \longrightarrow \mathbf{P}_1 \boxtimes \mathbf{P}_2. \]
\end{Proposition}

\begin{proof} 
The bracket $B$ on $\mathbf{P}_1 \boxtimes \mathbf{P}_2$ is given by the elementary Poisson bracket 
\[ \left\{ a_1 \boxtimes a_2, b_1 \boxtimes b_2 \right\}_{p,q} = \left\{ a_1, b_1 \right\}_{p_1,q_1}^{(1)} \boxtimes m^{\boxtimes}_{2;p_2,q_2}(a_2,b_2) + m_{1;p_1,q_1}^{\boxtimes}(a_1, b_1) \boxtimes \left\{ a_2,b_2 \right\}_{p_2,q_2}^{(2)} \]
where $p_1 + p_2 = p$ and $q_1 + q_2 = q$ and $a_1 \in (\mathbf{P}_1)_{p_1}$, $a_2 \in (\mathbf{P}_2)_{p_2}$, $b_1 \in (\mathbf{P}_1)_{q_1}$ and $b_2 \in (\mathbf{P}_2)_{q_2}$. Skew-symmetry comes directly from the skew-symmetry of $B_1$ and $B_2$ and the commutativity of $m_1^{\boxtimes}$ and $m_2^{\boxtimes}$. We now check the $\boxtimes$-Leibniz rule: for $a_1 \in (\mathbf{P}_1)_{p_1}$, $a_2 \in (\mathbf{P}_2)_{p_2}$, $b_1 \in (\mathbf{P}_1)_{q_1}$, $b_2 \in (\mathbf{P}_2)_{q_2}$ and $c_1 \in (\mathbf{P}_1)_{r_1}$, $c_2 \in (\mathbf{P}_2)_{r_2}$ where $p_1 + p_2 = p$, $q_1 + q_2 = q$ and $r_1 + r_2 = r$, one has
\[ \begin{array}{lll} \left\{ a_1 \boxtimes a_2, m_{q,r}^{\boxtimes} \big( b_1 \boxtimes b_2, c_1 \boxtimes c_2 \big) \right\}_{p,q+r} &=& \left\{ a_1 \boxtimes a_2, m_{1;q_1,r_1}(b_1,c_1) \boxtimes m_{2;q_2,r_2}(b_2,c_2) \right\}_{p,q+r} \\ 
&=& \left\{ a_1, m_{1;q_1,r_1}^{\boxtimes}(b_1, c_1) \right\}_{p_1, q_1+r_1}^{(1)} \boxtimes m_{2;p_2, q_2+r_2}^{\boxtimes} \big( a_2, m_{2;q_2,r_2}^{\boxtimes}(b_2, c_2) \big) \\
&& + m_{1;p_1,q_1+r_1}^{\boxtimes} \big( a_1, m_{1;q_1,r_1}^{\boxtimes}(b_1, c_1) \big) \boxtimes \left\{ a_2, m_{2;q_2,r_2}^{\boxtimes}(b_2, c_2) \right\}_{p_2, q_2+r_2}. \end{array}\]
We now use the $\boxtimes$-Leibniz rule for $B_1, B_2$ and the associativity of $m_1^{\boxtimes}, m_2^{\boxtimes}$ :
\[ \begin{array}{lll} \left\{ a_1 \boxtimes a_2, m_{q,r}^{\boxtimes}(b_1 \boxtimes b_2, c_1 \boxtimes c_2) \right\}_{p,q+r} &=& m_{1;p_1 + q_1, r_1}^{\boxtimes} \big( \left\{ a_1, b_1 \right\}_{p_1, q_1}^{(1)}, c_1 \big) \boxtimes m_{2;p_2 + q_2, r_2}^{\boxtimes} \big( m_{2;p_2, q_2}^{\boxtimes}(a_2, b_2), c_2 \big) \\ 
&& + \tau_{p_1,q_1,r_1} \cdot m_{1;q_1, p_1+r_1} \big( b_1, \left\{ a_1, c_1 \right\}_{p_1, r_1}^{(1)} \big) \boxtimes m_{2;p_2+q_2,r_2}^{\boxtimes} \big( m_{2;p_2,q_2}^{\boxtimes}(a_2, b_2), c_2 \big) \\
&& + m^{\boxtimes}_{1;p_1+q_1,r_1} \big( m_{1;p_1,q_1}^{\boxtimes}(a_1, b_1), c_1 \big) \boxtimes m_{2;p_2+q_2, r_2}^{\boxtimes} \big( \left\{ a_2, b_2 \right\}_{p_2,q_2}^{(2)}, c_2 \big) \\
&& + m_{1;p_1+q_1,r_1}^{\boxtimes} \big( m_{1;p_1,q_1}^{\boxtimes}(a_1, b_1), c_1 \big) \boxtimes \tau_{p_2,q_2,r_2} \cdot m_{2;q_2, p_2+r_2}^{\boxtimes} \big( b_2, \left\{ a_2, c_2 \right\}_{p_2,r_2}^{(2)} \big). \end{array} \]
Consider a block decomposition $(X,Y,Z)$ where $|X| = p$, $|Y|=q$ and $|Z|=r$, then each block $X$ (resp. $Y$ and $Z$) decomposes into $X = (X_1,X_2)$ with $|X_1| = p_1$ and $|X_2| = p_2$ (resp. $Y = (Y_1,Y_2)$ with $|Y_1| = q_1$, $|Y_2| = q_2$ and $Z = (Z_1, Z_2)$ with $|Z_1| = r_1$, $|Z_2| = r_2$). Then the global block permutation acts by permuting the blocks $X,Y,Z$ without changing their internal order, i.e.
\[ \tau_{p,q,r}(X,Y,Z) = (Y,X,Z) = (Y_1,Y_2,X_1,X_2,Z_1,Z_2). \] 
Hence it acts by $\tau_{p_1,q_1,r_1}$ on the $\mathbf{P}_1$-part and by $\tau_{p_2,q_2,r_2}$ on the $\mathbf{P}_2$-part. Therefore, by using commutativity of the two multiplications, this yields
\[ \begin{array}{lll} \left\{ a_1 \boxtimes a_2, m_{q,r}^{\boxtimes}(b_1 \boxtimes b_2, c_1 \boxtimes c_2) \right\}_{p,q+r} &=& m_{1;p_1 + q_1, r_1}^{\boxtimes} \big( \left\{ a_1, b_1 \right\}_{p_1, q_1}^{(1)}, c_1 \big) \boxtimes m_{2;p_2 + q_2, r_2}^{\boxtimes} \big( m_{2;p_2, q_2}^{\boxtimes}(a_2, b_2), c_2 \big) \\ 
&& + \tau_{p,q,r} \cdot \Big[ m_{1;q_1, p_1+r_1} \big( b_1, \left\{ a_1, c_1 \right\}_{p_1, r_1}^{(1)} \big) \boxtimes m_{2;p_2+q_2,r_2}^{\boxtimes} \big( m_{2;p_2,q_2}^{\boxtimes}(a_2, b_2), c_2 \big) \Big] \\
&& + m^{\boxtimes}_{1;p_1+q_1,r_1} \big( m_{1;p_1,q_1}^{\boxtimes}(a_1, b_1), c_1 \big) \boxtimes m_{2;p_2+q_2, r_2}^{\boxtimes} \big( \left\{ a_2, b_2 \right\}_{p_2,q_2}^{(2)}, c_2 \big) \\
&& + \tau_{p,q,r} \cdot \Big[ m_{1;p_1+q_1,r_1}^{\boxtimes} \big( m_{1;p_1,q_1}^{\boxtimes}(a_1, b_1), c_1 \big) \boxtimes m_{2;q_2, p_2+r_2}^{\boxtimes} \big( b_2, \left\{ a_2, c_2 \right\}_{p_2,r_2}^{(2)} \big) \Big]. \normalsize \end{array} \]
By combining the first and third term, and the second and the last term, yielding 
\[ \left\{ a_1 \boxtimes a_2, m_{q,r}^{\boxtimes}(b_1 \boxtimes b_2, c_1 \boxtimes c_2) \right\}_{p,q+r} = m_{p+q,r}^{\boxtimes} \Big( \left\{ a_1 \boxtimes a_2, b_1 \boxtimes b_2 \right\}_{p,q}, c_1 \boxtimes c_2 \Big) 
+ \tau_{p,q,r} \cdot m_{q,p+r}^{\boxtimes} \Big( b_1 \boxtimes b_2, \left\{ a_1 \boxtimes a_2, c_1 \boxtimes c_2 \right\}_{p,r} \Big) \]
which is the $\boxtimes$-Leibniz rule. We now check the Jacobi identity. Consider the jacobiator $J: (\mathbf{P}_1 \boxtimes \mathbf{P}_2)^{\boxtimes 3} \longrightarrow \mathbf{P}_1 \boxtimes \mathbf{P}_2$ defined by the elementary jacobiators
\[ \begin{array}{lll} J_{p,q,r} \big( a_1 \boxtimes a_2, b_1 \boxtimes b_2, c_1 \boxtimes c_2 \big) &=& \left\{ \left\{ a_1 \boxtimes a_2, b_1 \boxtimes b_2 \right\}_{p,q}, c_1 \boxtimes c_2 \right\}_{p+q,r} \\ 
&& + \sigma_{p,q,r} \cdot \left\{ \left\{ b_1 \boxtimes b_2, c_1 \boxtimes c_2 \right\}_{q,r}, a_1 \boxtimes a_2 \right\}_{q+r,p} \\
&&+ \sigma_{p,q,r}^{(2)} \cdot \left\{ \left\{ c_1 \boxtimes c_2, a_1 \boxtimes a_2 \right\}_{r,p}, b_1 \boxtimes b_2 \right\}_{r+p,q}. \end{array} \]
By using the $\boxtimes$-Leibniz rule using associativity, this yields six "pure" terms and twelve "mixed" terms. More precisely, 
\[ J_{p,q,r} = J_{p,q,r}^{\mathrm{pure}} + J_{p,q,r}^{\mathrm{mixed}} \]
with 
\[ \begin{array}{lll} J_{p,q,r}^{\mathrm{pure}}(a_1 \boxtimes a_2, b_1 \boxtimes b_2, c_1 \boxtimes c_2) &=& J_{1;p_1,q_1,r_1}(a_1,b_1,c_1) \boxtimes m_{2;p_2,q_2,r_2}^{\boxtimes}(a_2,b_2,c_2)\\ 
&& + m_{1;p_1,q_1,r_1}^{\boxtimes}(a_1,b_1,c_1) \boxtimes J_{2;p_2,q_2,r_2}(a_2,b_2,c_2) \end{array} \]
where $J_k$ is the jacobiator associated to the Poisson bracket $B_k$. Therefore, 
\[ J_{p,q,r}^{\mathrm{pure}} = 0 \]
and, 
\[ \begin{array}{llll} J_{p,q,r}^{\mathrm{mixed}} \big( a_1 \boxtimes a_2, b_1 \boxtimes b_2, c_1 \boxtimes c_2 \big) &=& m_{1;p_1+q_1, r_1}^{\boxtimes} \big( \left\{ a_1, b_1 \right\}_{p_1,q_1}^{(1)}, c_1 \big) \boxtimes m^{\boxtimes}_{2;p_2, q_2+r_2} \big( a_2, \left\{ b_2, c_2 \right\}_{q_2,r_2}^{(2)} \big) & (1) \\ 
&& + m^{\boxtimes}_{1;p_1+q_1,r_1} \big( \left\{ a_1, b_1 \right\}_{p_1,q_1}^{(1)}, c_1 \big) \boxtimes m_{2;q_2, p_2+r_2}^{\boxtimes} \big( b_2, \left\{ a_2, c_2 \right\}_{p_2, r_2}^{(2)} \big) & (2) \\ 
&& + m_{1;p_1,q_1+r_1}^{\boxtimes} \big( a_1, \left\{ b_1, c_1 \right\}_{q_1,r_1}^{(1)} \boxtimes m_{2;p_2+q_2, r_2}̂^{\boxtimes} \big( \left\{ a_2, b_2 \right\}_{p_2, q_2}^{(2)}, c_2 \big) &(3)  \\ 
&& + m_{1;q_1,p_1+r_1}^{\boxtimes} \big( b_1, \left\{ a_1, c_1 \right\}_{p_1,r_1}^{(1)} \big) \boxtimes m_{2;p_2+q_2,r_2}^{\boxtimes} \big( \left\{ a_2, b_2 \right\}_{p_2,q_2}^{(2)}, c_2 \big) &(4)\\ 

&+& m_{1;q_1+r_1, p_1}^{\boxtimes} \big( \left\{ b_1, c_1 \right\}_{q_1,r_1}^{(1)}, a_1 \big) \boxtimes m_{2;q_2, r_2 + p_2}^{\boxtimes} \big( b_2, \left\{ c_2, a_2 \right\}_{r_2,p_2}^{(2)} \big) &(5) \\ 
&& + m_{1;q_1+r_1, p_1}^{\boxtimes} \big( \left\{ b_1, c_1 \right\}_{q_1,r_1}^{(1)}, a_1 \big) \boxtimes m_{2;r_2, q_2+p_2}^{\boxtimes} \big( c_2, \left\{ b_2, a_2 \right\}_{q_2, p_2}^{(2)} \big) &(6) \\
&& + m_{1;q_1, r_1+p_1}^{\boxtimes} \big( b_1, \left\{ c_1, a_1 \right\}_{r_1, p_1}^{(1)} \big) \boxtimes m_{2;q_2+r_2,p_2}^{\boxtimes} \big( \left\{ b_2, c_2 \right\}_{q_2,r_2}^{(2)}, a_2 \big) &(7) \\ 
&& + m_{1;r_1, q_1+p_1}^{\boxtimes} \big( c_1, \left\{ b_1, a_1 \right\}_{q_1,p_1}^{(1)} \big) \boxtimes m_{2;q_2+r_2, p_2}^{\boxtimes} \big( \left\{ b_2, c_2 \right\}_{q_2, r_2}^{(2)}, a_2 \big) &(8) \\ 

&+& m_{1;r_1+p_1, q_1}^{\boxtimes} \big( \left\{ c_1, a_1 \right\}_{r_1, p_1}^{(1)}, b_1 \big) \boxtimes m_{2;r_2, p_2+q_2}^{\boxtimes} \big( c_2, \left\{ a_2, b_2 \right\}_{p_2,q_2}^{(2)} \big) &(9) \\ 
&& + m_{1;r_1+p_1, q_1}^{\boxtimes} \big( \left\{ c_1, a_1 \right\}_{r_1, p_1}^{(1)}, b_1 \big) \boxtimes m_{2;p_2, r_2+q_2}^{\boxtimes} \big( a_2, \left\{ c_2, b_2 \right\}_{r_2, q_2}^{(2)} \big) &(10) \\ 
&& + m_{1;r_1, p_1+q_1}^{\boxtimes} \big( c_1, \left\{ a_1, b_1 \right\}_{p_1,q_1}^{(1)} \big) \boxtimes m_{2;r_2+p_2, q_2}^{\boxtimes} \big( \left\{ c_2, a_2 \right\}_{r_2,p_2}^{(2)}, b_2 \big) &(11) \\ 
&& + m_{1;p_1, r_1+q_1}^{\boxtimes} \big( a_1, \left\{ c_1, b_1 \right\}_{r_1,q_1}^{(1)} \big) \boxtimes m_{2;r_2+p_2, q_2}^{\boxtimes} \big( \left\{ c_2, a_2 \right\}_{r_2,p_2}^{(2)}, b_2 \big). &(12) \end{array} \]
The following pairs of terms cancel out due to skew-symmetry of the bracket and commutativity of the multiplications: $(1)$-$(8)$, $(2)$-$(11)$, $(6)$-$(3)$, $(9)$-$(4)$, $(12)$-$(5)$, $(10)$-$(7)$. Therefore, it only remains 
\[ J_{p,q,r}(a_1 \boxtimes a_2, b_1 \boxtimes b_2, c_1 \boxtimes c_2) = 0.\]
Hence the bracket $B$ satisfies the Jacobi identity.
\end{proof}

In fact, $\boxtimes$ is the coproduct in the category $\mathbf{Pois}^{\boxtimes} \big( \mathcal{VB}_{\mathfrak{S}_{\bullet}}(M^{\bullet}) \big)$.

\begin{Proposition} 
Let $(\mathbf{P}_1, B_1)$ and $(\mathbf{P}_2, B_2)$ be equivariant Poisson $\boxtimes$-algebra bundles, then $\mathbf{P}_1 \otimes \mathbf{P}_2$ carries a canonical Poisson $\boxtimes$-algebra structure defined by 
\[ B = \big( B_1 \otimes m_2^{\boxtimes} + m_1^{\boxtimes} \otimes B_2 \big) \circ \zeta \]
that is 
\[ \left\{ a_1 \otimes a_2, b_1 \otimes b_2 \right\}_{p,q} = \left\{ a_1, b_1 \right\}_{p,q}^{(1)} \otimes m^{\boxtimes}_{p,q}(a_2, b_2) + m_{p,q}^{\boxtimes}(a_1, b_1) \otimes \left\{ a_2, b_2 \right\}_{p,q}^{(2)}. \]
\end{Proposition}

\begin{proof} 
We denote $m^{\boxtimes}$ the multiplication on $\mathbf{P}_1 \otimes \mathbf{P}_2$, recall that 
\[ m_{p,q}^{\boxtimes}(a_1 \otimes a_2, b_1 \otimes b_2) = m_{1;p,q}^{\boxtimes}(a_1, b_1) \otimes m_{2;p,q}^{\boxtimes}(a_2, b_2) \]
for $a_1 \in (\mathbf{P}_1)_p$, $a_2 \in (\mathbf{P}_2)_p$, $b_1 \in (\mathbf{P}_1)_{q}$ and $b_2 \in (\mathbf{P}_2)_q$. We check skew-symmetry 
\[ \tau_{q,p} \cdot \left\{ b_1 \otimes b_2, a_1 \otimes a_2 \right\}_{q,p} = \tau_{q,p} \cdot \Big( \left\{ b_1, a_1 \right\}_{q,p}^{(1)} \otimes m^{\boxtimes}_{2;q,p}(b_2, a_2) + m_{1;q,p}^{\boxtimes}(b_1, a_1) \otimes \left\{ b_2, a_2 \right\}_{q,p}^{(2)} \Big). \]
The permutation $\tau_{q,p}$ acts identically on both tensor factors. Using skew-symmetry of $B_1,B_2$ and commutativity for $m_2^{\boxtimes}, m_1^{\boxtimes}$, we get 
\[ \tau_{q,p} \cdot \left\{ b_1 \otimes b_2, a_1 \otimes a_2 \right\}_{q,p} = - \left\{ a_1 \otimes a_2, b_1 \otimes b_2 \right\}_{p,q}. \]
Hence $B$ is skew-symmetric. We now check the $\boxtimes$-Leibniz rule: for $a_1 \otimes a_2 \in (\mathbf{P}_1 \otimes \mathbf{P}_2)_p$, $b_1 \otimes b_2 \in (\mathbf{P}_1 \otimes \mathbf{P}_2)_q$ and $c_1 \otimes c_2 \in (\mathbf{P}_1 \otimes \mathbf{P}_2)_r$, a direct computation gives 
\[ \begin{array}{lll} \left\{ a_1 \otimes a_2, m_{q,r}^{\boxtimes} \big( b_1 \otimes b_2, c_1 \otimes c_2 \big) \right\}_{p,q+r} &=& \left\{ a_1, m_{1;q,r}^{\boxtimes}(b_1, c_1) \right\}_{p,q+r}^{(1)} \otimes m_{2;p,q+r}^{\boxtimes} \big( a_2, m_{2;q,r}^{\boxtimes}(b_2, c_2) \big) \\
&& + m_{1;p,q+r}^{\boxtimes} \big( a_1, m_{1;q,r}^{\boxtimes}(b_1, c_1) \big) \otimes \left\{ a_2, m_{2;q,r}^{\boxtimes}(b_2, c_2) \right\}_{p,q+r}^{(2)}. \end{array} \]
Applying the $\boxtimes$-Leibniz rule for $B_1$ and $B_2$ and the associativity of $m_1^{\boxtimes}, m_2^{\boxtimes}$, one obtains four terms
\[ \begin{array}{lll} \left\{ a_1 \otimes a_2, m_{q,r}^{\boxtimes} \big( b_1 \otimes b_2, c_1 \otimes c_2 \big) \right\}_{p,q+r} &=& m_{1;p+q,r}^{\boxtimes} \big( \left\{ a_1, b_1 \right\}_{p,q}^{(1)}, c_1 \big) \otimes m_{2;p+q,r}^{\boxtimes} \big( m_{2;p,q}(a_2,b_2), c_2 \big) \\ 
&& + \tau_{p,q,r} \cdot \Big[ m_{1;q,p+r}^{\boxtimes} \big( b_1, \left\{ a_1, c_1 \right\}_{p,r}^{(1)} \big) \otimes m_{2;q,p+r}^{\boxtimes} \big( b_2, m_{2;p,r}^{\boxtimes}(a_2, c_2) \big) \Big] \\ 
&& + m_{1;p+q,r}^{\boxtimes} \big( m_{1;p,q}^{\boxtimes}(a_1, b_1), c_1 \big) \otimes m_{2;p+q,r}^{\boxtimes} \big( \left\{ a_2, b_2 \right\}_{p,q}^{(2)}, c_2 \big) \\ 
&& + \tau_{p,q,r} \cdot \Big[ m_{1;q,p+r}^{\boxtimes} \big( b_1, (m_1)_{p,r}^{\boxtimes}(a_1, c_1) \big) \otimes m_{2;q,p+r}^{\boxtimes} \big( b_2, \left\{ a_2, c_2 \right\}_{p,r}^{(2)} \big) \Big].
\end{array} \]
By combining the first and third term, and the second and last terms, one gets 
\[ \begin{array}{lll} \left\{ a_1 \otimes a_2, m_{q,r}^{\boxtimes} \big( b_1 \otimes b_2, c_1 \otimes c_2 \big) \right\}_{p,q+r} &=& m_{p+q,r}^{\boxtimes} \Big( \left\{ a_1 \otimes a_2, b_1 \otimes b_2 \right\}_{p,q}, c_1 \otimes c_2 \Big) \\ 
&& +\tau_{p,q,r} \cdot m_{q,p+r}^{\boxtimes} \Big( b_1 \otimes b_2, \left\{ a_1 \otimes a_2, c_1 \otimes c_2 \right\}_{p,r} \Big) \end{array} \]
which is the $\boxtimes$-Leibniz rule. Now define the jacobiator $J: (\mathbf{P}_1 \otimes \mathbf{P}_2)^{\boxtimes 3} \longrightarrow \mathbf{P}_1 \otimes \mathbf{P}_2$ by 
\[ \begin{array}{lll} J_{p,q,r}(a_1 \otimes a_2, b_1 \otimes b_2, c_1 \otimes c_2) &=& \left\{ \left\{ a_1 \otimes a_2, b_1 \otimes b_2 \right\}_{p,q}, c_1 \otimes c_2 \right\}_{p+q,r} \\ 
&& + \sigma_{p,q,r} \cdot \left\{ \left\{ b_1 \otimes b_2, c_1 \otimes c_2 \right\}_{q,r}, a_1 \otimes a_2 \right\}_{q+r,p} \\
&& + \sigma_{p,q,r}^{(2)} \cdot \left\{ \left\{ c_1 \otimes c_2, a_1 \otimes a_2 \right\}_{r,p}, b_1 \otimes b_2 \right\}_{r+p,q}. \end{array} \]
Expanding all the terms using the $\boxtimes$-Leibniz rule, this yields again six "pure" terms and twelve "mixed" terms and similarly to the previous proof, all these terms vanish either by the fact that $B_1$ and $B_2$ satisfy the Jacobi identity or by skew-symmetry.
\end{proof}

In fact, $\Big( \mathbf{Pois}^{\boxtimes} \big( \mathcal{VB}_{\mathfrak{S}_{\bullet}}(M^{\bullet}) \big), \otimes \Big)$ is a symmetric monoidal category.

\begin{Example}
In particular, if $\mathbf{P}$ is a Poisson $\boxtimes$-algebra and $\mathbf{A}$ is simply a $\boxtimes$-algebra, then $\mathbf{P} \otimes \mathbf{A}$ carries a natural Poisson structure induced by $\mathbf{P}$, with zero bracket on $\mathbf{A}$.
\end{Example}

For the usual Poisson bracket, the Leibniz rule is interpreted as a biderivation property in each variable. We now introduce the analogous notion in the $\boxtimes$-setting. Let us fix an equivariant commutative $\boxtimes$-algebra $(\mathbf{A}, m^{\boxtimes}, u^{\boxtimes})$. 

\begin{Definition} 
\textbf{A left $(\mathbf{A}, \boxtimes)$-module} is an equivariant vector bundle $\mathbf{N} \in \mathcal{VB}_{\mathfrak{S}_{\bullet}}(M^{\bullet})$ together with an action map $\mu^{\boxtimes}: \mathbf{A} \boxtimes \mathbf{N} \longrightarrow \mathbf{N}$ satisfying the usual associativity and unit axioms, namely 
\[ \mu^{\boxtimes} \circ \big( m^{\boxtimes} \boxtimes \mathrm{Id}_{\mathbf{N}} \big) = \mu^{\boxtimes} \circ \big( \mathrm{Id}_{\mathbf{A}} \boxtimes \mu^{\boxtimes} \big) \]
and 
\[ \mu^{\boxtimes} \circ \big( u^{\boxtimes} \boxtimes \mathrm{Id}_{\mathbf{N}} \big) = \lambda^{\boxtimes}_{\mathbf{N}} \]
where $\lambda_{\mathbf{N}}^{\boxtimes}: \mathbf{I}_{\boxtimes} \boxtimes \mathbf{N} \overset{\cong} \longrightarrow \mathbf{N}$ is the left unitor map.
\end{Definition}

Concretely, this amounts to equivariant maps $\mu_{p,q}: A_p \boxtimes^{\mathrm{ext}} N_q \longrightarrow N_{p+q}$ such that for all $a \in A_p$, $b \in A_q$ and $x \in N_r$, one has the associativity relation 
\[ \mu_{p+q,r}^{\boxtimes} \Big( m_{p,q}^{\boxtimes}(a,b), x \Big) = \mu_{p,q+r}^{\boxtimes} \Big( a, \mu_{q,r}^{\boxtimes}(b,x) \Big) \]
together with the unit relation 
\[ \forall x \in N_q, \qquad \mu_{0,q}^{\boxtimes} \big( \mathds{1}_0^{\boxtimes}, x \big) = x. \]

Since $\mathbf{A}$ is commutative, a left $(\mathbf{A}, \boxtimes)$-module structure induces a right action on $\mathbf{N}$, namely the equivariant map $\widetilde{\mu}^{\boxtimes}: \mathbf{N} \boxtimes \mathbf{A} \longrightarrow \mathbf{N}$ defined by $\widetilde{\mu}^{\boxtimes} = \mu^{\boxtimes} \circ \beta^{\boxtimes}_{\mathbf{N}, \mathbf{A}}$, where $\beta^{\boxtimes}$ is the $\boxtimes$-braiding. Componentwise, this reads
\[ \forall x \in N_p, \quad \forall a \in A_q, \qquad \widetilde{\mu}_{p,q}^{\boxtimes}(x,a) = \tau_{p,q} \cdot \mu_{q,p}^{\boxtimes}(a,x) \] 
where $\tau_{p,q}$ is the block permutation exchanging the two blocks of sizes $p$ and $q$. In other words, $\mathbf{N}$ has a structure of a $\boxtimes$-bimodule over $\mathbf{A}$.

\begin{Definition}
Let $\mathbf{N}$ be a left $(\mathbf{A},\boxtimes)$-module. \textbf{A $\boxtimes$-derivation} of $\mathbf{A}$ with values in $\mathbf{N}$ is an equivariant bundle map $D: \mathbf{A} \rightarrow \mathbf{N}$ such that 
\[ D \circ m^{\boxtimes} = \mu^{\boxtimes} \circ \big( \mathrm{Id}_{\mathbf{A}} \boxtimes D \big) + \widetilde{\mu}^{\boxtimes} \circ \big( D \boxtimes \mathrm{Id}_{\mathbf{A}} \big). \]
Denote by $\mathrm{Der}_{\boxtimes}(\mathbf{A}, \mathbf{N})$ the vector space of all $\boxtimes$-derivations of $\mathbf{A}$ with values in the $(\mathbf{A}, \boxtimes)$-module $\mathbf{N}$. 
\end{Definition}

Concretely, this means that there is a collection of equivariant maps $D_n: A_n \longrightarrow N_n$ such that 
\[ \forall a \in A_p, \quad \forall b \in A_q, \qquad D_{p+q} \Big( m_{p,q}^{\boxtimes}(a,b) \Big) = \mu_{p,q}^{\boxtimes} \Big( a, D_q(b) \Big) + \widetilde{\mu}_{p,q}^{\boxtimes} \Big( D_p(a), b \Big) = \mu_{p,q}^{\boxtimes} \big( a, D_q(b) \big) + \tau_{p,q} \cdot \mu_{q,p}^{\boxtimes} \big( b, D_p(a) \big). \]
When there is no ambiguity, the action map $\mu^{\boxtimes}(a,x)$ will be denoted $a \cdot x$. The $\boxtimes$-algebra $\mathbf{A}$ carries a natural $\boxtimes$-module structure given by its $\boxtimes$-multiplication. Thus, when $\mathbf{N} = \mathbf{A}$, this recovers the usual notion of a $\boxtimes$-derivation on $\mathbf{A}$. 

As in the ordinary case, any $\boxtimes$-derivation annihilates the $\boxtimes$-unit. Indeed, applying the Leibniz rule to $m_{0,0}^{\boxtimes} \Big( \mathds{1}_{0}^{\boxtimes}, \mathds{1}_0^{\boxtimes} \Big)  = \mathds{1}_0^{\boxtimes}$ gives 
\[ D_0 \big( \mathds{1}_0^{\boxtimes} \big) = 2 D_0 \big( \mathds{1}_0^{\boxtimes} \big). \]
Hence, $D_0 \big( \mathds{1}_0^{\boxtimes} \big) = 0$.

Finally, a $\boxtimes$-biderivation is a map $B: \mathbf{A} \boxtimes \mathbf{A} \longrightarrow \mathbf{N}$ which is a $\boxtimes$-derivation in each variable. We denote by $\mathrm{Bider}_{\boxtimes}(\mathbf{A}, \mathbf{N})$ the vector space of all $\boxtimes$-biderivations. In particular, a $\boxtimes$-Poisson bracket is a $\boxtimes$-biderivation.

When $\mathbf{A}$ is a free $\boxtimes$-algebra, there is an easy description of $\mathrm{Der}_{\boxtimes}(\mathbf{A}, \mathbf{N})$. 

\begin{Proposition} \label{lem28}
Let $\mathbf{W} \in \mathcal{VB}_{\mathfrak{S}_{\bullet}}(M^{\bullet})$ and $\mathbf{N}$ is a $\big( \mathbf{S}^{\boxtimes}(\mathbf{W}), \boxtimes \big)$-module, then restricting along the inclusion $\mathbf{W} \hookrightarrow \mathbf{S}^{\boxtimes}(\mathbf{W})$ induces a natural bijection 
\[ \mathrm{Der}_{\boxtimes} \Big( \mathbf{S}^{\boxtimes}(\mathbf{W}), \mathbf{N} \Big) \cong \mathrm{Mor}_{\mathfrak{S}_{\bullet}}(\mathbf{W}, \mathbf{N}). \]
In particular, 
\[ \mathrm{Bider}_{\boxtimes} \big( \mathbf{S}^{\boxtimes}(\mathbf{W}), \mathbf{N} \big) \cong \mathrm{Mor}_{\mathfrak{S}_{\bullet}} \big( \mathbf{W} \boxtimes \mathbf{W}, \mathbf{N} \big). \]
\end{Proposition}

\begin{proof} 
Let $D \in \mathrm{Der}_{\boxtimes} \Big( \mathbf{S}^{\boxtimes}(\mathbf{W}), \mathbf{N} \Big)$, then on a monomial $w_1 \boxdot \ldots \boxdot v_n$, one has 
\[ D(w_1 \boxdot \ldots \boxdot w_n) = \sum \limits_{i=1}^{n} (w_1 \boxdot \ldots \boxdot \widehat{w_i} \boxdot \ldots \boxdot w_n) \cdot D(w_i). \]
Therefore, $D$ is entirely determined by its values on $\mathbf{W}$, and the map $\mathrm{Der}_{\boxtimes} \big( \mathbf{S}^{\boxtimes}(\mathbf{W}), \mathbf{N} \big) \longrightarrow \mathrm{Mor}_{\mathfrak{S}_{\bullet}}(\mathbf{W}, \mathbf{N})$ is injective. Conversely, let $f \in \mathrm{Mor}_{\mathfrak{S}_{\bullet}}(\mathbf{W}, \mathbf{N})$. Define $D_f \in \mathrm{Der}_{\boxtimes} \big( \mathbf{S}^{\boxtimes}(\mathbf{W}), \mathbf{N} \big)$ by $D_f \big( \mathds{1}_0^{\boxtimes} \big) = 0$ and by 
\[ D_f(w_1 \boxdot \ldots \boxdot w_n) = \sum \limits_{i=1}^{n} (w_1 \boxdot \ldots \boxdot \widehat{w_i} \boxdot \ldots \boxdot w_n) \cdot f(w_i) \] 
and extend by linearity. Now if $v = v_1 \boxdot \ldots \boxdot v_p$ and $w = w_1 \boxdot \ldots \boxdot w_q$, then 
\[ \begin{array}{lll} D_f(v \boxdot w) &=& D_f \big( v_1 \boxdot \ldots \boxdot v_p \boxdot w_1 \boxdot \ldots \boxdot w_q \big) \\ 
&=& \sum \limits_{i=1}^{p} \big( v_1 \boxdot \ldots \boxdot \widehat{v_i} \boxdot \ldots \boxdot w \big) \cdot f(v_i) + \sum \limits_{j=1}^{q} \big( v \boxdot w_1 \boxdot \ldots \boxdot \widehat{w_j} \boxdot \ldots \boxdot w_q \big) \cdot f(w_j) \\
&=& \Big( \sum \limits_{i=1}^{p} \big( v_1 \boxdot \ldots \boxdot \widehat{v_i} \boxdot \ldots \boxdot v_p \big) \cdot f(v_i) \Big) \cdot w + v \cdot \Big( \sum \limits_{j=1}^{q} \big( w_1 \boxdot \ldots \boxdot \widehat{w_j} \boxdot \ldots \boxdot w_q \big) \cdot f(w_j) \Big) \\ 
&=& D_f(v) \cdot w + v \cdot D_f(w) \end{array} \] 
where we have used also the associativity of the $\boxtimes$-action. Therefore, $D_f$ is a $\boxtimes$-derivation extending $f$. So the map $\mathrm{Der}_{\boxtimes} \big( \mathbf{S}^{\boxtimes}(\mathbf{W}), \mathbf{N} \big) \longrightarrow \mathrm{Mor}_{\mathfrak{S}_{\bullet}}(\mathbf{W}, \mathbf{N})$.
\end{proof}

Similarly, one can define the notion of $\otimes$-module over a commutative $\otimes$-algebra and derivations with values in $\otimes$-modules. More precisely, if $\mathbf{A}$ is a commutative $\otimes$-algebra bundle, then $\mathbf{N}$ is a $(\mathbf{A}, \otimes)$-module if we have an action map $\mu^{\otimes}: \mathbf{A} \otimes \mathbf{N} \longrightarrow \mathbf{N}$ satisfying the usual associativity and unit axioms. An $\otimes$-derivation on $\mathbf{A}$ with values in $\mathbf{N}$ is an equivariant bundle map $D: \mathbf{A} \rightarrow \mathbf{N}$ such that 
\[ D \circ m^{\otimes} = \mu^{\otimes} \circ \big( \mathrm{Id}_{\mathbf{A}} \otimes D \big) + \widetilde{\mu}^{\otimes} \circ \big( D \otimes \mathrm{Id}_{\mathbf{A}} \big). \]
We will denote by $\mathrm{Der}_{\otimes}(\mathbf{A}, \mathbf{N})$ the vector space of all $\otimes$-derivations of $\mathbf{A}$ with values in the $(\mathbf{A}, \otimes)$-module $\mathbf{N}$. When $\mathbf{A}$ is a free $\otimes$-algebra, there is an easy description of $\mathrm{Der}_{\otimes}(\mathbf{A}, \mathbf{N})$, analogous to the previous result.
\begin{Lemma} \label{lem27}
Let $\mathbf{V} \in \mathcal{VB}_{\mathfrak{S}_{\bullet}}(M^{\bullet})$ and let $\mathbf{N}$ be a $\big( \mathbf{S}^{\otimes}(\mathbf{V}), \otimes \big)$-module, then restriction along the inclusion $V \hookrightarrow \mathbf{S}^{\otimes}(\mathbf{V})$ induces a natural bijection 
\[ \mathrm{Der}_{\otimes} \big( \mathbf{S}^{\otimes}(\mathbf{V}),  \mathbf{N}) \cong \mathrm{Mor}_{\mathfrak{S}_{\bullet}}(\mathbf{V}, \mathbf{N}). \] 
\end{Lemma}

\subsection{Equivariant Poisson $2$-Algebra Bundles over Configuration Spaces} \label{4.2}

A Poisson structure on a $2$-algebra bundle $\mathbf{P}$ should satisfy an additional compatibility condition with the second monoidal structure $\otimes$, analogous in spirit to the $\boxtimes$-Leibniz rule stated above which we have interpreted as a $\boxtimes$-derivation property. In order to formulate this $\otimes$-compatibility, we first describe the natural $\otimes$-module structures carried by the components $\mathbf{P}_{p+q}$.

Consider the left extension map 
\[ \ell: a \in \mathbf{P} \longmapsto m^{\boxtimes} \big( a, \mathds{1}^{\otimes} \big) \in \mathbf{P}. \]
That is, for $p,q \in \mathbb{N}$, we have an induced map $\ell_{p,q}$ given by
\[ \ell_{p,q}: a \in \mathbf{P}_p \longmapsto m_{p,q}^{\boxtimes} \big( a, \mathds{1}_{q}^{\otimes} \big) \in \mathbf{P}_{p+q}. \]
The map $\ell_{p,q}$ inserts $q$ dummy points, thus extending an element defined on a $p$-configuration to an element defined on a $(p+q)$-configuration.

\begin{Lemma} 
The map $\ell_{p,q}$ is a $\otimes$-algebra morphism between $\big(\mathbf{P}_p, m_p^{\otimes}, \mathds{1}_p^{\otimes} \big)$ and $\big( \mathbf{P}_{p+q}, m_{p+q}^{\otimes}, \mathds{1}_{p+q}^{\otimes} \big)$.
\end{Lemma}

\begin{proof} 
This follows directly from the interchange relations in a double monoid.
\end{proof}

Therefore, $\mathbf{P}_{p+q}$ is canonically a (left) $(\mathbf{P}_p, \otimes)$-module with the action
\[ \forall a \in \mathbf{P}_p, \quad \forall v \in \mathbf{P}_{p+q}, \qquad a \cdot v = m_{p+q}^{\otimes} \big( \ell_{p,q}(a), v \big). \]
Similarly, there is a right extension map $r: b \in \mathbf{P} \longmapsto m^{\boxtimes} \big(\mathds{1}^{\otimes}, b \big) \in \mathbf{P}$, or equivalently, a collection of maps $(r_{p,q})_{p,q \in \mathbb{N}}$ given by 
\[ r_{p,q}: b \in \mathbf{P}_q \longmapsto m_{p,q}^{\boxtimes} \big( \mathds{1}_{p}^{\otimes}, b \big) \in \mathbf{P}_{p+q} \]
which turns $\mathbf{P}_{p+q}$ canonically into a (right) $(\mathbf{P}_q, \otimes)$-module with the action 
\[ \forall b \in \mathbf{P}_q, \quad \forall v \in \mathbf{P}_{p+q}, \qquad v \cdot b = m_{p+q}^{\otimes} \big( v, r_{p,q}(b) \big). \]
Since, $m^{\otimes}$ is commutative, these left and right actions commute and hence, $\mathbf{P}_{p+q}$ is naturally a $\mathbf{P}_p$-$\mathbf{P}_q$-bimodule.

\begin{Definition} 
An \textbf{equivariant Poisson $2$-algebra bundle} is an equivariant $2$-algebra $\big( \mathbf{P}, m^{\boxtimes}, m^{\otimes} \big)$ with a Poisson bracket 
\[ B: \mathbf{P} \boxtimes \mathbf{P} \longrightarrow \mathbf{P} \] 
such that \begin{enumerate}
    \item $\mathbf{P}$ is an equivariant Poisson $\boxtimes$-algebra,
    \item $B$ satisfies the $\otimes$-Leibniz rule induced by the $\otimes$-module structures defined above, namely for each $p,q \in \mathbb{N}$, \begin{itemize}
    \item[$\bullet$] $\left\{ \cdot, \; \cdot \right\}_{p,q}$ is a derivation in the first variable with values in the left $\mathbf{P}_p$-module $\mathbf{P}_{p+q}$,
    \item[$\bullet$] $\left\{ \cdot, \; \cdot \right\}_{p,q}$ is a derivation in the second variable with values in the right $\mathbf{P}_q$-module $\mathbf{P}_{p+q}$.
    \end{itemize}
\end{enumerate}
\end{Definition}

Equivalently, the second condition can be written explicitly as follows: for all $a,b \in \mathbf{P}_p$ and $c \in \mathbf{P}_q$,
\[ \begin{array}{lll} \left\{ m_p^{\otimes}(a,b), c \right\}_{p,q} &=& m_{p+q}^{\otimes} \Big( \left\{ a,c \right\}_{p,q}, \ell_{p,q}(b) \Big) + m_{p+q}^{\otimes} \Big( \ell_{p,q}(a), \left\{ b,c \right\}_{p,q} \Big) \\ 
&=& m_{p+q}^{\otimes} \Big( \left\{ a,c \right\}_{p,q}, m_{p,q}^{\boxtimes} \big(b, \mathds{1}_q^{\otimes} \big) \Big) + m_{p+q}^{\otimes} \Big( m_{p,q}^{\boxtimes} \big(a, \mathds{1}_q^{\otimes} \big), \left\{ b,c \right\}_{p,q} \Big). \end{array} \]
Similarly, for all $a \in \mathbf{P}_p$ and $b,c \in \mathbf{P}_q$,
\[ \left\{ a, m_{q}^{\otimes}(b,c) \right\}_{p,q} = m_{p+q}^{\otimes} \Big( \left\{ a,b \right\}_{p,q}, m_{p,q}^{\boxtimes} \big( \mathds{1}_{p}^{\otimes}, c \big) \Big) + m_{p+q}^{\otimes} \Big( \left\{a,c \right\}_{p,q}, m_{p,q}^{\boxtimes} \big( \mathds{1}_p^{\otimes}, b \big) \Big). \]

\begin{Lemma} 
A Poisson bracket automatically vanishes on $\otimes$-units, that is for all $p,q \in \mathbb{N}$, we have 
\[ \forall a \in \mathbf{P}_p, \quad \forall b \in \mathbf{P}_q, \qquad \left\{ a, \mathds{1}_q^{\otimes} \right\}_{p,q} = \left\{ \mathds{1}_p^{\otimes}, b \right\}_{p,q} = 0. \]
\end{Lemma}

\begin{proof} 
We compute 
\[ \begin{array}{lll} \left\{ \mathds{1}_p^{\otimes}, b \right\}_{p,q} &=& \left\{ m_{p}^{\otimes} \big( \mathds{1}_p^{\otimes}, \mathds{1}_p^{\otimes} \big), b \right\}_{p,q} \\ 
&=& m_{p+q}^{\otimes} \Big( \left\{ \mathds{1}_p^{\otimes}, b \right\}_{p,q}, \underbrace{\ell_{p,q}(\mathds{1}_p^{\otimes})}_{=\mathds{1}_{p+q}^{\otimes}} \Big) + m_{p+q}^{\otimes} \Big( \underbrace{m_{p,q}^{\boxtimes} \big( \mathds{1}_p^{\otimes}, \mathds{1}_q^{\otimes} \big)}_{=\mathds{1}_{p+q}^{\otimes}}, \left\{ \mathds{1}_p^{\otimes}, b \right\}_{p,q} \Big) \\ 
&=& 2 \left\{ \mathds{1}_p^{\otimes}, b \right\}_{p,q}. \end{array} \]
Thus, $\left\{ \mathds{1}_p^{\otimes}, b \right\}_{p,q} = 0$.
\end{proof}

The following lemmas will be used in the proof of our main theorem. Although they are analogous to standard results, the presence of two distinct monoidal structures requires a precise formulation. We fix a local vector bundle $V \longrightarrow M$ and set $\mathbf{A} = \mathbf{S}^{\otimes}(V)$.

\begin{Theorem} \label{free_poisson}
Let $V$ be a local vector bundle and $k: V \boxtimes V \longrightarrow \mathbf{I}_{\otimes}$ a skew-symmetric bundle map. Then, $k$ induces a unique Poisson bracket $B$ on $\mathbf{P} := \mathbf{S}^{\boxtimes} \big( \mathbf{S}^{\otimes} V \big)$ such that $B$ extends $k$, i.e. the following diagram commutes 
\[ \xymatrix{ V \boxtimes V \ar[rr]^{k} \ar[d]_{\iota \boxtimes \iota} && \mathbf{I}_{\otimes} \ar[d]^{u^{\otimes}} \\
\mathbf{P} \boxtimes \mathbf{P} \ar[rr]^{B} && \mathbf{P} } \]
where $\iota: V \rightarrow \mathbf{P}$ is the inclusion and $u^{\otimes}: \mathbf{I}_{\otimes} \longrightarrow \mathbf{P}$ is the $\otimes$-unit. 
\end{Theorem}

\begin{proof} 
Set $\mathbf{A} = \mathbf{S}^{\otimes}(V)$, so that $\mathbf{P} = \mathbf{S}^{\boxtimes}(\mathbf{A})$. Since $V$ is local, so is $\mathbf{A} = \mathbf{S}^{\otimes}(V)$. Therefore, $\mathbf{P}_2 = \mathbf{A} \boxtimes^{\mathrm{ext}} \mathbf{A}$. Now $\mathbf{P}_2$ carries two $(\mathbf{A}, \otimes)$-module structures, one in each variable 
\[ \forall a \in \mathbf{A}, \quad \forall u,v \in \mathbf{P}_2, \qquad a \cdot (v \otimes w) = (a \odot u) \otimes w \]
and
\[ \forall b \in \mathbf{A}, \quad \forall u,v \in \mathbf{P}_2, \qquad (v \otimes w) \cdot b = u \otimes (w \odot b). \]
We compose $k$ with the $\otimes$-unit $u^{\otimes} \circ k: V \boxtimes V \longrightarrow \mathbf{P}_2$ sending $v \boxtimes w$ to $k(v,w) \cdot \big( \mathds{1}^{\otimes} \otimes \mathds{1}^{\otimes} \big)$. We now extend in each variable separately. Fix $w \in V$, the map $v \in V \longmapsto k(v,w) \cdot \big( \mathds{1}^{\otimes} \otimes \mathds{1}^{\otimes} \big)$ is linear. Since, $\mathbf{A} = \mathbf{S}^{\otimes}(V)$ is the free commutative algebra generated by $V$, Lemma \ref{lem27} gives a unique $\otimes$-derivation $\mathbf{A} \longrightarrow \mathbf{P}_2$. Similarly in the second variable, Lemma \ref{lem27} extends $u^{\otimes} \circ k: V \boxtimes V \longrightarrow \mathbf{P}_2$ into a $\otimes$-biderivation 
\[ b: \mathbf{A} \boxtimes \mathbf{A} \longrightarrow \mathbf{P}_2 \subset \mathbf{P}. \]
Explicitly, on monomials of $\mathbf{A}$, one has 
\[ b \big( v_1 \odot \ldots \odot v_p, w_1 \odot \ldots \odot w_q \big) = \sum \limits_{1 \leqslant i,j \leqslant p,q} k(v_i, w_j) \cdot \big( v_1 \odot \ldots \odot \widehat{v_i} \odot \ldots \odot v_p \big) \otimes \big( w_1 \odot \ldots \odot \widehat{w_j} \odot \ldots \odot w_q \big). \]
Since $\mathbf{P} = \mathbf{S}^{\boxtimes}(\mathbf{A})$ is the free commutative $\boxtimes$-algebra generated by $\mathbf{A}$, Proposition \ref{lem28} extends $b: \mathbf{A} \boxtimes \mathbf{A} \longrightarrow \mathbf{P}$ to a $\boxtimes$-biderivation 
\[ B: \mathbf{P} \boxtimes \mathbf{P} \longrightarrow \mathbf{P}. \]
It remains to verify that $B$ is skew-symmetric and satisfies the $\otimes$-Leibniz rule and the Jacobi identity. For skew-symmetry, define $B' = - B \circ \beta^{\boxtimes}$. Then $B'$ is also a $\boxtimes$-derivation $\mathbf{P} \boxtimes \mathbf{P} \longrightarrow \mathbf{P}$ since $\beta^{\boxtimes}$ is a $\boxtimes$-algebra isomorphism and $B$ is a $\boxtimes$-derivation. By uniqueness in Proposition \ref{lem28}, it suffices to check that $B'$ coincides with $b$ on $\mathbf{A} \boxtimes \mathbf{A}$. This is immediately verified using the explicit expression of $b$ and skew-symmetry of $k$.

For the $\otimes$-Leibniz rule, fix $p,q \in \mathbb{N}$, $a \in \mathbf{P}_p$ and $b,c \in \mathbf{P}_q$, and consider the defect in the second variable: 
\[ \begin{array}{lll} \Delta_{p,q}(a,b,c) &=& \left\{ a, b \odot c \right\}_{p,q} - \left\{ a,b \right\}_{p,q} \odot r_{p,q}(c) - \left\{ a,c \right\}_{p,q} \odot r_{p,q}(b) \\ 
&=& \left\{ a, b \odot c \right\}_{p,q} - \left\{ a,b \right\}_{p,q} \odot (\mathds{1}_p^{\otimes} \boxdot c) - \left\{ a,c \right\}_{p,q} \odot (\mathds{1}_p^{\otimes} \boxdot b). \end{array} \]
Since the brackets satisfy the $\boxtimes$-Leibniz rule, it is immediate that $\Delta$ is a $\boxtimes$-derivation in the first argument. Hence, by Proposition \ref{lem28}, for fixed $b,c \in \mathbf{P}_q$, it suffices to show that the map $a \longmapsto \Delta(a,b,c)$ vanishes on $\mathbf{A} = \mathbf{S}^{\otimes}(V)$.
To encode the condition that both $b$ and $c$ must lie in the same configuration, consider the diagonal $\boxtimes$-algebra: 
\[ \mathbf{Q} = \mathbf{P} \otimes \mathbf{P} = \bigoplus \limits_{q \in \mathbb{N}} \mathbf{P}_q \otimes \mathbf{P}_q \]
equipped with the diagonal $\boxtimes$-product
\[ (b_1, c_1) \boxdot (b_2, c_2) = \big( b_1 \boxdot b_2, c_1 \boxdot c_2 \big). \]
The map $\mu: (b,c) \in \mathbf{Q} \longmapsto b \odot c \in \mathbf{P}$ is a $\boxtimes$-algebra morphism by the interchange law. Thus $\mathbf{P}$ becomes a $(\mathbf{Q}, \boxtimes)$-module via
\[ x \cdot (b,c) = x \boxdot \mu(b,c) = x \boxdot (b \odot c). \]
Now define, for fixed $a$, the map $\Delta_a = \Delta(a, \cdot, \cdot): \mathbf{Q} \longrightarrow \mathbf{P}$. We claim that this map is a $\boxtimes$-derivation. A direct computation gives
\[ \begin{array}{lll} \Delta_a \big( (b_1, c_1) \boxdot (b_2, c_2) \big) &=& \Delta_a \big( b_1 \boxdot b_2, c_1 \boxdot c_2 \big) \\ 
&=& \left\{ a, \big( b_1 \boxdot b_2 \big) \odot \big( c_1 \boxdot c_2 \big) \right\} \\
&& - \left\{ a, b_1 \boxdot b_2 \right\} \odot \big( \mathds{1}^{\otimes} \boxdot (c_1 \boxdot c_2) \big) \\ 
&& - \left\{ a, c_1 \boxdot c_2 \right\} \odot \big( \mathds{1}_p \boxdot (b_1 \boxdot b_2) \big). \end{array} \]
Using successively the interchange law and the $\boxtimes$-Leibniz rule, the first term becomes
\[ \begin{array}{lll} \left\{ a, \big( b_1 \boxdot b_2 \big) \odot \big( c_1 \boxdot c_2 \big) \right\} &=& \left\{ a, \big( b_1 \odot c_1 \big) \boxdot \big( b_2 \odot c_2 \big) \right\} \\
&=& \left\{ a, b_1 \odot c_1 \right\} \boxdot (b_2 \odot c_2) + (b_1 \odot c_1) \boxdot \left\{ a, b_2 \odot c_2 \right\}. \end{array} \]
Likewise, the second term becomes
\[ \begin{array}{lll} \left\{ a, b_1 \boxdot b_2 \right\} \odot \big( \mathds{1}^{\otimes} \boxdot (c_1 \boxdot c_2) \big) 
&=& \Big( \left\{ a,b_1 \right\} \boxdot b_2 \Big) \odot \big( (\mathds{1}^{\otimes} \boxdot c_1) \boxdot c_2 \big) + \Big( b_1 \boxdot \left\{ a, b_2 \right\} \Big) \odot \big( (\mathds{1}^{\otimes} \boxdot c_1) \boxdot c_2 \big) \\ 
&=& \Big( \left\{ a, b_1 \right\} \odot \big( \mathds{1}^{\otimes} \boxdot c_1 \big) \Big) \boxdot (b_2 \odot c_2) 
+ (b_1 \odot c_1) \boxdot \Big( \left\{ a, b_2 \right\} \odot \big( \mathds{1}^{\otimes} \boxdot c_2 \big) \Big).
\end{array} \]
Similarly, the third term becomes
\[ \begin{array}{lll} \left\{ a, c_1 \boxdot c_2 \right\} \odot \big( \mathds{1}^{\otimes} \boxdot (b_1 \odot b_2) \big) &=& \Big( \left\{ a, c_1 \right\} \odot \big( \mathds{1}^{\otimes} \boxdot b_1 \big) \Big) \boxdot (c_2 \odot b_2) + (c_1 \boxdot b_1) \boxdot \Big( \left\{ a, c_2 \right\} \odot \big( \mathds{1}^{\otimes} \boxdot b_2 \big) \Big) \\ 
&=& \Big( \left\{ a, c_1 \right\} \odot \big( \mathds{1}^{\otimes} \boxdot b_1 \big) \Big) \boxdot (b_2 \odot c_2) + (b_1 \odot c_1) \boxdot \Big( \left\{ a, c_2 \right\} \odot \big( \mathds{1}^{\otimes} \boxdot b_2 \big) \Big). \end{array} \]
Therefore, 
\[ \begin{array}{lll} \Delta_a \big( (b_1, c_1) \boxdot (b_2, c_2) \big) &=& (b_2 \odot c_2) \boxdot \Big( \left\{ a, b_1 \odot c_1 \right\} - \left\{ a,b_1 \right\} \odot \big( \mathds{1}^{\otimes} \boxdot c_1 \big) - \left\{ a, c_1 \right\} \odot \big( \mathds{1}^{\otimes} \boxdot b_1 \big) \Big) \\ 
&& + (b_1 \odot c_1) \boxdot \Big( \left\{ a, b_2 \odot c_2 \right\} - \left\{ a,b_2 \right\} \odot \big( \mathds{1}^{\otimes} \boxdot c_2 \big) - \left\{ a, c_2 \right\} \odot \big( \mathds{1}^{\otimes} \boxdot b_2 \big) \Big) \\ 
&=& \Delta_a(b_1, c_1) \boxdot (b_2 \odot c_2) + (b_1 \odot c_1) \boxdot \Delta_a(b_2, c_2). \end{array} \]
Hence, $\Delta_a$ is indeed a $\boxtimes$-derivation on $\mathbf{Q}$. Since $V$ is local, so is $\mathbf{A} = \mathbf{S}^{\otimes}(V)$, and therefore $\mathbf{Q}$ is the free commutative $\boxtimes$-algebra generated by $\mathbf{A} \otimes \mathbf{A}$. It follows that it suffices to prove 
\[ \forall a,b,c \in \mathbf{A}, \qquad \Delta_{1,1}(a,b,c) = 0. \]
But this holds by construction of $b = \left\{ \cdot, \; \cdot \right\}_{1,1}$, which satisfies the $\otimes$-Leibniz rule.

We now verify the Jacobi identity. Consider the jacobiator 
\[ J_{p,q,r}(a,b,c) = \left\{ a, \left\{ b,c \right\}_{q,r} \right\}_{p,q+r} + \left\{ b, \left\{ c, a \right\}_{r,p} \right\}_{q, r+p} + \left\{ c, \left\{ a,b \right\}_{p,q} \right\}_{r,p+q} \in \mathbf{P}_{p+q+r}\]
where $a \in \mathbf{P}_p$, $b \in \mathbf{P}_q$ and $c \in \mathbf{P}_r$. Note that $J$ is a $\boxtimes$-triderivation (i.e., a $\boxtimes$-derivation in each of its three variables), hence it is determined by its values on $\mathbf{A} \boxtimes \mathbf{A} \boxtimes \mathbf{A}$. Thus it suffices to prove 
\[ \forall a,b,c \in \mathbf{A}, \qquad J_{1,1,1}(a,b,c) = 0. \]
By construction, $J_{1,1,1}$ is a $\otimes$-triderivation (i.e., a $\otimes$-derivation in each of its three variables), so it suffices to show that $J_{1,1,1}$ vanishes on $V \boxtimes V \boxtimes V$. Recall that if $a,b,c \in V$, then 
\[ \left\{ b,c \right\}_{1,1} = k(b,c) \big( \mathds{1}^{\otimes}_1 \boxdot \mathds{1}^{\otimes}_1 \big) = k(b,c) \mathds{1}^{\otimes}_2. \]
Therefore, 
\[ \left\{ a, \left\{ b,c \right\}_{1,1} \right\}_{1,2} = k(b,c) \left\{ a, \mathds{1}^{\otimes}_2 \right\}_{1,2} = 0 \]
since the brackets vanish on units. Therefore, $J = 0$ on $\mathbf{P}$, and the Jacobi identity holds.
\end{proof}

\begin{Remark}[Computation rules] On $\mathbf{P} = \mathbf{S}^{\boxtimes} \big( \mathbf{S}^{\otimes}(V) \big)$, there is a well-defined bracket $B: \mathbf{P} \boxtimes \mathbf{P} \longrightarrow \mathbf{P}$ characterized by the following properties:
\begin{enumerate}
    \item The bracket vanishes on units, that is 
    \[ \forall a \in \mathbf{P}_p, \quad \forall b \in \mathbf{P}_q, \qquad \left\{ a, \mathds{1}_q^{\otimes} \right\}_{p,q} = \left\{ \mathds{1}_p^{\otimes}, b \right\}_{p,q} = 0, \]
    \item The bracket extends $k$, that is 
    \[ \forall a,b \in \mathbf{P}_1, \qquad \left\{ a,b \right\}_{1,1} = k(a,b) \big( \mathds{1}_1^{\otimes} \boxdot \mathds{1}_1^{\otimes} \big) = k(a,b) \mathds{1}^{\otimes}_2, \]
    \item The bracket satisfies the $\otimes$-Leibniz rule. More precisely, for all $a,b \in \mathbf{P}_p$ and $c \in \mathbf{P}_q$,
    \[ \left\{ a \odot b, c \right\}_{p,q} = \left\{ a,c \right\}_{p,q} \odot \big(b \boxdot \mathds{1}_q^{\otimes} \big) + \big( a \boxdot \mathds{1}_q^{\otimes} \big) \odot \left\{ b,c \right\}_{p,q}. \]
    Likewise, for all $a \in \mathbf{P}_p$ and $b,c \in \mathbf{P}_q$,
    \[ \left\{ a, b \odot c \right\}_{p,q} = \left\{ a,b \right\}_{p,q} \odot \big( \mathds{1}_p^{\otimes} \boxdot c \big) + \big( \mathds{1}_p^{\otimes} \boxdot b \big) \odot \left\{ a,c \right\}_{p,q}, \]
    \item The bracket satisfies the $\boxtimes$-Leibniz rule. That is, for all $a \in \mathbf{P}_p$, $b \in \mathbf{P}_q$ and $c \in \mathbf{P}_r$,
    \[ \left\{ a \boxdot b, c \right\}_{p+q, r} = \left\{ a,c \right\}_{p,r} \boxdot b + a \boxdot \left\{ b,c \right\}_{q,r}, \]
    \[ \left\{ a, b \boxdot c \right\}_{p,q+r} = \left\{ a,b \right\}_{p,q} \boxdot c + b \boxdot \left\{ a,c \right\}_{p,r}. \]
\end{enumerate}
\end{Remark}

\subsection{Examples} 

\subsubsection{Trivial bundles over $\mathbb{R}$}

We now illustrate the constructions of the previous sections in the simplest concrete setting. Recall from the introduction that a classical field on a spacetime manifold $M$ is a section $\varphi \in \Gamma(E)$ of a vector bundle $E \rightarrow M$. A polynomial local observable of degree $m$ at a point $x \in M$ is described by a kernel $f_m(x) \in S^{\otimes m}(E^*)_x$, which pairs with the field configuration via $\big \langle f_m(x), \varphi(x)^{\otimes m} \big \rangle \in \mathbb{R}$. More generally, a multilocal observable of degree $(m_1, \ldots, m_n)$ at a configuration $(x_1, \ldots, x_n) \in M^n$ is described by a kernel $f_{m_1}(x_1) \otimes \ldots \otimes f_{m_n}(x_n)$ where each $f_{m_j} \in S^{\otimes m_j}(E^*)_{x_j}$, which pairs to $\prod \limits_{j=1}^{n} \big \langle f_{m_j}(x_j), \varphi(x_j)^{\otimes m_j} \big \rangle$. The full space of such kernels over all configurations is precisely the Poisson $2$-algebra bundle $\mathbf{P} = \mathbf{S}^{\boxtimes} \big( \mathbf{S}^{\otimes}(V) \big)$ constructed in the previous sections, where $V = E^*$. The Hadamard product $\odot$ encodes pointwise polynomial multiplication of kernels at a given point, and the Cauchy product $\boxdot$ encodes the concatenation of independent observations at separate points. The passage from kernels to actual observables requires a further step of integration against densities and, along the diagonals, the renormalization of singular products of distributions, which is future work.

In what follows, we take $M = \mathbb{R}$ and $E = M \times F$ where $F$ is a finite-dimensional vector space of dimension $r$. The dual bundle is $V = E^* = M \times F^*$. Fix a basis $(e_1, \ldots, e_r)$ of $F$, then write $(\xi^1, \ldots, \xi^r)$ for the dual basis in $F^*$, viewed as coordinate functions on $F$. For a field $\varphi: \mathbb{R} \longrightarrow F$, its components are $\varphi^a = \xi^a \circ \varphi$. Note also that since $M=\mathbb{R}$ is orientable, the integration measure may be chosen canonically (up to a positive scalar) as the standard Lebesgue measure $\mathrm{d} t$. The treatment of densities mentioned in the introduction therefore does not intervene in this example.

The Hadamard symmetric algebra at a point $t \in M$ is 
\[ \mathbf{S}^{\otimes}(V)_t = S(F^*) = \mathbb{R} \big[ \xi^1, \ldots, \xi^r \big] \]
the algebra of polynomial functions on the fibre $F$. A monomial $\big(\xi^{a_1} \odot \ldots \odot \xi^{a_m} \big)_t$ is the kernel representing the local observable $\varphi^{a_1}(t) \ldots \varphi^{a_m}(t)$. Since $V$ is local, there is a canonical identification
\[ \mathbf{P}_n \cong \mathbf{S}^{\otimes}(V)^{\boxtimes^{\mathrm{ext}} n} \]
as $\mathfrak{S}_n$-equivariant vector bundles over $\mathbb{R}^n$. Its fibre at a point $(t_1, \ldots, t_n) \in \mathbb{R}^n$ is 
\[ \mathbf{P}_{(t_1, \ldots, t_n)} = \mathbb{R} \big[ \xi_{t_1}^{1}, \ldots, \xi_{t_1}^{r} \big] \otimes_{\mathbb{R}} \ldots \otimes_{\mathbb{R}} \mathbb{R} \big[ \xi_{t_n}^{1}, \ldots, \xi_{t_n}^{r} \big] \]
with the $\mathfrak{S}_n$-action permuting the factors together with the base points. 

The rank drop phenomenon is most naturally seen at the level of equivariant smooth sections. An equivariant smooth section is a smooth map $f: \mathbb{R}^n \longrightarrow S(F^*)^{\otimes n}$ such that 
\[ \forall \sigma \in \mathfrak{S}_n, \qquad  f \big( t_{\sigma^{-1}(1)}, \ldots, t_{\sigma^{-1}(n)} \big) = \sigma \cdot f(t_1, \ldots, t_n). \]
Hence if $x \in \mathbb{R}^n$ has isotropy subgroup $G_x \subset \mathfrak{S}_n$, then the value of an equivariant section at $x$ must lie in the fixed subspace $\big( S(F^*)^{\otimes n} \big)^{G_x}$. For instance, in degree $2$, the fibre at a generic point $(t_1, t_2)$ with $t_1 \neq t_2$ is $\mathbb{R} \big[ \xi_{t_1}^1, \ldots, \xi_{t_1}^r \big] \otimes \mathbb{R} \big[ \xi_{t_2}^1, \ldots, \xi_{t_2}^r \big]$, a polynomial algebra in $2r$ variables. Above a diagonal point $(t,t) \in \mathbb{R}^2$, the transposition $\sigma = \begin{pmatrix} 1&2 \end{pmatrix} \in \mathfrak{S}_2$ exchanges the two tensor factors, therefore equivariant sections must satisfy
\[ \sigma \cdot (f \otimes g)(t,t) = (g \otimes f)(t,t) \]
so their values lie in the invariant subspace $\big( S(F^*) \otimes_{\mathbb{R}} S(F^*) \big)^{\mathfrak{S}_2}$. In total polynomial degree $1$, the generic fibre is $2r$-dimensional with basis $\big( \xi_{t_1}^a \otimes 1, 1 \otimes \xi_{t_2}^a \big)_{1 \leqslant a \leqslant r}$, while the allowed values of equivariant sections along the diagonal form only the $r$-dimensional subspace spanned by $\big( \xi_{t_1}^a \otimes 1 + 1 \otimes \xi_{t_2}^a \big)_{1 \leqslant a \leqslant r}$. This is the rank drop phenomenon discussed in the introduction.

Over a fixed configuration $(t_1, \ldots, t_n)$, the Hadamard product multiplies polynomials fibrewise at each site: 
\[ (f_1 \otimes \ldots \otimes f_n) \odot (g_1 \otimes \ldots \otimes g_n) = (f_1 \odot g_1) \otimes \ldots \otimes (f_n \odot g_n) \]
where $f_i \odot g_i$ denotes the polynomial product in the fibre algebra $\mathbf{S}^{\otimes}(V)_{t_i}$. The Cauchy product concatenates configurations: if $f$ lies over $(t_1, \ldots, t_p)$ and $g$ lies over $(t_{p+1}, \ldots, t_{p+q})$, then $f \boxdot g$ lies over the concatenated configuration $(t_1, \ldots, t_{p+q})$. The Hadamard product corresponds to multiplying local polynomial kernels at the same spacetime point - for instance, $(\xi^a \odot \xi^b)_t$ is the kernel of the degree-$2$ observable $\varphi^a(t) \varphi^b(t)$ - while the Cauchy product concatenates kernels at different points into a multilocal kernel - for instance, $(\xi^a)_{t_1} \boxdot (\xi^b)_{t_2}$ is the kernel of the bilocal observable $\varphi^a(t_1) \varphi^b(t_2)$.

A bundle map $k: V \boxtimes V \rightarrow \mathbf{I}_{\otimes}$ is determined by smooth functions $k^{ab}: \mathbb{R}^2 \rightarrow \mathbb{R}$ for $1 \leqslant a,b \leqslant r$. The skew-symmetry of $k$ imposes 
\[ k^{ab}(t_1, t_2) = - k^{ba}(t_2, t_1). \]
In particular, notice that the matrix $\big( k^{ab}(t,t) \big)_{1 \leqslant a,b \leqslant r}$ is a skew-symmetric bilinear form on $F^*$. By Theorem \ref{free_poisson}, this kernel induces a Poisson bracket $B: \mathbf{P} \boxtimes \mathbf{P} \rightarrow \mathbf{P}$ determined on generators by 
\[ \left\{ \xi^a_{t_1}, \xi^b_{t_2} \right\}_{1,1} = k^{ab}(t_1, t_2) \cdot \mathds{1}^{\otimes}_{(t_1,t_2)} \in \mathbf{P}_{(t_1,t_2)}. \]
Iterating the $\otimes$-Leibniz rule yields the following formula. 
\begin{Corollary} \label{poisson_mech}
For $f(t_1) \in \mathbf{S}^{\otimes}(V)_{t_1}$ and $g(t_2) \in \mathbf{S}^{\otimes}(V)_{t_2}$, the elementary Poisson bracket is 
\[ \left\{ f(t_1), g(t_2) \right\}_{1,1} = \sum \limits_{1 \leqslant a,b \leqslant r} k^{ab}(t_1, t_2) \frac{\partial f}{\partial \xi^a}(t_1) \otimes \frac{\partial g}{\partial \xi^b}(t_2) \]
where the partial derivatives are the usual polynomial derivatives in the fibre coordinates. 
\end{Corollary}
The $\boxtimes$-Leibniz rule extends the bracket to multilocal observables. If $f(t_1, \ldots, t_n) = f_1(t_1) \boxdot \ldots \boxdot f_p(t_p)$ is a pure $\boxtimes$-tensor with $f_i(t_i) \in \mathbf{S}^{\otimes}(V)_{t_i}$ and $g(s) \in \mathbf{S}^{\otimes}(V)_{s}$. Then 
\[ \left\{ f_1(t_1) \boxdot \ldots \boxdot f_p(t_p), g(s) \right\}_{p,1} = \sum \limits_{i=1}^{p} f_1(t_1) \boxdot \ldots \boxdot f_{i-1}(t_{i-1}) \boxdot \left\{ f_i(t_i), g(s) \right\}_{1,1} \boxdot f_{i+1}(t_{i+1}) \boxdot \ldots \boxdot f_p(t_p) \]
where each summand lies over the $(p+1)$-point configuration obtained by reindexing the point $t_i$ with the pair $(t_i, s)$, and $\left\{ f_i(t_i), g(s) \right\}_{1,1} \in \mathbf{P}_{(t_i, s)}$ is computed by Corollary \ref{poisson_mech}. Thus the bracket acts as a derivation of the $\boxtimes$-product, introducing one additional base point at each step. 

\begin{Example} 
For $r=1$ and $f(t_1, t_2) = \xi_{t_1} \boxdot \xi_{t_2}$ and $g(s) = \xi_s$, 
\[ \left\{ f(t_1, t_2), g(s) \right\}_{2,1} = k(t_1, s) \cdot \mathds{1}_{(t_1,s)}^{\otimes} \boxdot \xi_{t_2} + \xi_{t_1} \boxdot k(t_2,s) \cdot \mathds{1}_{(t_2,s)}^{\otimes}. \]
The first term lies in $\mathbf{P}_{(t_1,s,t_2)}$ and the second in $\mathbf{P}_{(t_1,t_2,s)}$. These are distinct summands in the unshuffle decomposition of $\mathbf{P}_3$, reflecting positions at which the new base point $s$ is inserted into the configuration.
\end{Example}

\subsubsection{Classical mechanics}

We now specialize to the case of classical mechanics. Let $Q$ be a real vector space of dimension $d$ and set $F = Q \oplus Q^*$ so that $r=2d$. The field bundle is $E = M \times F$ and its dual is $V = E^* = M \times F^*$. Fix a basis $(e_1, \ldots, e_d)$ of $Q$ with dual basis $(e^1, \ldots, e^d)$ of $Q^*$, and write the corresponding fibre coordinates of $F^*$ as $(q^1, \ldots, q^d, p_1, \ldots, p_d)$ where $q^i$ is the coordinate function dual to $e_i$ and $p_j$ is dual to $e_j$. A field $\varphi: \mathbb{R} \longrightarrow Q \oplus Q^*$ has components 
\[ \varphi (t) = \big( q^1(t), \ldots, q^d(t), p_1(t), \ldots, p_d(t) \big). \]
The vector space $F = Q \oplus Q^*$ carries a canonical symplectic form
\[ \forall \alpha, \beta \in Q^*, \quad \forall v,w \in Q, \qquad \omega \big( (v, \alpha), (w, \beta) \big) = \beta(v) - \alpha(w). \]
In the ordered basis $\big( (e_1,0), \ldots, (e_d,0), (0, e^1), \ldots, (0,e^d) \big)$ of $F$, the matrix of $\omega$ is
\[ \big( \omega^{ab} \big)_{1 \leqslant a,b \leqslant 2d} = \begin{pmatrix} 0 & I_d \\ - I_d & 0 \end{pmatrix}. \] 
Since $\omega$ is a constant skew-symmetric form on $F$, we may define a skew-kernel by 
\[ k^{ab}(t_1, t_2) = G(t_1, t_2) \cdot \omega^{ab} \]
where $G: \mathbb{R}^2 \rightarrow \mathbb{R}$ is any smooth symmetric function, for instance $G=1$. The elementary Poisson bracket are 
\[ \left\{q^i(t_1), p_j(t_2) \right\}_{1,1} = G(t_1,t_2) \delta_j^i \cdot \mathds{1}_{(t_1,t_2)}^{\otimes} \qquad \left\{ q^i(t_1), q^j(t_2) \right\}_{1,1} = 0, \qquad \left\{ p_i(t_1), p_j(t_2) \right\}_{1,1} = 0. \]
For general polynomials $f \in \mathbb{R} \big[ q_{t_1}^1, \ldots, q_{t_1}^d, p_{1;t_1}, \ldots, p_{d;t_1} \big]$ and $g \in \mathbb{R} \big[ q_{t_2}^1, \ldots, q_{t_2}^d, p_{1;t_2}, \ldots, p_{d;t_2} \big]$, one has 
\[ \left\{f(t_1), g(t_2) \right\}_{1,1} = G(t_1, t_2) \sum \limits_{i=1}^{d} \Big( \frac{\partial f}{\partial q^i}(t_1) \otimes \frac{\partial g}{\partial p_i}(t_2) - \frac{\partial f}{\partial p_i}(t_1) \otimes \frac{\partial g}{\partial q^i}(t_2) \Big). \]

\begin{Example} 
Take $d=1$ and $G=1$. For a single degree of freedom $F = \mathbb{R}^2$ with coordinates $(q,p)$. Some elementary brackets are 
\[ \left\{q^2(t_1), p(t_2) \right\}_{1,1} = 2  q(t_1) \otimes \mathds{1}_{t_2}^{\otimes}, \qquad \left\{ q(t_1) p(t_1), q(t_2) p(t_2) \right\}_{1,1} = p(t_1) \otimes q(t_2) - q(t_1) \otimes p(t_2) \]
and 
\[ \left\{ q(t_1)^2 p(t_1), q(t_2) p(t_2)^2 \right\}_{1,1} = 2 q(t_1) p(t_1) \otimes 2 q(t_2) p(t_2) - q(t_1)^2 \otimes p(t_2)^2 = 4 q(t_1)p(t_1) \otimes q(t_2) p(t_2) - q(t_1)^2 \otimes p(t_2)^2. \]
\end{Example}

\begin{Remark} 
The bracket $B: \mathbf{P}_1 \boxtimes \mathbf{P}_1 \longrightarrow \mathbf{P}_2$ maps local observables to bilocal observables. To recover the usual Poisson bracket of classical mechanics, one must pull back along the diagonal and then compose with the Hadamard multiplication. More precisely, one composes the elementary bracket with the diagonal pullback $\Delta_2^*: \mathbf{P}_2 \rightarrow \Delta_2^* \mathbf{P}_2 \cong \mathbf{S}^{\otimes}(V) \otimes \mathbf{S}^{\otimes}(V)$ and applies the $\otimes$-multiplication within a single fibre. This gives 
\[ m^{\otimes} \circ \Delta_2^* \Big( \left\{f,g \right\} \Big) = \sum \limits_{i=1}^{d} \Big( \frac{\partial f}{\partial q^i} \frac{\partial g}{\partial p_i} - \frac{\partial f}{\partial p_i} \frac{\partial g}{\partial q^i} \Big) \]
which is the standard canonical Poisson bracket on $\mathbb{R}\big[ q^1, \ldots, q^d, p_1, \ldots, p_d \big]$. More generally, for a modified kernel $k^{ab}(t_1, t_2) = G(t_1, t_2) \omega^{ab}$, the diagonal restriction gives the classical bracket scaled by $G(t,t)$. If $G$ vanishes on the diagonal, then the diagonal bracket vanishes.
\end{Remark}

The rank drop phenomenon takes a concrete form in the canonical setting. At a diagonal point $(t,t) \in \mathbb{R}^2$, the fibre of $\mathbf{P}_2$ is $\mathbb{R}[q_1,p_1] \otimes \mathbb{R}[q_2,p_2]$ where the subscripts distinguish the two copies of the fibre at $t$. The transposition acts by simultaneous exchange $(q_1, p_1) \leftrightarrow (q_2,p_2)$. In total polynomial degree $1$, the generic fibre has dimension $4$ with basis $(q_1 \otimes 1, p_1 \otimes 1, 1 \otimes q_2, 1 \otimes p_2)$. However, the value at $(t,t)$ of an equivariant smooth section must lie in the $2$-dimensional invariant subspace spanned by $(q_1 \otimes 1 + 1 \otimes q_2, p_1 \otimes 1 + 1 \otimes p_2)$.
Observe that the bracket output $\left\{ q(t_1), p(t_2) \right\}_{1,1} = 1 \otimes 1$ is symmetric and therefore survives on the diagonal. In contrast $\left\{ q(t_1) p(t_1), q(t_2) p(t_2) \right\}_{1,1} = p \otimes q - q \otimes p$ is antisymmetric and therefore cannot occur as the value at a diagonal point of an equivariant smooth section.